\documentclass{article}
\usepackage{latexsym}
\usepackage{amsmath,amsthm}
\usepackage{amssymb}
\usepackage{graphics}
\usepackage{color}
\usepackage{marvosym}
\usepackage{upgreek}
\newcommand{\nabb}{\mbox{$\nabla \mkern-13mu /$\,}}
\newcommand{\pa}{\partial}

\newtheorem{lemma}{Lemma}[section]
\newtheorem{theorem}{Theorem}[section]
\newtheorem{proposition}{Proposition}[section]

\newtheorem{corollary}{Corollary}[section]

\title{A proof of the uniform boundedness of solutions to the wave equation
on slowly rotating Kerr
backgrounds}
\author{Mihalis Dafermos\thanks{University of Cambridge,
Department of Pure Mathematics and Mathematical Statistics,
Wilberforce Road, Cambridge CB3 0WB United Kingdom}
\and
 Igor Rodnianski\thanks{Princeton University,
Department of Mathematics, Fine Hall, Washington Road,
Princeton, NJ 08544 United States}
}
\begin{document}
\maketitle
\begin{abstract}
We consider Kerr 
spacetimes with parameters $a$ and $M$ such that $|a|\ll M$,
Kerr-Newman spacetimes with parameters $|Q|\ll M$, $|a|\ll M$, and more generally,
stationary axisymmetric black hole
exterior spacetimes $(\mathcal{M},g)$ which are sufficiently close to a
Schwarzschild metric with parameter $M>0$,
with appropriate geometric assumptions on the plane spanned by the
Killing fields. We show uniform boundedness
on the exterior for
sufficiently regular solutions to the wave equation $\Box_g\psi=0$, i.e.~we show
that solutions $\psi$ arising from smooth 
initial data $(\uppsi, \uppsi')$
prescribed on an arbitrary Cauchy surface $\Sigma$  satisfy 
$|\psi|\le C\, {\bf Q}_1(\uppsi,\uppsi')$ in the domain of outer communications.
In particular, the bound  holds up to and including the event horizon.
Here, ${\bf Q}_1(\uppsi,\uppsi')$ is a norm on initial data and $C$ depends only
on the parameters of the nearby Schwarzschild metric.
No unphysical restrictions are imposed on the behaviour of $\psi$ near the bifurcation surface
of the event horizon.
The norm ${\bf Q}_1$ is finite if
$\uppsi\in H^2_{\rm loc}(\Sigma)$,
$\uppsi'\in H^1_{\rm loc}(\Sigma)$ and $\uppsi$ is well-behaved at spatial 
infinity, in particular, it is sufficient to assume 
$\nabla\uppsi$ is supported away from spatial infinity. 
The pointwise estimate derives in fact from
the uniform boundedness of a positive definite energy flux.
Note that in view of the very general assumptions,
the separability properties of the wave equation on the Kerr background are not used.
\end{abstract}
\tableofcontents
\section{Introduction}
The Kerr family, discovered in 1963~\cite{kerr}, 
comprises perhaps the most important
family of exact solutions to
the Einstein vacuum equations
\begin{equation}
\label{Eeq}
R_{\mu\nu}=0,
\end{equation}
the governing equations of general relativity.
For parameter values $0\le|a|<M$ (here $M$ denotes the mass and $a$ angular
momentum per unit mass), 
the Kerr solutions represent black hole spacetimes:
i.e.~asymptotically flat spacetimes which possess  a region which cannot
communicate with future null infinity. The celebrated  Schwarzschild family sits as
the one-parameter subfamily of Kerr corresponding to
$a=0$.  Much of current theoretical astrophysics is based on the hypothesis that
isolated systems described by Kerr metrics are ubiquitous
in the observable universe.

Despite the centrality of the Kerr family to the general relativistic world picture, 
the most basic questions about
the behaviour of linear waves on Kerr backgrounds have remained
to this day unanswered. 
This behaviour is in turn intimately connected to the stability properties of the
Kerr metrics themselves as solutions of $(\ref{Eeq})$, 
and thus, with the very physical tenability of the notion of black hole.
In particular, even
the question of the uniform boundedness (pointwise, or in energy)
of solutions $\psi$ to the linear
wave equation
\begin{equation}
\label{waveee}
\Box_g\psi =0
\end{equation}
in the domain of outer communications
has not been previously resolved, except for the Schwarzschild subfamily.

The main theorems of this paper give the resolution of the boundedness problem
for  $(\ref{waveee})$, 
for the case $|a|\ll M$. Solutions to $(\ref{waveee})$
arising from regular initial data remain uniformly bounded in the domain of outer communications.
The bound is quantitative, i.e.~it is
computable in terms of the initial supremum and initial energy-type quantities
on initial data. 

In fact, the results of this paper apply to a much more general setting than
the specific Kerr metric: Boundedness
is proven for solutions of $(\ref{waveee})$ on the exterior region
of any stationary axisymmetric spacetime sufficiently 
close to a Schwarzschild spacetime with mass $M>0$.
Thus, the methods may be of relevance in the ultimate goal of this analysis: 
understanding
the dynamics of the Einstein equations $(\ref{Eeq})$
in a neighborhood of a Kerr metric.

We first give a statement of the main results for the special case of Kerr and the related
Kerr-Newman family (this is a family of solutions to the coupled Einstein-Maxwell system).

\subsection{Statement of the theorem for Kerr and Kerr-Newman}
We refer the reader to~\cite{brandon, he:lssst}.
Let $(\mathcal{M},g)$ denote the Kerr solution with parameters 
\[
0\le |a|<M
\]
or more generally the Kerr-Newman solution with parameters $(a,Q, M)$, with 
\[
0\le \sqrt{a^2+Q^2}<M,
\]
and let $\mathcal{D}$ denote the closure of a domain
of outer communications. (The parameter $Q$ is known as the charge.)
Let $\Sigma$ be a Cauchy hypersurface\footnote{For definiteness, our
``Kerr solution'' or ``Kerr-Newman'' solution 
is the Cauchy development of a complete asymptotically
flat spacelike hypersurface with two asymptotically flat ends. This is a globally
hyperbolic subdomain of the maximal analytic Kerr-Newman described in~\cite{he:lssst}.}
in $(\mathcal{M},g)$ crossing the event horizon to the future of the 
sphere of bifurcation, and such that $\Sigma\cap \mathcal{D}$ 
coincides with a constant-$t$ hypersurface,
for large $r$, where $t$ and $r$ denote here the standard Boyer-Lindquist 
coordinates on ${\rm int}(\mathcal{D})$. Recall that in such coordinates,
the stationary Killing field $T$ is given by $T= \frac{\partial}{\partial t}$.  
The Kerr-Newman solutions are moreover axisymmetric.
The Penrose diagram, say along the axis of symmetry (where the axisymmetric
Killing field vanishes), is depicted below:
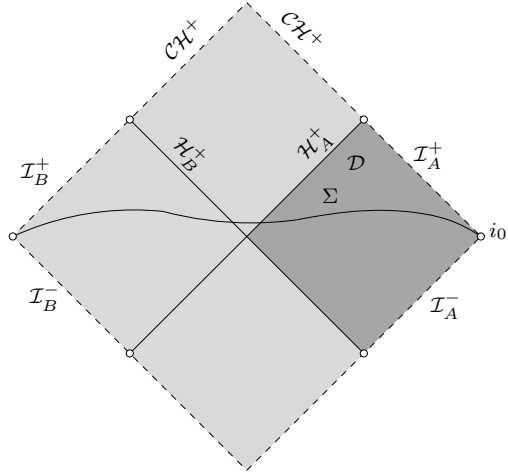
\[
\begin{picture}(0,0)%
\includegraphics{kerr.pstex}%
\end{picture}%
\setlength{\unitlength}{2763sp}%
\begingroup\makeatletter\ifx\SetFigFont\undefined%
\gdef\SetFigFont#1#2#3#4#5{%
  \reset@font\fontsize{#1}{#2pt}%
  \fontfamily{#3}\fontseries{#4}\fontshape{#5}%
  \selectfont}%
\fi\endgroup%
\begin{picture}(4730,4224)(3036,-7123)
\put(6076,-4411){\makebox(0,0)[lb]{\smash{{\SetFigFont{8}{9.6}{\rmdefault}{\mddefault}{\updefault}{\color[rgb]{0,0,0}$\mathcal{D}$}%
}}}}
\put(5701,-4336){\rotatebox{45.0}{\makebox(0,0)[lb]{\smash{{\SetFigFont{8}{9.6}{\rmdefault}{\mddefault}{\updefault}{\color[rgb]{0,0,0}$\mathcal{H}_{A}^+$}%
}}}}}
\put(7351,-5011){\makebox(0,0)[lb]{\smash{{\SetFigFont{8}{9.6}{\rmdefault}{\mddefault}{\updefault}{\color[rgb]{0,0,0}$i_0$}%
}}}}
\put(4501,-3436){\rotatebox{45.0}{\makebox(0,0)[lb]{\smash{{\SetFigFont{8}{9.6}{\rmdefault}{\mddefault}{\updefault}{\color[rgb]{0,0,0}$\mathcal{CH}^+$}%
}}}}}
\put(5476,-3061){\rotatebox{315.0}{\makebox(0,0)[lb]{\smash{{\SetFigFont{8}{9.6}{\rmdefault}{\mddefault}{\updefault}{\color[rgb]{0,0,0}$\mathcal{CH}^+$}%
}}}}}
\put(5851,-4711){\makebox(0,0)[lb]{\smash{{\SetFigFont{8}{9.6}{\rmdefault}{\mddefault}{\updefault}{\color[rgb]{0,0,0}$\Sigma$}%
}}}}
\put(6676,-4336){\makebox(0,0)[lb]{\smash{{\SetFigFont{8}{9.6}{\rmdefault}{\mddefault}{\updefault}{\color[rgb]{0,0,0}$\mathcal{I}_{A}^+$}%
}}}}
\put(3151,-4486){\makebox(0,0)[lb]{\smash{{\SetFigFont{8}{9.6}{\rmdefault}{\mddefault}{\updefault}{\color[rgb]{0,0,0}$\mathcal{I}_{B}^+$}%
}}}}
\put(6826,-5686){\makebox(0,0)[lb]{\smash{{\SetFigFont{8}{9.6}{\rmdefault}{\mddefault}{\updefault}{\color[rgb]{0,0,0}$\mathcal{I}_{A}^-$}%
}}}}
\put(3226,-5611){\makebox(0,0)[lb]{\smash{{\SetFigFont{8}{9.6}{\rmdefault}{\mddefault}{\updefault}{\color[rgb]{0,0,0}$\mathcal{I}_{B}^-$}%
}}}}
\put(4501,-4186){\rotatebox{315.0}{\makebox(0,0)[lb]{\smash{{\SetFigFont{8}{9.6}{\rmdefault}{\mddefault}{\updefault}{\color[rgb]{0,0,0}$\mathcal{H}_{B}^+$}%
}}}}}
\end{picture}%

\]
Note that $\Sigma\cap \mathcal{D}$ is a past Cauchy
hypersurface for $J^+(\Sigma)\cap \mathcal{D}$.\footnote{Here,
$J^+$ denotes causal future, not to be confused with currents $J_\mu$ to 
be defined later.}
 We have that
$J^+(\Sigma)\cap\mathcal{D}$ is foliated by $\Sigma_\tau$ for $\tau\ge 0$,
where $\Sigma_\tau$ is the
future translation of $\Sigma\cap\mathcal{D}$ by the flow generated by
the stationary Killing field $T=\frac{\partial}{\partial t}$ for time $\tau$.
Let $n_\Sigma$ denote the unit future normal of $\Sigma_\tau$. 
Let $n_\mathcal{H}$ denote a translation invariant null generator for $\mathcal{H}^+$,
and give $\mathcal{H}^+\cap \mathcal{D}$ the induced volume from
$g$ and $n_{\mathcal{H}}$.
Let
 $T_{\mu\nu}(\psi)$ denote the standard energy momentum tensor
associated to  a solution $\psi$ of the wave equation $(\ref{waveee})$
\[
T_{\mu\nu}=\partial_\mu\psi\partial_\nu\psi-\frac12
g_{\mu\nu}\nabla^\alpha\psi\nabla_\alpha\psi,
\]
define $J^{n_\Sigma}_\mu(\psi)$ by
\[
J^{n_\Sigma}_\mu(\psi)\doteq
T_{\mu\nu}(\psi) n_\Sigma^\mu
\]
and $J^T_\mu(\psi)$ by
\[
J^T_\mu(\psi)\doteq
T_{\mu\nu}(\psi) T^\mu.
\]
Note that the former current is positive definite when contracted with 
a future-timelike vector field, but is not conserved, whereas the latter current
is conserved, but not positive definite when so contracted.
\begin{theorem}
\label{prwto}
Let $(\mathcal{M},g)$, $\mathcal{D}$, $\Sigma_\tau$ be as above.
There exists a universal positive constant $\epsilon>0$,
and a constant $C$ depending on $M$ and the choice of $\Sigma_0$
such that if 
\begin{equation}
\label{smallness}
0\le a <\epsilon\, M, \qquad 0\le Q<\epsilon \, M,
\end{equation}
 then the following statement holds.
Let    $\psi$ be a solution of $(\ref{waveee})$
on $(\mathcal{M},g)$
such that $\int_{\Sigma_0} J^{n_\Sigma}_\mu(\psi) n^\mu_\Sigma <\infty$. 
Then   
\begin{equation}
\label{prwtoprwto}
\int_{\Sigma_\tau} 
J^{n_\Sigma}_\mu(\psi) n^\mu_\Sigma
\le C \int_{\Sigma_0} J^{n_\Sigma}_\mu(\psi) n^\mu_\Sigma, 
\end{equation}
\begin{equation}
\label{stovorizovta}
\left|\int_{\mathcal{H}^+\cap J^+(\Sigma_0)} 
J^{T}_\mu(\psi) n^\mu_{\mathcal{H}}\right|  \le
C\int_{\Sigma_0}J^{n_\Sigma}_\mu(\psi) n^\mu_\Sigma,
\end{equation}
\begin{equation}
\label{stoapeiro}
\int_{\mathcal{I}^+}J^{T}_\mu(\psi) n^\mu_{\mathcal{I}}  \le C\int_{
\Sigma_0}J^{n_\Sigma}_\mu(\psi) n^\mu_\Sigma.
\end{equation}
Here the integrals are with respect to the induced volume forms. The 
integral on the left hand side of $(\ref{stoapeiro})$ can be defined via a limiting procedure.
\end{theorem}

\begin{theorem}
\label{deutero}
Under the assumptions
of the previous theorem, 
the following holds.
Let $\psi$ be a solution of the wave equation $(\ref{waveee})$
on $(\mathcal{M},g)$ such that 
\[
{\bf Q}_1\doteq \sqrt{\sup_{\Sigma_0}|\psi|^2+\int_{\Sigma_0} \left(J_\mu^{n_0}(\psi)
+J_\mu^{n_0}(n_\Sigma\psi)\right) n^\mu}<\infty.
\]
Then 
\[
|\psi| \le C\,  {\bf Q}_1
\]
in ${\mathcal{D}}\cap J^+(\Sigma_0)$.
\end{theorem}
The hypothesis of Theorem~\ref{prwto} can be re-expressed as
the statement that local energy as measured by
a local observer be finite, i.e.~that $\nabla^{\Sigma_0}\psi|_{\Sigma_0}$, 
$n_\Sigma\psi|_{\Sigma_0}$
be in $L^2_{\rm loc}$, together with the global assumption that
\[
\int_{\Sigma_0}J^T_\mu(\psi) n^\mu_\Sigma<\infty.
\]
The latter in turn is certainly satisfied if $\nabla\psi$ vanishes in a neighborhood of $i_0$.

Similarly, the hypothesis of Theorem~\ref{deutero}  is satisfied for
$\nabla^{\Sigma_0}\psi$, $n_\Sigma\psi|_{\Sigma_0}$ in $H^1_{\rm loc}$,
if $\nabla\psi$ vanishes in a neighborhood of $i_0$.

Finally, note that given an arbitrary Cauchy surface $\tilde{\Sigma}$ for Kerr,
sufficiently well behaved at $i_0$,
it follows that the right hand side of $(\ref{prwtoprwto})$ is bounded by
\[
C(\Sigma_0,\tilde{\Sigma})
\int_{\tilde{\Sigma}\cap (J^-(\Sigma_0)\cup  J^+(\Sigma_0))}
 J_\mu^{n_{\tilde\Sigma}}(\psi) n^\mu_{\tilde\Sigma},
\]
thus the above regularity assumptions could be imposed on an arbitrary Cauchy surface.
There are no unphysical
restrictions on the support of the solution in a neighborhood of $\mathcal{H}^+\cap
\mathcal{H}^-$.

\subsection{Statement for general stationary axisymmetric perturbations of Schwarzschild}
The results of Theorems~\ref{prwto} and~\ref{deutero} remain true
when the Kerr or Kerr-Newman metric
is replaced by an arbitrary stationary axisymmetric 
black hole exterior metric suitably close to Schwarzschild, and with suitable
assumptions on the geometry of the Killing fields. In particular, in addition 
to smallness, it is required
that--as in the Kerr solution--the null generator of the horizon is in the span
of the Killing fields. The precise assumptions are outlined in Section~\ref{ta3n}.

\subsection{Dispersion and the redshift vs.~superradiance}
The elusiveness of the results of Theorems~\ref{prwto} and~\ref{deutero}
stems from the well-known phenomenon of 
\emph{superradiance}. This is  related to    the fact that the Killing field $T$ (with respect
to which the Kerr solution is stationary) is not everywhere-timelike in the domain of outer communications.
In particular, there is a region of spacetime where $T$ is spacelike, the so-called
\emph{ergoregion}. The boundary of this region is called the
\emph{ergosphere}.

The presence of the ergoregion means that the energy current $J^T$ is not positive
definite when integrated over spacelike hypersurfaces. Thus, the conservation of
$J^T$ does not imply \emph{a priori} bounds on an $L^2$-based quantity. 
In particular, the local energy of the solution can be greater than the initial total energy,
even if the energy is initially supported where $J^T$ is positive definite. 
A test-particle version of this
fact, where a particle coming in from infinity splits into one of negative energy entering
the black hole and one of greater positive energy returning to infinity,
is known as the \emph{Penrose process}. The pioneering study by Christodoulou~\cite{chrthe}
of the ``black hole transformations'' obtainable via a Penrose process 
led to a subject known as ``black hole thermodynamics''.

In the physics literature, where discussion of these issues
is inextricably linked to the separability of $(\ref{waveee})$
and decomposition of $\psi$ into modes,   
the problem of the ergoregion appears as a formidable and perhaps
intractable obstacle. It turns out, however,
that there are other physical mechanisms at play which have an important role
but are not necessarily well reflected from the point of view of separability. In particular,  the tendency
of waves to eventually  disperse (true in any asymptotically flat spacetime)
coupled with the powerful red-shift effect at the horizon.
Indeed, these properties, which depend only loosely on the stationarity, 
tend to make solutions not only stay bounded but decay to a constant in time, even if the
local energy increases for a short time.

Unfortunately, the dispersive properties of waves on black hole backgrounds
are severely complicated by the presence of trapped null geodesics. (The presence of these
can easily be inferred by a continuity argument in view of the fact that there exist
both null geodesics crossing the horizon and going to null infinity.)
It is only very recently that the role of trapping has been sufficiently well understood 
in the special case of the Schwarzschild family to allow for the first proofs of 
decay for general solutions of $(\ref{waveee})$ on such backgrounds.
See the results described in Section~\ref{prevsch}.

In the case of Kerr, the techniques introduced
for controlling trapping on Schwarzschild cannot be readily
perturbed. This has to do with the fact that these techniques seem to exploit the special property
that the trapping concentrates asymptotically 
on a set of codimension $1$ in physical space, the so-called photon sphere. 
In contrast, in Kerr the codimensionality of the space of trapped geodesics
can only be properly understood in phase space. This indicates that controlling trapping
requires a far more delicate analysis.

It would appear from the above that the problem of superradiance could
in principle be overcome, but at the expense of a
very delicate analysis of trapping. 
A closer look, however, reveals that the situtation is considerably more
favourable. At a heuristic level, the reason for this is the following remark:
If one could separate out the ``superradiant'' part of the solution from the
``non-superradiant'' part, then one only has to exploit dispersion for
the superradiant part. This latter problem turns out to be much easier than
understanding dispersion for the whole solution.

To decompose the solution, we must first cut off the solution $\psi$  in the ``time''-interval of
interest
to obtain $\psi_{\hbox{\Rightscissors}}$ and then
decompose
into two pieces 
\[
\psi_{\hbox{\Rightscissors}}= \psi_{\flat}+\psi_{\sharp}
\]
where $\psi_{\flat}$ is to be supported in frequency space (real frequencies $\omega$ and 
integer $k$ here
defined with respect to coordinates $t$ and $\phi$) only
in the range $\omega^2 \precsim   \omega_0^2\, k^2$,
whereas $\psi_{\sharp}$ is to be supported in frequency space only in the range
$\omega^2 \succsim \omega_0^2\, k^2$. For spacetimes sufficiently close to Schwarzschild,
for a suitable choice of the parameter $\omega_0$,
one can view
$\psi_{\sharp}$ as essentially non-superradiant, 
and
$\psi_{\flat}$ as the superradiant part. If one can show boundedness for $\psi_{\sharp}$
and dispersion for $\psi_{\flat}$, then one will have proven
the uniform boundedness of the sum $\psi$. For spacetimes sufficiently
close to Schwarzschild, one can choose $\omega_0$ sufficiently small so that
trapping essentially does not occur for $\psi_{\flat}$, and the dispersive
mechanism of Schwarzschild is stable. This relies on the stability of the red-shift effect
for considerations close to the horizon.
In complete contrast to the standard picture, it is the superradiant part of the
solution which
would be the better behaved one.

In practice, the analysis is of course not as simple as what has been portrayed above,
and here again, 
the stabilising effect of the red-shift acting near the horizon plays an important role.
In view of the cutoffs in time,
the equations satisfied by $\psi_\flat$ and $\psi_{\sharp}$ are coupled. 
Moreover, the statement that
$\psi_{\sharp}$ is non-superradiant while $\psi_{\flat}$ is dispersive must also
be understood modulo error terms. 
It turns out that to control these error terms, one of necessity must have at their
disposal an energy quantity which does not degenerate on the horizon, that
is to say, the $L^2$-based quantity for which one shows uniform boundedness
must be the one of Theorem~\ref{prwto}, and not a quantity analogous
to $J^T$ in Schwarzschild. In particular, one must understand the
red-shift mechanism even for the ``non-superradiant'' part $\psi_{\sharp}$, for which
one does not understand dispersion.
Such stable estimates at the horizon (corresponding
to the energy measured by local observers) exploiting the red-shift effect
were first attained for Schwarzschild in our previous~\cite{dr3}.
It is interesting to note, however, that in~\cite{dr3}, understanding of the red-shift mechanism was
always coupled with understanding dispersion, i.e.~controlling the trapping
phenomenon. In particular, one had to appeal to an understanding
of dispersion even to obtain the result of Theorems~\ref{prwto} for Schwarzschild.
In this paper, we show
how understanding the red-shift can be decoupled from understanding dispersion
in the non-superradiant case. In addition, we show that the red-shift effect
allows us to commute the wave equation with a vector field transverse to the horizon,
yielding a new route to higher order estimates and pointwise estimates.
 An extra side-benefit
of our results is thus a new, simpler and more robust proof 
of Theorems~\ref{prwto} and~\ref{deutero} even for the case of Schwarzschild.
See Section~\ref{further1}.

\subsection{Previous results}
We review in detail previous work on this and related problems.
Results of the type of Theorems~\ref{prwto} and~\ref{deutero}
for static perturbations of Minkowski space
pose little difficulty. (Indeed, the analogue
of Theorem~\ref{prwto} is immediate, and Theorem~\ref{deutero}
can be proven with the help of Sobolev inequalities after commuting the equation with
the static Killing field.) Thus, we shall pass directly to the black hole case.

\subsubsection{Schwarzschild}
\label{prevsch}
The analogue of Theorem~\ref{deutero} for Schwarzschild is a celebrated result of 
Kay and Wald~\cite{kw:lss},
building on previous work of Wald~\cite{drimos}
where the theorem had been proven for the restricted
class of data whose support was assumed not to contain the bifurcation sphere 
$\mathcal{H}^+\cap \mathcal{H}^-$. In view of the positive definiteness
of $J^T$ in the domain of outer communications, the only essential difficulty is obtaining
bounds for $\psi$ up to the horizon (where $T$ becomes null), 
as bounds away from the horizon can be obtained
essentially as described immediately above for static
perturbations of Minkowski space.

The arguments of Kay and Wald to prove the analogue of 
Theorem~\ref{deutero} relied on 
the staticity to realize a solution $\psi$ as $\partial_t\tilde{\psi}$ where
$\tilde{\psi}$ is again a solution of~$(\ref{waveee})$ constructed by inverting
an elliptic operator acting on initial data. In addition, a pretty
geometric construction exploiting
the discrete symmetries of maximal Schwarzschild was used
to remove the unphysical restriction on the support near $\mathcal{H}^+\cap \mathcal{H}^-$
necessary for
constructing $\tilde\psi$ in the original~\cite{drimos}. Unfortunately,
neither of these methods is particularly
robust to perturbation. 
The reason the authors had to resort to such techniques was 
that Theorem~\ref{deutero} was proven \emph{without}
proving the analogue of
Theorem~\ref{prwto}, rather, using only the conserved flux $J^T$ whose
control degenerates as $\mathcal{H}^+$ is approached. 
Theorem~\ref{prwto} for Schwarzschild was only proven as part of the decay
results of~\cite{dr3} to be discussed below.

Turning now to the issue of decay, 
the first non-quantitative decay result for $(\ref{waveee})$ on Schwarzschild
is contained in the thesis of Twainy~\cite{ft:td}.
The first quantitative decay results for solutions of~$(\ref{waveee})$
on Schwarzschild (and more generally, Reissner-Nordstr\"om)
were proven in~\cite{dr1}, but were restricted
to spherically symmetric solutions, or alternatively, 
the $0$'th spherical harmonic
$\psi_0$ of a general solution $\psi$.
(In fact, this was a byproduct
of  the main result of~\cite{dr1}, which concerns
decay rates for spherically symmetric solutions to the coupled
Einstein-(Maxwell)-scalar field system.)

Quantitative decay results for the whole solution $\psi$,
both pointwise and in energy, 
were proven in~\cite{dr3}, in 
particular, the uniform decay result
\begin{equation}
\label{sharp}
|\psi|\le C\,{\bf Q} v_+^{-1}
\end{equation}
in the domain of outer communications.
Here $v$ is an Eddington-Finkelstein advanced time coordinate and
${\bf Q}$ is an appropriate quantity computable on initial data, and $v_+$ denotes
say $\max\{v,1\}$.
Inequality $(\ref{sharp})$ is sharp as a uniform decay rate in $v$.
The results of~\cite{dr3} exploit both the red-shift
effect near the horizon and the dispersive properties. The estimates are derived
using a variety of vector field multipliers, in particular, a vector field multiplier
$Y$ such that the flux of $T+Y$ gives the local energy at the horizon. The
energy identity of $Y$ quantifies the red-shift effect.

Weaker decay results were proven independently
by Blue and Sterbenz~\cite{bs:le} for initial data vanishing on 
$\mathcal{H}^+\cap \mathcal{H}^-$, but with control which degenerates
on the horizon. In particular, the estimates of~\cite{bs:le} are unstable to
perturbation near the horizon. The stability of the estimates of~\cite{dr3}
near the horizon will be of critical importance here.

Both~\cite{dr3} and~\cite{bs:le} control trapping effects 
with the help of vector field multipliers which must be carefully chosen for each spherical 
harmonic separately.  These were inspired by a series of papers by Soffer
and collaborators, for instance~\cite{bs1}; see, however~\cite{bs2}.
The first proof of decay for $\psi$
not relying on spherical harmonic decomposition for
the construction of these multipliers
is provided by our more recent~\cite{dr4}.

\subsubsection{Kerr}
Since uniform boundedness is the most basic question which can be asked
about $(\ref{waveee})$ on Kerr, previous results in this setting
are of necessity of a partial nature. 
In particular, essentially all previous work on $(\ref{waveee})$ is restricted
to the projection of $\psi$ to a single azimuthal frequency, or equivalently,
to the case where the data are of the form
\begin{equation}
\label{restricted}
\uppsi=\uppsi_k(r,\theta)e^{-ik\phi},\qquad \uppsi' = \uppsi_k'(r,\theta)e^{-ik\phi}.
\end{equation}
Solutions arising from $(\ref{restricted})$ are then of the form $\psi_k(r,\theta, t)e^{-ik\phi}$.
Let us call such solutions azimuthal modes.
In principle, one could attempt to deduce properties of general $\psi$ by
summing relations deduced for each individual azimuthal mode.
As we shall see, however, due to the non-quantitative nature of the results described
below, in of themselves they unfortunately yield no information about general $\psi$.
Nonetheless, even the study of such $\psi_k$ without regard to uniform control in $k$
turns out to be a non-trivial problem. 
Indeed,
even for such individual azimuthal modes,
the most basic questions had not been previously
answered, in particular, 
the analogue of Theorems~\ref{prwto} or~\ref{deutero}.

This being said, there are interesting partial results
concerning $(\ref{restricted})$ that had been
previously
obtained. In particular, most recently Finster \emph{et~al.}~\cite{fksy} had been able to
show for smooth $\psi_k$ that
for fixed $r>M+\sqrt{M^2-a^2}$ and $\theta$,
\begin{equation}
\label{to0}
\lim_{t\to\infty}\psi_k(t,r,\theta)\to 0,
\end{equation}
under the assumption that the support of $\psi_k$ does not contain $\mathcal{H}^+\cap 
\mathcal{H}^-$. See however~\cite{fksy2}.  
In particular, one can deduce
\begin{equation}
\label{supd}
\sup_{-\infty<t<\infty}\psi_k(t,r,\theta)<\infty,
\end{equation}
for each fixed $r>M+\sqrt{M^2-a^2}$ and $\theta$,
without however a bound on the $\sup$.
The results rest on an 
explicit integral representation of the solution which is
derived using the remarkable (but all too fragile) separability properties
of the Kerr metric. The arguments contain many
pretty applications of contour integral methods
of classical complex analysis.
Since these techniques are essentially algebraic,
no restriction on the size of $|a|$ need be made provided $|a|<M$.
In~\cite{tdee}, under the same assumption on the initial support,
the authors deduce that for each $\delta>0$,
\begin{equation}
\label{nonq}
\sup_{\tau\ge 0}  \int_{\Sigma_\tau\cap \{r\ge M+\sqrt{M^2-a^2}+\delta\}}J^{T}_\mu (
\psi_ke^{-ik\phi})n_\Sigma^\mu  <\infty.
\end{equation}
Thus, the energy of each mode in the region $r\ge M+\sqrt{M^2-a^2}+\delta$
remains finite but again, no quantitative bound in terms of data can be produced. Moreover, 
from the results of~\cite{fksy, tdee}, one cannot deduce 
that the $\sup$ of $(\ref{supd})$ and $(\ref{nonq})$
commute with taking $\lim_{r\to M+\sqrt{M^2-a^2}}$ or
$\lim_{\delta\to 0}$, i.e.~$(\ref{supd})$ is \emph{a priori}
compatible with $\psi_k$ blowing up along the horizon:
\[
\sup_{\mathcal{H}}|\psi_k|=\infty
\]
and $(\ref{nonq})$ is compatible with 
infinite energy 
concentration near the horizon:
\[
\sup_{\tau\ge 0}  \int_{\Sigma_\tau\cap \{M+\sqrt{M^2-a^2}\le r\le 
M+\sqrt{M^2-a^2}+\delta\}}J^{T}_\mu (
\psi_ke^{-ik\phi})n_\Sigma^\mu  = \infty.
\]

As explained before, no statement could be inferred
for the general solution $\psi$ from the above
statements on individual azimuthal
modes, not even a weak statement like $(\ref{supd})$ or
$(\ref{nonq})$. 
This is because the $\lim$ and $\sup$ of
$(\ref{to0})$, $(\ref{supd})$ and $(\ref{nonq})$ do not 
\emph{a priori} commute with summation over
$k$.\footnote{Note
that in the abstract of~\cite{fksy}, $\psi$  
must be understood as the projection $\psi_k$,
to agree with what is proven in the
body of the paper.}
Of course, in view of Theorems~\ref{prwto} and~\ref{deutero}, one can now
infer from $(\ref{to0})$ Corollary~\ref{nqc} of Section~\ref{nqd}.

The somewhat unsatisfying nature of the above previous results deduced with the help of
separability are indicative of how difficult it is to obtain
quantitative statements about solutions of the wave equation $(\ref{waveee})$
 even in the 
algebraically special
case where one has explicit representations of the solution. Perhaps this is for the best,
however. Remarkable though they are, the separability properties of the Kerr metric are 
unstable to perturbation. 
Just as in the case of stability of Minkowski space~\cite{book}, 
understanding the stability
properties of the Einstein equations near the Kerr solution will undoubtedly require
robust methods. We hope that the techniques employed here will have further
applications in this direction.

\subsubsection{Klein-Gordon}
\label{KGsec}
A related problem to the wave equation
is that of the Klein-Gordon equation
\begin{equation}
\label{KGnewlab}
\Box_g \psi = m^2\psi
\end{equation}
with $m>0$. There is a well-developed scattering theory on Schwarzschild
for the class of solutions of $(\ref{KGnewlab})$
with finite energy associated to the Killing $T$. In particular,
an asymptotic completeness statement has been proven in~\cite{bachelot}. 
This analysis in of itself, however, when specialised to $H^1_{\rm loc}$
solutions in the geometric sense,
only gives very weak information about the solution. In particular,
it does not give $L^2$ control of $\psi$ or its
angular derivatives on $\mathcal{H}^+$. 

In the case of Kerr, there are again certain partial results for $(\ref{KGnewlab})$ 
in the direction of scattering for a ``non-superradiant'' subspace of initial data~\cite{haf}.
These interesting results do not, however, address
the characteristic difficulties of superradiance. See also~\cite{horstbeyer}.

\subsubsection{Dirac on Kerr}

Finally, we mention that there has been a series of interesting papers
concerning the Dirac equation on Kerr and Kerr-Newman. See~\cite{hn, fksy0}.
For Dirac, considerations turn out to be much easier as this
equation does not exhibit the phenomenon of superradiance.
We shall not comment more about this here but
refer the reader to the nice article~\cite{hn}.

\subsection{Heuristic work}
We cannot do justice here to the vast work on this subject in the physics literature.
See~\cite{Kokkotas} for a nice survey.

\subsection{Acknowledgements}
The authors thank G.~Holzegel for comments on the manuscript.
The authors also thank the Niels Bohr Institute, Copenhagen, for hospitality while 
this work was being completed.
M.~D.~is support in part by a grant from the European Research Council. 
I.~R.~is support in part by NSF grant DMS-0702270.

\section{Constants and parameters}
\label{constants!}
Constants will play an important role in this paper and it is imperative to set
the conventions early. 
In the next section we shall fix a Schwarzschild metric with parameter $M$.

We shall use the notation $B$ and $b$
for general positive constants which only 
depend on the choice of $M$. An inequality
true with a constant $B$ will be true if $B$ is replaced by a larger constant,
and similary, for $b$ if $b$ is replaced by a smaller positive constant.\footnote{In the case of 
chains of inequalities, e.g.~$f_1\le B f_2\le Bf_3$ this convention is
obviously violated and has to be reinterpreted appropriately.}
We shall use the notation $f_1\sim f_2$ to denote
\[
bf_1 \le  f_2 \le Bf_1.
\]
Since $B$ and $b$ denote general constants, we shall apply without comment the obvious
algebraic rules $B^2=B$, $B^{-1}=b$, $b^2=b$, etc.

We will also require various particular parameters which can be chosen depending
only on $M$:
\[
r^\pm_Y, r_{\hat{Y}}, \,  \delta, \, \delta_{\hat{Y}},  \, q, \, \lambda, 
\, \omega_0, \,\alpha, \, R_0,\, R_1,\,R, \, e, \,\tau_{\rm step}, \,\epsilon_{\rm close}.
\]
The above parameters are not explicitly computed but are 
determined implicitly by various constraints. Before choosing a parameter,
say $\alpha$, 
we shall use notation like $B(\alpha)$, $b(\alpha)$ to denote constants depending only on $M$
and the as of yet unchosen $\alpha$. It is to be understood that again here, the notation $B$ 
indicates that the constant can always be replaced by a bigger one, and $b$ by a smaller one.
We shall also use the notation $R_1(\alpha)$ to indicate that the parameter $R_1$ 
depends on the still unchosen $\alpha$.
Once $\alpha$ is determined, we may replace the expressions $B(\alpha)$, $R(\alpha)$
etc., with $B$, $R$, etc.

\section{The class of spacetimes}
\label{ta3n}
In this section we shall describe the general class
of metrics for which our results will apply. To set
the stage, we must first fix some structures 
associated to a Schwarzschild metric.

\subsection{Schwarzschild}
We refer the reader to our previous~\cite{dr3} for a review of  the
geometry of Schwarzschild.
We must first fix a certain subregion of Schwarzschild with parameter $M>0$,
relevant coordinates, and a choice of axisymmetric Killing field. 
This will provide the underlying
manifold with stratified boundary\footnote{The boundary will be 
the union of two manifolds with boundary intersecting along their common boundary.}
for the class of metrics to be considered later. Let us use the notation
$g_M$ to denote the Schwarzschild metric.

Refer to the diagram below:
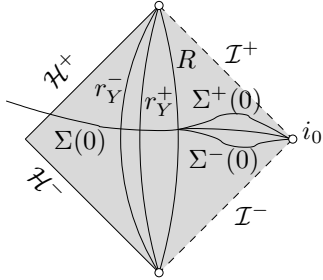
\[
\begin{picture}(0,0)%
\includegraphics{regions.pstex}%
\end{picture}%
\setlength{\unitlength}{3158sp}%
\begingroup\makeatletter\ifx\SetFigFont\undefined%
\gdef\SetFigFont#1#2#3#4#5{%
  \reset@font\fontsize{#1}{#2pt}%
  \fontfamily{#3}\fontseries{#4}\fontshape{#5}%
  \selectfont}%
\fi\endgroup%
\begin{picture}(2752,2178)(5014,-6100)
\put(6357,-4453){\makebox(0,0)[lb]{\smash{{\SetFigFont{10}{12.0}{\rmdefault}{\mddefault}{\updefault}{\color[rgb]{0,0,0}$R$}%
}}}}
\put(5718,-4671){\makebox(0,0)[lb]{\smash{{\SetFigFont{10}{12.0}{\rmdefault}{\mddefault}{\updefault}{\color[rgb]{0,0,0}$r^-_Y$}%
}}}}
\put(6110,-4774){\makebox(0,0)[lb]{\smash{{\SetFigFont{10}{12.0}{\rmdefault}{\mddefault}{\updefault}{\color[rgb]{0,0,0}$r^+_Y$}%
}}}}
\put(6826,-5686){\makebox(0,0)[lb]{\smash{{\SetFigFont{10}{12.0}{\rmdefault}{\mddefault}{\updefault}{\color[rgb]{0,0,0}$\mathcal{I}^-$}%
}}}}
\put(7351,-5011){\makebox(0,0)[lb]{\smash{{\SetFigFont{10}{12.0}{\rmdefault}{\mddefault}{\updefault}{\color[rgb]{0,0,0}$i_0$}%
}}}}
\put(5401,-4636){\rotatebox{45.0}{\makebox(0,0)[lb]{\smash{{\SetFigFont{10}{12.0}{\rmdefault}{\mddefault}{\updefault}{\color[rgb]{0,0,0}$\mathcal{H}^+$}%
}}}}}
\put(5176,-5311){\rotatebox{315.0}{\makebox(0,0)[lb]{\smash{{\SetFigFont{10}{12.0}{\rmdefault}{\mddefault}{\updefault}{\color[rgb]{0,0,0}$\mathcal{H}^-$}%
}}}}}
\put(6751,-4411){\makebox(0,0)[lb]{\smash{{\SetFigFont{10}{12.0}{\rmdefault}{\mddefault}{\updefault}{\color[rgb]{0,0,0}$\mathcal{I}^+$}%
}}}}
\put(5401,-5086){\makebox(0,0)[lb]{\smash{{\SetFigFont{10}{12.0}{\rmdefault}{\mddefault}{\updefault}{\color[rgb]{0,0,0}$\Sigma(0)$}%
}}}}
\put(6484,-4766){\makebox(0,0)[lb]{\smash{{\SetFigFont{10}{12.0}{\rmdefault}{\mddefault}{\updefault}{\color[rgb]{0,0,0}$\Sigma^{+}(0)$}%
}}}}
\put(6457,-5241){\makebox(0,0)[lb]{\smash{{\SetFigFont{10}{12.0}{\rmdefault}{\mddefault}{\updefault}{\color[rgb]{0,0,0}$\Sigma^{-}(0)$}%
}}}}
\end{picture}%

\]
We will denote by $\mathcal{D}$ the closure of
a domain of outer communications in
maximal Schwarzschild. We have $\mathcal{D}\setminus{\rm int}(\mathcal{D})=
\mathcal{H}^+\cup\mathcal{H}^-$ where $\mathcal{H}^+$ denotes the \emph{future
event horizon} and $\mathcal{H}^-$ the \emph{past event horizon}.
The intersection
$\mathcal{H}^+\cap\mathcal{H}^-$ is known as the \emph{bifurcation sphere}.

Recall the static Killing field $T$, timelike on ${\rm int}(\mathcal{D})$ and null
on $\mathcal{H}^+\cup\mathcal{H}^-$.
Flow by integral curves of
$T$ defines a one-parameter
family of diffeomorphisms $\rho_s: \mathcal{D}\to \mathcal{D}$.

Recall now the area-radius function $r$.
On the horizons $\mathcal{H}^+\cup\mathcal{H}^-$ we have $r=2M$. 
We will use the notation $\mu$ for the function defined by
$\mu=2M/r$.

Associated to Schwarzschild will be the constants
$2M<r^-_Y<r^+_Y$ determined in Section~\ref{Ysec}.
We may assume say that
\begin{equation}
\label{specialassump}
r^-_Y \le \frac{9M}4.
\end{equation}

Let $\chi$ be a cutoff function such that $\chi=1$ for $r\le 2M+(r^-_Y-2M)/2$ and $\chi=0$
for $r\ge r^-_Y$. Define the hypersurface 
$\Sigma(0)$ by
\begin{equation}
\label{Notby}
t= -\chi(r)2M\log(r-2M).
\end{equation}
This
can be extended beyond $\mathcal{H}^+$--not by the expression $(\ref{Notby})$, 
however--to a spacelike hypersurface in maximal Schwarzschild.
Let us actually define $\Sigma(0)$ to include its limit points on the horizon
$\mathcal{H}^+$.
Note of course that in view of the support of $\chi$,
it follows that in the region $r\ge r^-_Y$, $\Sigma(0)$ coincides
with the constant $t=0$ hypersurface.

We may define a new coordinate 
\[
t^*\doteq t + \chi(r)2M\log(r-2M)
\]
This coordinate is regular on $\mathcal{H}^+\setminus\mathcal{H}^-$.
We have that
\[
\Sigma(0)=\{ t^*=0\}.
\]
Let us define
\[
\Sigma(\tau)=\{t^*=\tau\}.
\]
Clearly $\Sigma(\tau)=\rho_\tau(\Sigma(0))$.

We have that
\begin{equation}
\label{forspacelike}
B \ge -g_M(\nabla t^*,\nabla t^*) \ge b>0
\end{equation}
for some constants $B$, $b$. Recall here the conventions of Section~\ref{constants!}.

For technical reasons, we shall require two auxilliary sets of spacelike hypersurfaces.
Let $\chi(x)$ be a cutoff function such that $\chi(x)=1$ for $x \le 0$ and $\chi(x)=0$ 
for $x \ge 1$.
Let us define
\[
t^{+}= t^*- \chi(-r+R)(1+r-R)^{1/2}
\]
and 
\[
t^{-}= t^* + \chi(-r+R)(1+r-R)^{1/2}
\]
for an $R$ to be determined later with $R\ge r^-_Y+1$.
Let us define 
\[
\Sigma^{+}(\tau)\doteq\{t^{+} =\tau\},
\qquad
\Sigma^{-}(\tau)\doteq\{t^{-} =\tau\}.
\]
Independently of the choice of $R$, we have that
$\Sigma^{\pm}$ are spacelike, in fact
\begin{equation}
\label{forspacelike2}
B \ge -g_M(\nabla t^{+},\nabla t^{+}) \ge b>0, \qquad B \ge -g_M(\nabla t^{-},\nabla t^{-}) \ge b>0.
\end{equation}

In what follows we shall restrict to 
\[
\mathcal{R}\doteq {\mathcal{D}}\cap J^+_{g_M}(\Sigma^{-}(0)).
\]
The set $\mathcal{R}$ is again a manifold with stratified boundary (as was $\mathcal{D}$), where
the boundary is given by
$\Sigma^{-}(0)\cup (\mathcal{H}^+\cap J^+_{g_M}(\Sigma^{-}(0)))$.

Choosing a coordinate atlas consisting of two charts
$(\xi^A,\xi^B)$, $(\tilde{\xi}^A,\tilde{\xi}^B)$ on
$\mathbb S^2$, then setting
$x^A=r^{-1}\xi^A$, $\tilde{x}^A=r^{-1}\tilde{\xi}^A$, 
it follows that 
\begin{equation}
\label{theatlas}
 (r,t^*, x^A,x^B), \qquad (r,t^*,\tilde{x}^A,\tilde{x}^B)
\end{equation}
form a coordinate atlas
for $\mathcal{R}$. We can ensure moreover that the regions of
the sphere covered by the charts are restricted so that the metric functions
satisfy 
\begin{equation}
\label{nearone}
(g_M)_{ij}\le B ,\qquad g_M^{ij} \le B
\end{equation}
 in these coordinates.
Note that with respect to both these charts,
the vector $\frac{\partial}{\partial t^*}$ is the stationary Killing field $T$.

We will use the above coordinate atlas $(\ref{theatlas})$ in formulating our closeness
assumptions. A third set of coordinates will be useful for us, namely
the coordinates arising from a choice of standard spherical coordinates\footnote{Here,
$\phi$ denotes an azimuthal coordinate.}
$(\theta, \phi)$ on $\mathbb S^2$. 
With respect to 
\begin{equation}
\label{tp}
(r,t^*,\theta, \phi)
\end{equation}
 coordinates, it follows that
$\frac\partial {\partial \phi}$ is a Killing field. 
Let us denote the smooth extension to $\mathcal{D}$ of this Killing field as $\Phi$.
Note that $\Phi$ vanishes precisely at two points on each sphere of
symmetry.
This corresponds to the locus of points where the $(\ref{tp})$ coordinates break down.
Because $(\ref{nearone})$ is not satisfied with respect to these coordinates, they
will not be as useful in formulating the closeness assumptions.

We will say that Schwarzschild is axisymmetric and $\Phi$ is a choice of axisymmetric
Killing field.

Finally, we shall also at times refer to so-called Regge-Wheeler coordinates
\[
(r^*,t,x^A,x^B).
\]
Here $t$ is the standard Schwarzschild time and 
the coordinate $r^*$ is defined by
\[
r^*\doteq r+2M\log(r-2M) -3M -2M\log M.
\]
Note that this coordinate is regular in ${\rm int}(\mathcal{R})$, but sends the
boundary to $r^*=-\infty$.
With respect to these coordinates, 
we note that $\partial_{r^*}$ extends to a smooth vector field 
on all of $\mathcal{R}$ (i.e.~to the event horizon), and in fact, in the limit
$\partial_{r^*}= T$ on $\mathcal{H}^+\cap\mathcal{R}$.

This last coordinate system is not useful for formulating closeness assumptions
in view of the fact that it breaks down on the horizon. We shall only use
Regge-Wheeler coordinates for making calculations with respect to 
the Schwarzschild metric. 

Finally, a word of caution. Since we have several coordinate systems which
will be considered, coordinate vectors like $\partial_{t^*}$ will always be referred
to in conjunction with a specific coordinate system.

\subsection{The general class}
\label{genclass}
We now describe the class of metrics to be allowed.

We consider the manifold with stratified
boundary $\mathcal{R}$ defined above.
We consider the class of all smooth Lorentzian metrics
$g$ such that:
\begin{enumerate}
\item 
For
$\epsilon_{\rm close}>0$ sufficiently small,
\begin{equation}
\label{locsmal}
|g_{ij}- {(g_M)}_{ij}|
\le \epsilon_{\rm close} r^{-2},\qquad  |g^{ij}- {(g_M)}^{ij}|
\le \epsilon_{\rm close} r^{-2}
\end{equation}
\begin{equation}
\label{locsmal2}
|\partial_m g_{ij}-\partial_m({g_M})_{ij}|\le \epsilon_{\rm close} r^{-2}
\end{equation}
with respect to the atlas $(\ref{theatlas})$.\footnote{When specialized to the case of
Kerr-Newman, this clearly will \emph{not} be the Boyer-Lindquist $r$ referred to 
previously. For the relation to Kerr-Newman, see Section~\ref{KNse}.} 
\item
The vector fields $T=\partial_{t^*}$ and $\Phi=\partial_{\phi}$ with respect
to $(r,t^*,\theta, \phi)$ coordinates
are again Killing with respect to $g$.
\item
\label{lastassump}
There is a function $\gamma$ defined on $\mathcal{H}^+$ such that
$T+ \gamma \Phi$ is null on the horizon,
and
\begin{equation}
\label{planeeq}
|\gamma|<\epsilon_{\rm close}.
\end{equation}
\end{enumerate}

In particular, Assumption~\ref{lastassump} above implies that 
$\mathcal{H}^+$ is null with respect to $g$ and its null generator lies in the span
of $T$ and
$\Phi$. 
We may define the  \emph{ergoregion} to be the
region where $T$ itself is \emph{not} timelike.

For sufficiently small $\epsilon_{\rm close}$,
assumptions $(\ref{locsmal})$ and $(\ref{forspacelike})$ 
imply
that $\Sigma(0)$ is spacelike with respect to $g$,
in fact,  with our conventions on constants,
\begin{equation}
\label{xwroeides}
B\ge -g(\nabla t^*,\nabla t^*)\ge b.
\end{equation}
Similarly, we have from $(\ref{forspacelike2})$ that for $\epsilon_{\rm close}$ sufficiently
small,  $\Sigma^{+}(\tau)$ 
and $\Sigma^{-}(\tau)$ are spacelike, in fact
\begin{equation}
\label{xwroeides2}
B\ge -g(\nabla t^{\pm},\nabla t^{\pm})\ge b,
\end{equation}
independently of the choice of $R$.

Note that $\Sigma(\tau)$ is again isometric to $\Sigma(0)$ with respect to $g$,
and similarly $\Sigma^\pm(\tau)$ is isometric to $\Sigma^\pm(0)$.
We will denote by $n_\Sigma$ the future normal of $\Sigma(\tau)$: 
\[
n_\Sigma^\mu\doteq (-g(\nabla t^*,\nabla t^*))^{-1/2}\nabla^\mu t^*.
\] 
This defines
a translation invariant smooth timelike unit vector field on $\mathcal{R}$.
Similarly, we define
\[
n_{\Sigma^{\pm}}^\mu\doteq (-g(\nabla t^{\pm},\nabla t^{\pm}))^{-1/2}\nabla^\mu t^{\pm}.
\]

We will use the notations
\[
\mathcal{R}(\tau',\tau'') \doteq \bigcup_{\tau'\le \bar{\tau}\le \tau''}
\Sigma(\bar{\tau}),
\]
\[
\mathcal{R}^{+}(\tau',\tau'') \doteq \bigcup_{\tau'\le \bar{\tau}\le \tau''}
\Sigma^{+}(\bar{\tau}),
\]
\[
\mathcal{R}^{-}(\tau',\tau'') \doteq \bigcup_{\tau'\le \bar{\tau}\le \tau''}
\Sigma^{-}(\bar{\tau}),
\]
\[
\mathcal{H}(\tau',\tau'')\doteq
\mathcal{H}^+\cap\mathcal{R}(\tau',\tau'').
\]

All integrals without an explicit measure of integration are to be taken
with respect to the volume form in the case of a region of spacetime or
a spacelike hypersurface, and an induced volume form connected
to the choice of a $\rho_t$-invariant tangential vector field $n^\mu_\mathcal{H}$, in the case
of $\mathcal{H}(\tau',\tau'')$.

Note the following property of the volume integral
with respect to the (almost) global $(t, r, \phi, \theta)$ coordinate system:
There exist smooth $\nu(\theta, r)\ge 0$, $\tilde{\nu}(\theta)\ge 0$
such that
for all  continuous $f$:
\begin{eqnarray*}
\int_{\mathcal{R}(\tau',\tau'')}f 
&=& 
\int_{2M}^\infty \int_0^\pi \nu(\theta,r)\left(\int_{\tau'}^{\tau''} 
\left( \int_0^{2\pi}f \, d\phi\right)\, dt^*\right) \,d\theta \, dr , 
\end{eqnarray*}
\begin{eqnarray*}
\int_{\mathcal{H}(\tau',\tau'')}f 
&=& 
 \int_0^\pi \tilde{\nu}(\theta)\left(\int_{\tau'}^{\tau''} 
\left( \int_0^{2\pi}f \, d\phi\right)\, dt^*\right) \,d\theta.  \\
\end{eqnarray*}

Also let us note that 
\[
\int_{\mathcal{R}(\tau',\tau'')} f
=
\int_{\tau'}^{\tau''}\left(\int_{\Sigma(\bar\tau)}\left(-g(\nabla t^*,\nabla t^*)\right)^{-1/2} f \right)
d\bar\tau .
\]
By $(\ref{xwroeides})$, it follows
that if $f_1\sim f_2$ in the sense $0<b f_1 \le f_2 \le B f_1$, it follows
that
\[
\int_{\mathcal{R}(\tau',\tau'')} f_1
\sim 
\int_{\tau'}^{\tau''}\left(\int_{\Sigma(\bar{\tau})}f_2 \right)d\bar\tau.
\]
A similar relation holds with $\mathcal{R}^{\pm}$ and $\Sigma^{\pm}$.

\subsection{The Kerr and Kerr-Newman metrics}
\label{KNse}
\begin{proposition}
Let $Q\ll M$, $a\ll M$. Then the Kerr-Newman metric
with parameters $Q$, $a$ satisfies the assumptions of Section~\ref{genclass}.
\end{proposition}
Let us sketch how one can implicitly define a Kerr-Newman metric on $\mathcal{R}$
in our $(r,t^*,\theta,\phi)$ coordinate system.

For convenience, 
let us do this by defining a new set of coordinates on ${\rm int}(\mathcal{R})$
which are to represent 
Boyer-Lindquist coordinates
$(\hat{r},\hat{t},\hat{\theta},\hat{\phi})$. 
For this
define $\hat{r}$ by
\[
r^2-2Mr = \hat{r}^2-2M\hat{r}+Q^2+a^2
\]
$\hat{t}$ by
\[
\hat{t} =  t^*-h(\hat{r})
\]
where $h$  is defined by $\frac{dh}{d\hat{r}}=
\frac{2M\hat{r} -Q^2}{\hat{r}^2-2M\hat{r}+Q^2+a^2}$, 
and $\hat{\phi}$ by
\[
\hat{\phi} = \phi- P(\hat{r})
\]
where $\frac{dP}{d\hat{r}}=\frac{a}{\hat{r}^2-2M\hat{r}+Q^2+a^2}$,
and $\hat{\theta}$ by
\[
\hat\theta=\theta.
\]
Now consider the metric on ${\rm int}(\mathcal{R})$ defined
in these new coordinates by
\begin{align*}
-\left(1-\frac{2M-\frac{Q^2}{\hat{r}}}
{\hat{r}\left(1+ \frac{a^2\cos^2\hat\theta }{\hat{r}^2}\right) } \right)   \,d\hat{t}^2
+\frac{1+\frac{a^2\cos^2\hat{\theta}}{\hat{r}^2}}{1-\frac{2M}{\hat{r}}+\frac{Q^2}{\hat{r}^2}+
\frac{a^2}{\hat{r}^2}} \, d\hat{r}^2 
+ \hat{r}^2\left(1+\frac{a^2\cos^2\hat\theta}{\hat{r}^2}\right) \, d\hat\theta^2\\
+\hat{r}^2 \left (1+\frac{a^2}{\hat{r}^2}
+\left(\frac{2M}{\hat{r}}-\frac{Q^2}{\hat{r}^2}\right)\frac{a^2\sin^2\hat\theta}{\hat{r}^2\left(1+\frac{a^2\cos^2
\hat\theta}{\hat{r}^2}
\right)}\right)
\sin^2\hat{\theta}\, d\hat{\phi}^2\\
-2\left(2M-\frac{Q^2}{\hat{r}}\right)\frac{a\sin^2\hat{\theta}}{\hat{r}\left(1+\frac{a^2\cos^2{\hat{\theta}}}{\hat{r}^2}\right)}\, 
d\hat{t}  \, d\hat{\phi}.
\end{align*}
Writing the metric in $(r,t^*,\theta,\phi)$ coordinates, and then relating this form in turn
to the coordinates of $(\ref{theatlas})$ one sees immediately that 
\[
r^2 (g_{ij} - (g_{M})_{ij})\to 0, \qquad r^2(g^{ij}-g_{M}^{ij})\to 0
\]
uniformly as $a\to 0$,
and
\[
r^3 (\partial_k g_{ij} - \partial_k(g_{M})_{ij})\to 0
\]
uniformly as $a\to 0$, where $i,j,k$ denote coordinates of $(\ref{theatlas})$.
It follows that given $\epsilon_{\rm close}$, the assumptions $(\ref{locsmal})$
and $(\ref{locsmal2})$ hold.
The remaining assumptions are well-known properties of Kerr
which are manifest from the Boyer-Lindquist form.

\section{The class of solutions $\psi$}
\label{sols}
Let $(\mathcal{R}, g)$, $\Sigma(\tau)$ be as in Section~\ref{genclass}, and let $\uppsi$ be 
an $H^1_{\rm loc}$ function on $\Sigma(0)$, and let $\uppsi'$ be an 
$L^2_{\rm loc}$ function on $\Sigma(0)$. Here the $L^2$ norm is defined naturally
with respect to the induced Riemannian metric on $\Sigma(0)$.
 By standard theory, there exists a unique solution
$\psi$ of
the initial value problem
\begin{equation}
\label{waveq}
\Box_g\psi=0, \qquad \psi|_{\Sigma(0)}=\uppsi, \qquad  n_\Sigma\psi|_{\Sigma(0)}=\uppsi',
\end{equation}
with the property that 
\[
\psi\in C^1(H^1_{\rm loc}(\Sigma(\tau)), \qquad n_\Sigma\psi \in C^0(L^2_{\rm loc}(\Sigma(\tau)).
\]

We will in fact require that 
\begin{equation}
\label{infactassume}
\nabla^{\Sigma}\uppsi \in L^2(\Sigma(0)),\qquad  \uppsi' \in L^2(\Sigma(0)).
\end{equation} 
By density arguments, the main results of this paper would follow if they
were proven under the additional restriction that 
$\uppsi$, $\uppsi'$ are in fact smooth,
and thus, that $\psi$ is smooth. 
Moreover, we can safely assume that $\nabla^{\Sigma} \uppsi$ and $\uppsi'$
are supported away from infinity.
Let us assume this in what follows so as
not to have to comment on regularity issues or the \emph{a priori}
finiteness of certain quantities. 
It follows in particular from this assumption that
\begin{equation}
\label{folpart}
\nabla\psi \in L^2(\Sigma(\tau)), \qquad \nabla\psi \in L^2(\Sigma^{\pm}(\tau)),
\end{equation} 
moreover, that
$\nabla\psi$ is supported away from spatial infinity.

\section{The main theorem}
For a sufficiently regular function $\Psi$, let us define 
\begin{equation}
\label{Tdefwm}
T_{\mu\nu}(\Psi)\doteq 
\partial_\mu \Psi \partial_\nu \Psi -\frac12g_{\mu\nu} g^{\alpha\beta}\partial_\alpha
\Psi \partial_\beta \Psi
\end{equation}
and for $V^\mu$ a vector field, 
\begin{equation}
\label{Jdefwm}
J_\mu^V(\Psi)	\doteq 	T_{\mu\nu}(\Psi)V^\nu.
\end{equation}
In addition, let us define the quantity
\[
{\bf q}(\Psi) \doteq J^{n_\Sigma}_\mu(\Psi) n_\Sigma^\mu.
\]
Note that this is non-negative. Moreover, in the coordinate charts of the atlas $(\ref{theatlas})$,
we have
\begin{equation}
\label{xrnsimo}
{\bf q}(\Psi) \sim \sum_i(\partial_i\Psi)^2.
\end{equation}

By $(\ref{infactassume})$ and $(\ref{folpart})$, 
we have that,
\[
\int_{\Sigma(0)}{\bf q}(\psi) \le
 B\left(\lVert \uppsi' \rVert_{L^2}^2+\lVert\nabla^{\Sigma}\uppsi\rVert_{L^2}^2 \right)
\]
and for all $\tau\ge 0$,
\[
\int_{\Sigma(\tau)}{\bf q}(\psi) <\infty,
\qquad
\int_{\Sigma^{\pm}(\tau)}{\bf q}(\psi) <\infty.
\]
Key to our results will be the uniform boundedness of this quantity.
\begin{theorem}
\label{kevtriko}
There exist positive constants $\epsilon_{\rm close}$, $C$
depending only on $M>0$ such  that the following holds. 
Let $g$, $\Sigma(\tau)$ be as in Section~\ref{genclass}
and let $\uppsi$, $\uppsi'$, $\psi$ be as in Section~\ref{sols} where
$\psi$ satisfies $(\ref{waveee})$.
Then, for $\tau\ge0$,
\begin{equation}
\label{newstatement}
\int_{\Sigma(\tau)} {\bf q}(\psi)\le C\int_{\Sigma(0)} {\bf q}(\psi).
\end{equation}
\end{theorem}
Inequality 
$(\ref{prwtoprwto})$ of Theorem~\ref{prwto} follows from Theorem~\ref{kevtriko}.
(The universality of the constant $\epsilon$ in the statement
of that theorem follows \emph{a posteriori} from a simple scaling argument.)

\section{The auxiliary positive definite quantities ${\bf q}_e$ and ${\bf q}_e^\bigstar$}
We note that given $e>0$, for small enough $\epsilon_{\rm close}\ll e$, 
the vector field $T+e n_{\Sigma}$ is timelike. 
For sufficiently regular $\Psi$, let us define 
\[
{\bf q}_e(\Psi) =  J^{T+en_\Sigma}_\mu(\Psi) n_\Sigma^\mu.
\]
Note that
\[
e b\, {\bf q}(\Psi)\le {\bf q}_e(\Psi) \le B\, {\bf q}(\Psi).
\]
Thus, to prove Theorem~\ref{kevtriko}, it is sufficient to prove $(\ref{newstatement})$
with ${\bf q}_e$ 
replacing ${\bf q}$. The significance of the parameter $e$ will become
clear in the context of the proof.

We shall need also a weaker positive definite quantity defined as follows.
Let $\chi$ denote a cutoff function such that $\chi=1$ for $r\ge r^-_Y$, and $\chi=0$ 
for say $r\le r^-_Y-(r^-_Y-2M)/2$. For a sufficiently regular function $\Psi$,
define
\[
{\bf q}_e^\bigstar(\Psi) = r^{-2} J^{\chi T+en_\Sigma}_\mu(\Psi) n_\Sigma^\mu.
\]
Note that we have
\[
 e b\, r^{-2} {\bf q}(\Psi)\le {\bf q}_e^\bigstar
 (\Psi) \le B\, r^{-2}{\bf q}(\Psi).
\]

Note also that for
$r\ge r^-_Y$,
we have
\[
 {\bf q}_e(\Psi)\sim  {\bf q}(\Psi).
\]
and
\begin{equation}
\label{prebest}
 {\bf q}_e^\bigstar (\Psi)\sim  r^{-2}\,  {\bf q}(\Psi)\sim  r^{-2}\, {\bf q}_e(\Psi).
\end{equation}
For all $r\ge 2M$, we have
\begin{equation}
\label{best}
\, {\bf q}_e(\Psi)\le Be^{-1}
r^2  {\bf q}_e^\bigstar (\Psi),
 \end{equation}
\begin{equation}
\label{best2}
\, {\bf q}(\Psi)\le Be^{-1}
r^2  {\bf q}_e^\bigstar (\Psi).
 \end{equation}

\section{The basic identity for currents}
\label{basicidsec}
For an arbitrary suitably regular function $\Psi$ such that
$\nabla\Psi$ is supported away from spatial infinity,
recall from $(\ref{Tdefwm})$ and $(\ref{Jdefwm})$ the definitions of $T_{\mu\nu}$
and $J_{\mu}$.
Define also
\[
K^V(\Psi) 	\doteq	T_{\mu\nu}(\Psi)\nabla^\mu V^\nu .
\]

We have
\[
\nabla^\mu  J_\mu(\Psi)  = K^V(\Psi) + F\,V^\nu  \Psi_\nu 
\]
where 
\[
F\doteq \Box_g\Psi.
\]
Thus,
setting 
\begin{equation}
\label{basikne}
\mathcal{E}^V(\Psi) =-F\, V^\nu \Psi_\nu,
\end{equation}
 we have
 the identity
\begin{align}
\nonumber
\label{basicid}
&\int_{\mathcal{H}(\tau',\tau'')} J_\mu^V(\Psi) n_\mathcal{H}^\mu
+ \int_{\Sigma(\tau'')} J^V_\mu (\Psi) n_\Sigma^\mu 
+\int_{\mathcal{R(\tau'.\tau'')}} K^V(\Psi) \\
&= \int_{\Sigma(\tau')} 
J^V_\mu (\Psi) n^\mu_\Sigma+
\int_{\mathcal{R}(\tau',\tau'')} \mathcal{E}^V(\Psi).
\end{align}

We will also consider currents modified as follows. Given a function $w$, define
$J_{\mu}^{V, w}$
by 
\begin{equation}
\label{modcur}
J_\mu^{V,w}(\Psi) =J_\mu^V(\Psi) + \frac18 w \partial_\mu(\Psi^2)
-\frac18(\partial_\mu w) \Psi^2,
\end{equation}
\[
K^{V,w}(\Psi) =
K^V(\Psi)  -\frac18\Box_g w(\Psi^2)+\frac14 w\nabla^\alpha\Psi
\nabla_\alpha\Psi,
\]
\begin{equation}
\label{basikne2}
\mathcal{E}^{V,w}(\psi)
=\mathcal{E}^V(\psi) - \frac14 w\Psi F.
\end{equation}
Identity $(\ref{basicid})$ also holds for $J^{V,w}$ as long as appropriate
assumptions are made in a neighborhood of spatial infinity. We will always
apply $J^{V,w}$ to $\Psi$ with $\Psi_0=0$, and thus, by our assumptions
on $\nabla\Psi$, such $\Psi$ will in fact be supported away from spatial infinity.

It will be useful to have a separate notation for currents as defined with respect
to the Schwarzschild metric.
For these we use the notation $(J^V_{g_M})_\mu$, 
$K^V_{g_M}$, $(J^{V,w}_{g_M})_\mu$, etc.

Suppose that $V$ is a vector field such that its components
$V^i$ are bounded in the atlas $(\ref{theatlas})$. It follows
from $(\ref{locsmal})$
that
\begin{equation}
\label{cws0}
\left|(J^{V,w}_{g_M})_\mu (\Psi) n^\mu -J^{V,w}_\mu(\Psi)  n^\mu\right|
\le B\epsilon_{\rm close} r^{-2}
\max_{i} |V_i|
\sum (\partial_i\Psi)^2.
\end{equation}
The above applies in particular if $w=0$, i.e.~for the case $J^V_{g_M}$. (In fact,
the $w$ term disappears from the difference above.)
Note that if the components of $n_\mu-\tilde{n}_\mu$ are less than $B\epsilon_{\rm close} r^{-2}$
we have by the triangle inequality
\begin{eqnarray}
\label{cws00}
\nonumber
\left|(J^{V,w}_{g_M})_\mu (\Psi) n^\mu - J^{V,w}_\mu(\Psi)  \tilde{n}^\mu\right|&\le&
 B\epsilon_{\rm close} r^{-2}
(|w|+ \max_{i} (|V_i|+|\partial_i w|))\\
&&\hbox{}\cdot
\sum (\partial_i\Psi)^2.
\end{eqnarray}

Note also that if $V^j$, $\partial_iV^j$, $w$, $\partial_i w$ and
$\partial_i\partial_j w$ are 
bounded with respect to $(\ref{theatlas})$, where
then
from $(\ref{locsmal})$, $(\ref{locsmal2})$, we obtain
\begin{eqnarray}
\nonumber
\label{cws}
\left|K^{V,w}_{g_M}(\Psi) -K^{V,w}(\Psi)\right| 
&\le& B\epsilon_{\rm close} r^{-2}
\left(
\max_{ij}\max_{p_k=0,1}|\partial_j^{p_j}V^i|+|\partial_i^{p_i}\partial_j^{p_j} w|
\right)\\
&&\hbox{}\cdot
\sum(\partial_i\Psi)^2.
\end{eqnarray}

If $F$ above vanishes, then $J_\mu^{V,w}$ are examples of 
\emph{compatible currents} in the sense of~\cite{book2}.
This is a unifying principle for understanding the structure
behind much of the analysis for Lagrangian equations like $(\ref{waveee})$.

\section{The vector fields and their currents}

\subsection{The vector field $T$}
Since
$T$ is Killing we have
\[
K^T(\Psi)=0.
\]
In $r\ge r^-_Y$, $T$ is timelike and moreover we have
\[
J^T_\mu(\Psi) n^\mu_\Sigma \sim J^{n_\Sigma}_\mu(\Psi) n^\mu_\Sigma
\]
in that region.
In all regions we have
\[
|J_\mu^T(\Psi)n_\Sigma^\mu|\le B\,J_\mu^{n_\Sigma}(\Psi) n^\mu_\Sigma, \qquad
|J_\mu^T(\Psi)n^\mu_\mathcal{H}|\le B\, J_\mu^{n_\Sigma}(\Psi) n^\mu_{\mathcal{H}}.
\]
For $\epsilon_{\rm close}\ll e$ we have
\begin{equation}
\label{inpartic}
|J_\mu^T(\Psi) n_\Sigma^\mu|\le B\, {\bf q}_e(\Psi).
\end{equation}

\subsection{The vector fields $Y$ and $N_e=T+eY$}
\label{Ysec}
Let $(u,v)$ denote Eddington-Finkelstein null coordinates\footnote{See~\cite{dr4,dr5}. Our use of this terminology is
somewhat non-standard. Here $v\doteq (t+r^*)/2$, $u\doteq (t-r^*)/2$.} on ${\rm int}(\mathcal{D})$
and let $r^*$ denote the Regge-Wheeler coordinate. 
In the paragraph that follows, coordinate
derivatives are with respect to say $(u,v,x^A,x^B)$ coordinates, whereas
$y_1'$, $y_2'$ denote $\frac{dy_1}{dr^*}$, etc.

Recall from~\cite{dr3} that for a vector field $Y$ of the form:
\[
Y= y_1(r^*) \frac{1}{1-\mu}\frac{\partial}{\partial u} +y_2(r^*) \frac\partial {\partial v},
\]
we have
\begin{eqnarray*}
K_{g_M}(\Psi)
&=&
\frac{(\partial_u\Psi)^2}{2(1-\mu)^2}
\left(\frac{y_1\mu}r-y_1'\right) + (\partial_v\Psi)^2\frac{y_2'}{2(1-\mu)}\\
&&\hbox{}
+\frac12|\nabb\Psi|^2_{g_M}\left(\frac{y_1'}{1-\mu}-\frac{(y_2(1-\mu))'}{1-\mu}\right)\\
&&\hbox{}-\frac1r\left(\frac{y_1}{1-\mu}-y_2\right)\partial_u\Psi\partial_v\Psi.
\end{eqnarray*}
Let us define $y_1=\xi (r^*) ( 1+(1-\mu)) $, $y_2=\xi(r^*)\delta^{-1}(1-\mu)$ where
$\xi$ is a cutoff function such that $\xi = 1$ for $r\le r^-_Y$, and $\xi =0$ for $r\ge r^+_Y$,
for two parameters $2M<r^-_Y<r^+_Y$, and a small constant $\delta$.
One sees easily that there exist such parameters such that
for $r\le r^-_Y$,
\[
\left(\frac{y_1\mu}r-y_1'\right)  \ge b, \qquad 
\frac{y_2'}{2(1-\mu)} \ge b,
\]
\[
\frac{y'_1}{1-\mu}-\frac{\left(y_2(1-\mu)\right)'}{1-\mu}\ge b
\]
\[
\left|-\frac1r\left(\frac{y_1}{1-\mu}-y_2\right)\partial_u\Psi\partial_v\Psi\right|\
\le \frac12\left(\frac{(\partial_u\Psi)^2}{2(1-\mu)^2}
\left(\frac{y_1\mu}r-y_1'\right) + (\partial_v\Psi)^2\frac{y_2'}{2(1-\mu)}\right).
\]

Let us return now to the coordinate charts of our $(\ref{theatlas})$.
We see from the above that
the vector field $Y$  has the property that in $r\le r^-_Y$, 
\begin{equation}
\label{because}
B \sum_i(\partial_i\Psi)^2 \ge 
K^{Y}_{g_M}(\Psi) \ge b \sum_i(\partial_i\Psi)^2
\end{equation}
where $i, j$ refer to the coordinate charts of $(\ref{theatlas})$
whereas we easily see also that in $r^-_Y\le r \le r^+_Y$ 
\begin{equation}
\label{because2}
|K^Y_{g_M}| \le  B \sum_i(\partial_i\Psi)^2.
\end{equation}
Finally,  for $r\ge r^+_Y$, $Y=0$.

Moreover, we note that $Y$ is a regular vector field, 
in particular, when expressed with respect
to the coordinates of $(\ref{theatlas})$, we have $\max|Y^i|\le B$, $\max|\partial_i Y^j|\le B$.

Because all derivatives appear on the right hand sides of $(\ref{because})$
and $(\ref{because2})$,
these inequalities are stable, i.e.~it
follows from $(\ref{cws})$ 
that for $\epsilon_{\rm close}$ sufficiently small,
\begin{equation}
\label{becbec}
K^{Y}(\Psi) \sim  \sum_i(\partial_i\Psi)^2
\sim J^{n_\Sigma}_\mu(\Psi) n_{\Sigma}^\mu
\end{equation}
in $r\le r^-_Y$, and
\begin{equation}
\label{becbec2}
\left|K^Y\right| \le  B \sum_i(\partial_i\Psi)^2\le
B\,J^{n_\Sigma}_\mu(\Psi) n_{\Sigma}^\mu
\end{equation}
in $r^-_Y\le r \le r^+_Y$, while certainly 
$K^Y=0$ for $r\ge r^+_Y$.

Define 
\[
N_e=T+eY.
\] 
Note that
\[
K^{N_e} = K^T+e\, K^Y= e\, K^Y.
\]

In the region $r^-_Y<r <r^+_Y$, we have by $(\ref{becbec2})$
\[
|K^{N_e}(\Psi) |\le B e \,J_\mu^{n_\Sigma}(\Psi) n_\Sigma^\mu\le eB\, 
{\bf q}_e^\bigstar(\Psi).
\]
Note the factor of $e$.
In the region $r\le r^-_Y$, we certainly have by $(\ref{becbec})$
\begin{equation}
\label{kiautoetsi}
K^{N_e}(\Psi) \ge b\, {\bf q}^\bigstar_e(\Psi).
\end{equation}
For $r\ge r^+_Y$,
we have of course
\[
K^{N_e}=e\, K^Y=0.
\]
In particular, the bound
\begin{equation}
\label{muhim}
-K^{N_e}(\Psi) \le eB \,{\bf q}_e^\bigstar(\Psi)
\end{equation}
holds in all regions.

With the help of $(\ref{becbec})$ and $(\ref{muhim})$, 
we obtain easily that 
\begin{equation}
\label{muh2}
{\bf q}_e (\Psi) \le B \left(K^{N_e}(\Psi) + J_\mu^T(\Psi) n^\mu_\Sigma\right)
\end{equation}
holds everywhere, if $e$ is sufficiently small.

By similar considerations to the above, we see that 
given $e$, by requiring $\epsilon_{\rm close}\ll e$ sufficiently small, we
have that $N_e$ is timelike everwhere up to the boundary, and in fact
\begin{equation}
\label{muhimb}
J^{N_e}_\mu(\Psi) n^\mu_\Sigma \sim {\bf q}_e(\Psi).
\end{equation}
On the other hand, since by Assumption~\ref{lastassump}, $\mathcal{H}^+$ is null,
$J_\mu^{N_e}(\Psi)n^\mu_{\mathcal{H}}$ controls all tangential derivatives. 
More precisely, we have
\begin{equation}
\label{tangcont}
(\partial_t\psi)^2\le  (B+B\epsilon_{\rm close}e^{-1})
J_\mu^{N_e}(\Psi)n^\mu_{\mathcal{H}}\le B\,
J_\mu^{N_e}(\Psi)n^\mu_{\mathcal{H}},
\end{equation}
\begin{equation}
\label{tangcont2}
(\partial_\phi\psi)^2\le Be^{-1}
J_\mu^{N_e}(\Psi)n^\mu_{\mathcal{H}},
\end{equation}
on $\mathcal{H}^+$. For the above we have used the full content of Assumption~\ref{lastassump},
as well as the translation invariance of $n_{\mathcal{H}}$, $n_\Sigma$, $\partial_\phi$,
$\partial_t$ and $N_e$, which allows us to choose uniform constants $B$.

\subsection{The vector fields $X^a$ and $X^b$}
In this section we shall often use Regge-Wheeler coordinates as many
of the computations refer to the Schwarzschild metric $g_{M}$.

In particular, we will consider vector fields of the form 
$V=f(r^*)\partial_{r^*}$.
In what follows $f'$ will denote
$\frac{d f}{d r^*}$. 

In $(t,r^*,x^A,x^B)$ coordinates\footnote{Careful, $t$ not the $t^*$ of our chart! Of course,
$t^*$ coincides with $t$ for $r\ge r^-_Y$.}
we have
\begin{eqnarray}
\label{genform00}
\nonumber
 K^V_{g_M} &=& 
\frac{f'}{1-\mu}(\partial_{r^*}\Psi)^2 + \frac12|\nabb\Psi|^2_{g_M} \left(\frac{2-3\mu}r\right)f\\
&&\hbox{}
-\frac14\left(2f'+4\frac{1-\mu}r f\right) g^{\mu\nu}_{M}\partial_\mu\Psi\partial_\nu\Psi
\end{eqnarray}
where
$|\nabb\Psi|_{g_M}^2$ denotes the induced metric from $g_M$ on the spheres.
We may rewrite the above as
\begin{eqnarray}
\label{genform02}
\nonumber
K^V_{g_M} &=& 
\left(\frac{f'}{2(1-\mu)}-\frac{f}r\right)(\partial_{r^*}\Psi)^2 + |\nabb\Psi|^2_{g_M}\left(-\frac\mu{2r}
f-\frac12 f'\right)\\
&&\hbox{}
+\left(\frac{f'}{2(1-\mu)}+\frac{f} r\right)(\partial_t\Psi)^2.
\end{eqnarray}

Let $\alpha$, $R_1(\alpha)\gg M$ be parameters  to be chosen 
in what follows. Let $R(\alpha)=\exp(4) R_1(\alpha)$.
Given these, we  
define a function $f_a$ such that 
\begin{eqnarray*}
f_a&=& -r^{-4}(r^-_Y)^{4}, \qquad {\rm\ for\ } r\le r^-_Y\\
f_a&=&-1, \qquad {\rm\ for\ }  r^-_Y\le r\le R_1(\alpha),\\
f_a&=& -1+\int_{R_1(\alpha)}^r \frac{d\tilde{r}}{4\tilde{r}} \qquad {\rm\ for\ } R_1(\alpha)\le
  r\le R(\alpha),\\
f_a&=&0 {\rm\ for\ } r\ge R(\alpha).
\end{eqnarray*}
(One can smooth this function, although this is
irrelevant.) 
We call the resulting vector field $X_a$.

We obtain
that in $R_1(\alpha)\ge r>r^-_Y$
\begin{equation}
\label{calcul}
K^{X_a}_{g_M}(\Psi) =
|\nabb\Psi|^2_{g_M}\left(\frac{\mu}{2r}\right) +r^{-1}|\partial_{r^*}\Psi|^2
-r^{-1} |\partial_t\Psi|^2 .
\end{equation}
Since $t=t^*$ for $r\ge r^-_Y$, we can rewrite this
as
\begin{equation}
\label{rrw}
K^{X_a}_{g_M}(\Psi) =
|\nabb\Psi|^2_{g_M}\left(\frac{\mu}{2r}\right) +r^{-1}(1-\mu)^2|\partial_r\Psi|^2
-r^{-1} |\partial_{t^*}\Psi|^2,
\end{equation}
where the coordinate derivatives in the last line can now be understood
with respect to the atlas $(\ref{theatlas})$.
For $\epsilon_{\rm close}$ sufficiently small we obtain from $(\ref{cws})$
\begin{eqnarray}
\label{rrw2}
\nonumber
K^{X_a}(\Psi) &\ge&
|\nabb\Psi|^2\left(\frac{\mu}{2r}\right) +r^{-1}(1-\mu)^2|\partial_r\Psi|^2
\\
&&\hbox{}
-r^{-1} |\partial_{t^*}\Psi|^2 -\epsilon_{\rm close}B\,  {\bf q}_e^\bigstar(\Psi)
\end{eqnarray}
in this region, where we have used $(\ref{prebest})$.

By $(\ref{muhim})$, it follows that in $r^-_Y\le r\le R_1(\alpha)$.
\begin{eqnarray}
\label{legelege}
\nonumber
K^{X_a}+ K^{N_e}(\Psi) &\ge&
 |\nabb \Psi|^2 \left(\frac{\mu}{2r}\right)+r^{-1}(1-\mu)^2|\partial_r\Psi|^2
 \\
 &&\hbox{}
-r^{-1} |\partial_{t^*}\Psi|^2
 - e  B \, {\bf q}_e^\bigstar(\Psi)
\end{eqnarray}
for small enough $\epsilon_{\rm close}\ll e$.

Consider now the region $2M\le r \le r^-_Y$.
We have
\[
f' = 4  r^{-5}(r^-_Y)^{4}(1-\mu),
\]
and thus
\[
\left(\frac{f'}{2(1-\mu)}+\frac fr \right) = 
 (r^-_Y)^{4}r^{-5},
\]
\[
\left(\frac{f'}{2(1-\mu)}-\frac fr \right) = 
3 (r^-_Y)^{4} r^{-5},
\]
\[
\left(-\frac{\mu}{2r}f-\frac12f'\right) =(r^-_Y)^4 \left(\frac{5\mu-4}{2}r^{-5}\right).
\]
We have thus
\[
K_{g_M}^{X_a}(\Psi)
\ge
0
\]
in this region.

Thus, by $(\ref{cws})$, $(\ref{xrnsimo})$ and $(\ref{best2})$
we have
\[
K_g^{X_a}(\Psi) \ge
- \epsilon_{\rm close} e^{-1} B
\,{\bf q}_e^\bigstar (\Psi)
\]
in this region.
It follows now from $(\ref{kiautoetsi})$ that
\begin{equation}
\label{veofavns}
K^{X_a}(\Psi)+ K^{N_e}(\Psi)\ge 
b\,{\bf q}_e^\bigstar(\Psi)
\end{equation}
in this region, 
for small enough $\epsilon_{\rm close}\ll e$.

In view of $(\ref{calcul})$, 
$K^{X_a}$ will ``have\footnote{After integration over appropriate domains and modulo
error terms} a sign'' when applied to $\psi_{\flat}^\tau$ (see Section~\ref{flest})
except for
very large values of $r$, namely $r\ge R_1(\alpha)$.
To control the behaviour there we will need an additional current.
First, let us notice that for the $X_a$ we have selected,
the coefficient of $(\partial_r\Psi)^2$ is always nonnegative.
Finally we notice that for $r\ge R_1(\alpha)$,  the coefficient of $|\nabb\Psi|^2$ satisfies
\begin{equation}
\label{yIldIzcIk}
-\frac{\mu}{2r}f_a-\frac12f'_a \ge -\frac{1}{8r}.
\end{equation}

To choose an additional vector field,
let us
define 
\[
f_b\doteq \chi(r^*)\pi^{-1}\int_0^{r^*}\frac{\alpha\, dx}{x^2+\alpha^2},
\]
where $\chi$ is a smooth cutoff with $\chi=0$ for $r^*\le 0$ and $\chi=1$
for $r^*\ge 1$, and
let $X_b$ be the vector
\[
X_b=f_b\partial_{r^*}.
\]
Finally, define the function
\[
w_b\doteq f'_b+2\frac{1-\mu}r f_b -\frac{2M(1-\mu)f_b}{r^2}
\]
and consider the modified current $J_\mu^{X_b,w_b}$ defined
by $(\ref{modcur})$, as well as the associated $K^{X_b,w_b}$ and 
$\mathcal{E}^{X_b,w_b}$.

Note that for general $f$, we can rewrite
\begin{eqnarray}
\label{rewrite}
\nonumber
K^V_{g_M}&=& 
\left(\frac{f'}{1-\mu}\right)(\partial_{r^*}\Psi)^2 +\frac12 \left(
\frac{2-3\mu}r\right)
f |\nabb\Psi|^2_{g_M}\\
\nonumber
&&\hbox{}- \frac{M(1-\mu)f}{r^2}g^{\mu\nu}_{M}\partial_\mu \Psi \partial_\nu\Psi\\
&&\hbox{}
-\frac18\left(2f'+4\frac{1-\mu}r f-\frac{4M(1-\mu) f}{r^2}\right)(\Box_{g_M}\Psi^2-2\Psi F) 
\end{eqnarray}
from which we see
\begin{eqnarray*}
K^{X_b,w_b}_{g_M}(\Psi)&=&
\left(\frac{f'_b}{1-\mu}-\frac{M f_b}{r^2}\right)(\partial_{r^*}\Psi)^2
+\frac{M f_b}{r^2}(\partial_t\Psi)^2\\
&&\hbox{} +\left(
\frac{2-3\mu}{2r}-\frac{M(1-\mu)}{r^2}\right)
f_b |\nabb\Psi|^2_{g_M}\\
&&\hbox{}
-\frac18\Box_{g_M}\left(2f'_b+4\frac{1-\mu}r f_b-\frac{4M(1-\mu)}{r^2} f_b\right)\Psi^2.
\end{eqnarray*}
Note also the modified error term
\[
\mathcal{E}^{X_b,w_b}(\Psi) = 
\mathcal{E}^{X_b}(\Psi)-\frac14\left(2f'_b+4\frac{1-\mu}rf_b
-\frac{4M(1-\mu)f_b}{r^2}\right)\Psi F .
\]
Finally, let
us define the currents
\[
J^{\bf X}_\mu =  J^{X_a}_\mu  + J^{X_b,w_b}_\mu,
\]
\[
K^{\bf X} =K^{X_a}+K^{X_b,w_b},
\]
\[
\mathcal{E}^{\bf X}=\mathcal{E}^{X_a}+\mathcal{E}^{X_b,w_b}.
\]
By our previous remarks, $(\ref{basicid})$ holds for $J^{\bf X}$.
Also, in view of the definition of $w$, identities $(\ref{cws0})$, 
$(\ref{cws00})$ and $(\ref{cws})$ hold for $J^{\bf X}$, $K^{\bf X}$.

Let us expand
\[
K^{\bf X}_{g_M} =
H_1(\partial_{r^*}\Psi)^2 + 
H_2(\partial_t\Psi)^2
+H_3|\nabb\Psi|^2_{g_M} 
+H_4\Psi^2
\]
where
\[
H_1=\frac{f'_a}{2(1-\mu)}-\frac{f_a}r + \frac{f'_b}{1-\mu}-\frac{M f_b}{r^2},
\]
\[
H_2= \frac{f'_a}{2(1-\mu)}+\frac{f_a}r+\frac{M f_b}{r^2},
\]
\[
H_3=
-\frac{\mu}{2r}f_a-\frac12f'_a+
\left(\frac{2-3\mu}{2r} -\frac{M(1-\mu)}{r^2} \right)f_b,
\]
\[
H_4=
-\frac18\Box_{g_M}\left(2f'_b+4\frac{1-\mu}r f_b-\frac{4M(1-\mu)}{r^2}f_b\right).
\]

Note that  for $r^*\ge 1$, we have
\[
{f'_b}=\frac1{\pi} \frac{\alpha}{(r^*)^2+\alpha^2}.
\]
In particular, for $r\ge R_1(\alpha)$ for sufficiently large $R_1(\alpha)$ 
we
have that
\[
H_1= \frac{f'_a}{2(1-\mu)}-\frac{f_a}r + \frac{f'_b}{1-\mu}-\frac{M f_b}{r^2} \ge \frac{\alpha}{2\pi r^2}
\]
while in $r^-_Y\le r\le R_1(\alpha)$, we have
\[
H_1=
\frac1r + \frac{f'_b}{1-\mu}-\frac{M f_b}{r^2} \ge  \frac1{2r}.
\]

For $H_2$, let us simply remark that for $r\ge R(\alpha)$, we have
\[
H_2 = \frac{M f_b}{r^2} \ge b(\alpha) r^{-2}.
\]

For $H_3$,
we note first that 
we have the following asymptotic formula 
\[
\left(\frac{2-3\mu}{2r} -\frac{M(1-\mu)}{r^2} \right)f_b \sim \frac{1}{r},
\]
i.e.~for
$r\ge R_1(\alpha)$ for sufficiently big $R_1(\alpha)$,
we have
\[
\left(\frac{2-3\mu}{2r} -\frac{M(1-\mu)}{r^2} \right)f_b  \ge  \frac{7}{8r}
\]
and thus by  $(\ref{yIldIzcIk})$
\[
H_3=-\frac{\mu}{2r}f_a-\frac12f'_a+
\left(\frac{2-3\mu}{r} -\frac{M(1-\mu)}{r^2} \right)f_b
\ge \frac{3}{4r}.
\]
To consider the behaviour for $r\le R_1(\alpha)$, let us first note that 
there exists an $R_0$ depending only on $M$--i.e.~independent of $\alpha$ if
we require $\alpha$ to be sufficiently large--such that
for $r>R_0$   we have
\[
\left(
\frac{2-3\mu}{2r}-\frac{M(1-\mu)}{r^2} \right)f_b
\ge 0
\]
and thus, in $R_0\le r\le R_1(\alpha)$ we have
\[
H_3= \frac{\mu}{2r}+\left(
\frac{2-3\mu}{2r}-\frac{M(1-\mu)}{r^2} \right)f_b
\ge \frac{M}{r^2}.
\]
For $r^-_Y \le r\le R_0$ we have
\[
\left| \left(
\frac{2-3\mu}{2r}-\frac{M(1-\mu)}{r^2} \right)f\right|
\le  B\alpha^{-1}
\]
and thus, say
\[
H_3\ge \frac{M}{2r^2}
\]
for $\alpha$ sufficiently large.

Turning to $H_4$, we note first
\begin{eqnarray*}
-\frac18 \Box_{g_M}\left(2f_b'+4\frac{1-\mu}r f_b-\frac{4M(1-\mu)}{r^2} f_b\right)
&=&-\frac14\frac{1}{1-\mu}f_b'''-\frac{1}r f'' +\frac{\mu'}{r(1-\mu)}f_b'\\
&&\hbox{}
-\frac{1}{2(1-\mu)r}\left(\frac{\mu'(1-\mu)}r-\mu''\right)f_b\\
&&\hbox{}
+\frac12\Box_{g_M}\left(\frac{M(1-\mu)}{r^2}f_b\right)\\
&\sim& \frac{7\alpha} {2\pi r^4}
\end{eqnarray*}
for large $r$, i.e.~we
have
\[
H_4=
-\frac18 \Box_{g_M}\left(2f'_b+4\frac{1-\mu}r f_b-\frac{4M(1-\mu)}{r^2} f_b\right)\ge
\frac{7\alpha}{4\pi r^4}
\]
for $r\ge R_1(\alpha)$ for $R_1(\alpha)$ suitably chosen.
On the other hand, one sees easily that $R_0$ before could have been chosen such 
that for all $\alpha$
we have
\[
H_4=-\frac18\Box_{g_M}\left(2f'_b+4\frac{1-\mu}r f_b-\frac{4M(1-\mu)}{r^2}f_b\right) \ge 0
\]
for $r\ge R_0$.
For $r^-_Y\le r\le R_0$, 
we have
\[
\left|-\frac18\Box_{g_M}\left(2f'_b+4\frac{1-\mu}r f_b-\frac{4M(1-\mu)}{r^2}f_b\right)\right| \le B\alpha^{-1}.
\]
We may thus choose $\alpha$ large enough so that in this region
\[
|H_4 |\le  \frac{M}{8r^4}\le \frac1{4r^2} H_3.
\]
{\bf Let $\alpha$ be now chosen. It follows that $R_1=R_1(\alpha)$ and $R=R(\alpha)$
can be chosen. These choices thus can be made to depend only on $M$.}

Let us assume in what follows in this section that $\Psi_0=0$.
We thus have
\begin{equation}
\label{thusthus}
\int_0^{2\pi}\Psi^2\,d\phi \le \int_0^{2\pi} (\partial_\phi\Psi )^2 d\phi.
\end{equation}
It follows that
\[
\int_0^{2\pi}\Psi^2\,d\phi \le \int_0^{2\pi}(\partial_\phi \Psi)^2 \,d\phi\le
r^2 \int_0^{2\pi} |\nabb \Psi|^2_{g_M}\,d\phi.
\]
Thus, in the region $r^-_Y\le r \le R$, we have
\[
\int_0^{2\pi} (H_3|\nabb\Psi|^2_{g_M}+H_4\Psi^2)\, d\phi  \ge \frac12 \int_0^{2\pi}
H_3|\nabb\Psi|^2_{g_M}\,d\phi.
\]

Note that, in the support of $f_b$, we have
\[
(\partial_r^*\Psi)^2\sim (\partial_r\Psi)^2, \qquad
|\nabb\Psi|^2\sim |\nabb \Psi|_{g_M}^2.
\]

We have then by the above bounds and $(\ref{cws})$, $(\ref{xrnsimo})$
and $(\ref{prebest})$
 that for $r\ge  R$, 
\begin{eqnarray}
\label{refto1}
\nonumber
\int_0^{2\pi}( K^{\bf X}+  K^{N_e})(\Psi)\, d\phi 
&\ge&
\nonumber
\int_0^{2\pi}(  K^{\bf X}_{g_M}+  K^{N_e}_{g_M})(\Psi)\, d\phi \\
\nonumber
&&\hbox{}
-
\int_0^{2\pi}\epsilon_{\rm close}B\, {\bf q}_e^\bigstar (\Psi) \,d\phi\\
\nonumber
&=&
\int_0^{2\pi}\left( H_1(\partial_{r^*}\Psi)^2+H_2(\partial_t\Psi)^2
+H_3|\nabb\Psi|^2_{g_M}\right.\\
\nonumber
&&\hbox{}\left. +H_4\Psi^2\right)\,d\phi\\
\nonumber
&&\hbox{}-
\int_0^{2\pi}\epsilon_{\rm close}B\,{\bf q}_e^\bigstar (\Psi) \,d\phi
\\
\nonumber
&\ge&
\int_0^{2\pi} \left(b\, {\bf q}_e^\bigstar (\Psi)
-\epsilon_{\rm close}B\,{\bf q}_e^\bigstar (\Psi)\right)\,d\phi\\
&\ge&
\int_0^{2\pi} b\,{\bf q}_e^\bigstar (\Psi)\, d\phi
\end{eqnarray}
for $\epsilon_{\rm close}$ suitably small,
whereas
for $r^-_Y \le r\le  R$ we may write
\begin{eqnarray}
\label{refto2}
\nonumber
\int_0^{2\pi}( K^{\bf X}+ K^{N_e})(\Psi)\, d\phi &\ge& 
b\int_0^{2\pi}  {\bf q}_e^\bigstar (\Psi) \,d\phi\\
\nonumber
&&\hbox{}
+ \int_0^{2\pi}(b|\nabb\Psi|^2- B(\partial_t\Psi)^2)\,d\phi\\
\nonumber
&&\hbox{}-b\epsilon_{\rm close}\int_0^{2\pi} {\bf q}_e^\bigstar(\Psi)\\
\nonumber
&\ge& 
b\int_0^{2\pi} \,{\bf q}_e^\bigstar (\Psi) \,d\phi\\
&&\hbox{}
+ \int_0^{2\pi}(b|\nabb\Psi|^2- B(\partial_t\Psi)^2)\,d\phi
\end{eqnarray}
where for the second inequality we require that $\epsilon_{\rm close}$ be sufficiently small.
From $(\ref{veofavns})$ and the fact that $f_b$  vanishes identically in $r\le r^-_Y$,
we have
\begin{eqnarray}
\label{refto3}
\nonumber
\int_0^{2\pi}(K^{\bf X}+ K^{N_e})(\Psi)\, d\phi &=&
\int_0^{2\pi}(K^{X_a}+K^{N_e})(\Psi)\, d\phi \\
&\ge& b\int_0^{2\pi} {\bf q}_e^\bigstar(\Psi)
\,d\phi,
\end{eqnarray}
in the region $r\le r^-_Y$.

To give bounds for the boundary terms, note first
that $X_a=- \left(\frac{r^-_Y}{2M}\right)^4T$ on $\mathcal{H}^+$. 
It follows that on the horizon, we have
\[
J^{X_a}_\mu n^\mu_{\mathcal{H}} = - \left(\frac{r^-_Y}{2M}\right)^4
J^T_\mu n^\mu_{\mathcal{H}}.
\]
One sees easily that for $\mathcal{H}^+$ or $\Sigma(\tau)$ where $n^\mu=n^\mu_{\mathcal{H}}$
or $n^\mu=n^\mu_{\Sigma}$, we have
\[
\left|J^{X_a}_\mu n^\mu\right|
 \le B \left|J^T_\mu n^\mu\right|+ B\epsilon_{\rm close}e^{-1}\, 
 J_\mu^{N_e}n^\mu 
\le B\, J^{N_e}_\mu n^\mu
\]
for $\epsilon_{\rm close}$ sufficiently small.
In view of the fact that we also have
\[
|J^T_\mu n^\mu| \le J^T_\mu n^\mu + B\epsilon_{\rm close}e^{-1}\,
 J_\mu^{N_e} n^\mu
\le B\, J^{N_e}_\mu n^\mu ,
\]
it follows that
\[
\left|J^{X_a}_\mu n^\mu  \right|\le 
B\left| J^T_\mu n^\mu\right|  + B\epsilon_{\rm close}e^{-1}\, 
 J^{N_e}_\mu n^\mu 
\le B\, J^{N_e}_\mu n^\mu
\]
on $\mathcal{H}^+$ or $\Sigma(\tau)$.
On the other hand, 
in view of the the assumption $\Psi_0=0$, we
 have similarly
\begin{eqnarray*}
\left|\int_0^{2\pi} J_{\mu}^{X_b}n^\mu \, d\phi \right| &\le&
B \left|\int_0^{2\pi} J_\mu^T n^\mu d\phi\right|  + B \epsilon_{\rm close}e^{-1}
 \int_0^{2\pi} J_\mu^{N_e} n^\mu\, d\phi \\
 &\le & B \int_0^{2\pi} J_\mu^{N_e} n^\mu\, d\phi.
\end{eqnarray*}
It follows from the above inequalities that
\begin{equation}
\label{yeniyIldIz}
\left|\int_0^{2\pi}J_\mu^{\bf X} n ^\mu \, d\phi\right| \le 
B \int_0^{2\pi} J_\mu^{N_e} n^\mu d\phi
\end{equation}
on both $\Sigma(\tau)$ and $\mathcal{H}^+$.

\section{The high-low frequency decomposition}
As explained in the introduction, the arguments  of this paper
hinge on separating the ``superradiant'' part of the solution from the
non-``superradiant'' part, and then exploiting dispersion for the former and 
positive definiteness for the  $J^T$ flux through the event horizon for the latter.
These two parts will be characterized by their support in frequency space.
As we certainly do not know, however, \emph{a priori} that $\psi$ is in $L^2(t^*)$, we will first
need  to cut off $\psi$ in $t^*$. 
This construction, together with propositions which control the errors that arise,
are given in this section.

\subsection{$\psi$ cut off: the definition of $\psi^\tau_{\hbox{\Rightscissors}}$}
Let $\chi(x)$ be a cutoff function such that $\chi(x)=1$ for $x\le 0$ and
$\chi(x)=0$ for $x\ge 1$.
Given $\tau\ge 2$, define 
\begin{eqnarray*}
\psi^\tau_{\hbox{\Rightscissors}}&=&
\chi(t^{+}+1-\tau)\chi(-t^{-}+1)\psi.
\end{eqnarray*}
We may express this as
\[
\psi^\tau_{\hbox{\Rightscissors}}
=\chi^\tau_{\hbox{\Rightscissors}}\psi=({}^{+}\chi^\tau_{\hbox{\Rightscissors}}
+{}^{-}\chi^\tau_{\hbox{\Rightscissors}})\psi,
\]
where ${}^{+}\chi^\tau_{\hbox{\Rightscissors}}$ and
${}^{-}\chi^\tau_{\hbox{\Rightscissors}}$ are smooth
functions on $\mathcal{R}$ 
with
\[
{\rm supp} ({}^{+}\chi^\tau_{\hbox{\Rightscissors}})
\subset \mathcal{R}^{+} (\tau-1,\tau),
\]
\[
{\rm supp} ({}^{-}\chi^\tau_{\hbox{\Rightscissors}}) 
\subset \mathcal{R}^{-} (0,1),
\]
\[
0\le {}^{-}\chi^\tau_{\hbox{\Rightscissors}} \le 1,
\qquad 0\le{}^{+}\chi^\tau_{\hbox{\Rightscissors}}\le 1
\]
and
\[
|\partial_{(i)} {}^{+}\chi^\tau_{\hbox{\Rightscissors}} |\le B_q,
\qquad 
|\partial_{(i)} {}^{-}\chi^\tau_{\hbox{\Rightscissors}} |\le B_q,
\]
with respect to the charts of $(\ref{theatlas})$, for any multi-index $(i)$ of
order $q$.
Moreover,
\begin{equation}
\label{anote}
\partial_\theta {}^{+}\chi^\tau_{\hbox{\Rightscissors}} =0,
\qquad 
\partial_\theta {}^{-}\chi^\tau_{\hbox{\Rightscissors}} =0.
\end{equation}

The reader may wonder why the cutoff region is related to $\Sigma^\pm$,
indeed, why $\Sigma^\pm$ have been introduced in the first place.
Essentially,  
this is necessary to achieve the propositions of Section~\ref{sugkri}--\ref{revisit}.
We would like to express all errors in terms of the positive definite
quantity ${\bf q}_e(\psi)$. This quantity does not contain $\psi$ itself
but only derivatives. Of course, in view of the fact that, as we shall see, 
the spherical average $\psi_0$ does
not give rise to errors, this does not generate problems for the region $r\le R$
for 
$(\psi-\psi_0)^2$ can be controlled by ${\bf q}_e(\psi)$ via a Poincar\'e
inequality. As $r\to \infty$, one needs extra negative powers of $r$. Our cutoff
region diverges from $\mathcal{R}(0,\tau)$ as $r\to\infty$ and
this allows us to ``gain'' powers of $r$ necessary to control $0$'th order terms via
a Poincar\'e inequality in $\mathcal{R}(0,\tau)$. One can
then retrieve estimates all the way to the boundary of the cutoff region using 
 the positive definiteness of $J^T$ for large $r$.

\subsection{Definition of $\Psi_{\flat}$ and $\Psi_{\sharp}$}
\label{defflat}
Let $\zeta$ be a smooth cutoff supported in $[-2,2]$ with the property
that $\zeta=1$ in $[-1,1]$, and let $\omega_0>0$ be a parameter to be determined
later. 

For a smooth function $\Psi(t^*, \cdot)$ of compact support in $t^*$, 
let $\Psi_k$ denote its $k$'th azimuthal mode. Let $\hat{\Psi}$ denote
the Fourier transform of $\Psi$ in $t^*$. Note that $\widehat{\Psi_k}=\hat{\Psi}_k$.

Define
\[
\Psi_{\flat} (t^*,\cdot)
\doteq\sum_{k\ne 0} e^{-ik \theta} \int_{-\infty}^{\infty} \zeta((\omega_0k)^{-1}\omega)\, \hat\Psi_{k}(\omega,\cdot) \, e^{i\omega t^*}d\omega,
\]
\[
\Psi_{\sharp}(t^*,\cdot) \doteq
\Psi_0+
 \sum_{k\ne 0} e^{-ik \theta} 
 \int_{-\infty}^{\infty} \left(1-\zeta((\omega_0k)^{-1}\omega)\right)\, \hat\Psi_{k}
 (\omega,\cdot)\, e^{i\omega t^*}d\omega.
\]

Note of course
\begin{equation}
\label{fusika}
\Psi_{\flat}+\Psi_{\sharp}= \Psi.
\end{equation}
Note in addition that
\begin{equation}
\label{torecall}
(\Psi_\flat)_0 =0
\end{equation}
whereas
\begin{equation}
\label{whereas}
(\Psi_{\sharp})_0 = \Psi_0.
\end{equation}

In the application to $\Psi=\psi^\tau_{\hbox{\Rightscissors}}$, we shall write
simply $\psi^\tau_{\sharp}$ and $\psi^\tau_{\flat}$. Note
finally, that in view of $(\ref{anote})$,
$(\psi_k)^\tau_{\hbox{\Rightscissors}}=(\psi^\tau_{\hbox{\Rightscissors}})_k$.

Note that for $k\ne 0$,
$$
(\Psi_{\flat})_k(t^*)=\int_{-\infty}^{\infty} \zeta((\omega_0k)^{-1}\omega)\, \hat\Psi_{k}(\omega) 
\, e^{i\omega t^*}d\omega=\int_{-\infty}^{\infty} P^<_{k}(t^*-s^*) \Psi_k(s^*)\, ds^*,
$$
where 
$$
P^{<}_{k}(t^*)= \omega_0 k\int_{-\infty}^{\infty} \zeta(\omega) e^{-i\omega (\omega_0 k t^*)}\,d\omega.
$$
The kernel $P^<_{k}(t^*)$
is a rescaled copy of a Schwarz function of $t^*$. 
As a consequence, for any $m, q\ge 0$,
\begin{equation}
\label{convest}
|\pa_{t^*}^m P^<_{k} (t^*)|\le B_{mq} (\omega_0 |k|)^{m+1} \left (1+|\omega_0 k t^*|\right)^{-q}.
\end{equation}
On the other hand,
let $\tilde \zeta$ be a smooth cut-off function supported in $(-3,3)$ such that $\tilde \zeta=1$ on 
$[-2,2]$. Then, since $\tilde \zeta\, \zeta =\zeta$, we have the reproducing formula
$$
(\Psi_{\flat})_k(t^*)= \int_{-\infty}^\infty \tilde \zeta((\omega_0k)^{-1}\omega) 
(\hat \Psi_{\flat})_k(\omega)\,  e^{i \omega t^*}\, d\omega=
\int_{-\infty}^\infty  \tilde P^{<}_{k} (t^*-s^*) (\Psi_{\flat})_k(s^*)\, ds^*,
$$
where the kernel $\tilde P^<_{k}$ also satisfies $(\ref{convest})$.

Finally, let $\xi(\omega)$ be a function smooth away from $\omega=0$ and with the property that 
$\xi(\omega)=\omega^{-1}$ for $|\omega|\le 1/2$ and $\xi(\omega)=1$ for $|\omega|\ge 1$. 
In particular,
the function $\tilde\xi(\omega)=\omega \xi(\omega)$ is smooth and 
$\tilde \xi(\omega)=1$ for $|\omega|\le 1/2$ and $\tilde \xi(\omega)=\omega$ for $|\omega|\ge 1$.
Since $\xi (1-\zeta)=1-\zeta$, we can write for $k\ne 0$,
\[
(\Psi_{\sharp})_k(t^*) = \int_{-\infty}^\infty Q^>_{k}(t^*-s^*) (\Psi_{\sharp})_k(s^*)\, ds^*, 
\]
where
\[
Q^>_{k}(t^*)=\omega_0k \int_{-\infty}^\infty 
\xi(\omega)\,e^{i\omega (\omega_0 k t^*)}\, d\omega.
\]
Furthermore,
\[
\pa_{t^*} (\Psi_{\sharp})_k(t^*) = \omega_0 k  \int_{-\infty}^\infty  \tilde Q^>_{k} (t^*-s^*) 
(\Psi_{\sharp})_k(s^*)\, ds^*,
\]
where
\[
\tilde Q^>_{k}(t^*)=\omega_0 k \int_{-\infty}^\infty 
\tilde \xi(\omega)\, e^{i\omega (\omega_0 k t^*)}\, d\omega.
\]
and 
\[
(\Psi_{\sharp})_k(t^*) = (\omega_0k)^{-1}\int_{-\infty}^\infty 
R^>_{k} (t^*-s^*) \pa_{s^*} (\Psi_{\sharp})_k(s^*)\, ds^*, 
\]
where
\[
R^>_{k}(t^*)=
\omega_0 k \int_{-\infty}^\infty 
\left (\tilde \xi(\omega)\right)^{-1}\,e^{i\omega (\omega_0kt^*)} \,d\omega.
\]
The function $a(\omega)=(\tilde \xi(\omega))^{-1}$ is equal to one on $(-1/2,1/2)$ and 
$\omega^{-1}$ for $|\omega|\ge 1$. The kernel $R^>_{k}(t^*)$ satisfies
$$
|R^>_{k}(t^*)|\le B_q (\omega_0|k|)^{1-q} (t^*)^{-q}
$$
for any $q>0$.
In addition, we have a uniform bound (coming from $1/\omega$ decay)
$$
|R^>_{k}(t^*)|\le B \omega_0 |k|\,| \log (\omega_0 |k| t^*)|.
$$
Combining we obtain
$$
|R^>_{k}(t^*)| \le B_q \omega_0 |k|\,| \log (\omega_0 |k| t^*)|\,
(1+|\omega_0 kt^*|)^{-q}.
$$

\subsection{Comparing $\partial_{t^*}\Psi$ and $\partial_\phi\Psi$}
The decomposition of $\Psi$ into $\Psi_{\flat}$ and $\Psi_{\sharp}$
is motivated
by the desire to compare various $L^2$-type norms 
of the $\partial_\phi$ and $\partial_{t^*}$ derivatives. 
Since this is required at a localised level, however,
error terms arise. The precise relations one can make are 
recorded in this section. The estimates
of this section employ standard techniques of elementary Fourier analysis.
We must be careful, however, to express all ``error terms'' in a form which 
can be related to our bootstrap assumptions which will be introduced later
on.

\subsubsection{Comparisons for $\Psi_{\flat}$}
First a lemma.
\begin{lemma}
\label{flatlem}
Let $\tau''\ge\tau'$ and let $\Psi$ be smooth and of
compact support in $t^*$. Then
\begin{eqnarray*}
\int_{\mathcal{R}(\tau',\tau'')\cap \{r^-_Y\le r\le R\}}(\partial_{t^*}\Psi_{\flat})^2
&\le& B {\omega_0}^2 \int_{\mathcal{R}(\tau',\tau'')\cap \{r^-_Y\le r\le R\}}
  (\partial_\phi \Psi_{\flat})^2
\\
&&\hbox{}+ B{\omega_0} \sup_{-\infty\le \bar\tau\le  \infty}\int_{\bar{\tau}}^{\bar\tau+1}
\left(\int_{\Sigma(\tilde\tau)\cap
\{r^{-}_Y\le r\le R\} } 
(\partial_\phi\Psi_{\flat})^2\right) d\tilde\tau.
\end{eqnarray*}
\end{lemma}
\begin{proof}
Recall $(\ref{torecall})$.
Note first that by the relations of Section~\ref{defflat}, it follows
that for any $q>0$, we have
\[
\left|\partial_{t^*} (\Psi_\flat)_k (t^*,\cdot)\right|
\le B_q (\omega_0 k)^2\int_{-\infty}^{\infty} \left(1+\left|\omega_0 k(t^*-s^*)\right|\right)^{-q}
\left|(\Psi_\flat)_k\right|(s^*,\cdot) ds^*.
\]
We have thus 
\begin{eqnarray*}
\left|\partial_{t^*} (\Psi_\flat)_k (t^*,\cdot)\right|&\le
&B_q(\omega_0k)^2\sum_{\ell=-\infty}^\infty
\int_{t^*+\frac{\ell}{\omega_0 |k|}}^{t^*+\frac{\ell+1}{\omega_0 |k|}}
 \left(1+|\ell| \right)^{-q}
\left|(\Psi_\flat)_k\right|(s^*,\cdot) ds^*\\
&\le&B_q(\omega_0k)^2\sum_{\ell=-\infty}^\infty  \left(1+|\ell| \right)^{-q}
\int_{t^*+\frac{\ell}{\omega_0 |k|}}^{t^*+\frac{\ell+1}{\omega_0 |k|}}
\left|(\Psi_\flat)_k\right|(s^*,\cdot) ds^*\\
&\le&B_q(\omega_0|k|)^{3/2} \sum_{\ell=-\infty}^\infty  \left(1+|\ell| \right)^{-q}
\left(\int_{t^*+\frac{\ell}{\omega_0 |k|}}^{t^*+\frac{\ell+1}{\omega_0 |k|}}
(\Psi_\flat)_k^2(s^*,\cdot) ds^*\right)^{1/2}.
\end{eqnarray*}

It follows that, for $q>1$,
\begin{eqnarray*}
&&\int_{\tau'}^{\tau''} (\partial_{t^*} (\Psi_\flat)_k)^2 (t^*,\cdot)\,dt^*\\
&\le& B_q (\omega_0|k|)^3 \int_{\tau'}^{\tau''}\left(  \sum_{\ell=-\infty}^\infty  \left(1+|\ell| \right)^{-q}
\left(\int_{t^*+\frac{\ell}{\omega_0 |k|}}^{t^*+\frac{\ell+1}{\omega_0 |k|}}
(\Psi_\flat)_k^2(s^*,\cdot) \,ds^*\right)^{1/2}\right)^2
    \,dt^*\\
   &\le&
   B_q (\omega_0|k|)^3 \int_{\tau'}^{\tau''}\left( \sum_{\ell=-\infty}^\infty  \left(1+|\ell| \right)^{-q}\right)
    \sum_{\ell=-\infty}^\infty \left(1+|\ell| \right)^{-q}
\int_{t^*+\frac{\ell}{\omega_0 |k|}}^{t^*+\frac{\ell+1}{\omega_0 |k|}}
(\Psi_\flat)_k^2(s^*,\cdot)\, ds^*
    \,dt^*\\
&\le&     B_q (\omega_0|k|)^3 \int_{\tau'}^{\tau''} 
    \sum_{\ell=-\infty}^\infty \left(1+|\ell| \right)^{-q}
\int_{t^*+\frac{\ell}{\omega_0 |k|}}^{t^*+\frac{\ell+1}{\omega_0 |k|}}
(\Psi_\flat)_k^2(s^*,\cdot)\, ds^*
    \,dt^*\\
 &=&B_q (\omega_0|k|)^3
    \sum_{\ell=-\infty}^\infty \left(1+|\ell| \right)^{-q} \int_{\tau'}^{\tau''} 
\int_{t^*+\frac{\ell}{\omega_0 |k|}}^{t^*+\frac{\ell+1}{\omega_0 |k|}}
(\Psi_\flat)_k^2(s^*,\cdot)\, ds^*
    \,dt^*\\
     &=&B_q (\omega_0|k|)^3
    \sum_{\ell=-\infty}^\infty \left(1+|\ell| \right)^{-q}
 \int_{\frac{\ell}{\omega_0 |k|}}^{\frac{\ell+1}{\omega_0 |k|}} 
\int_{\tau'}^{\tau''}
(\Psi_\flat)_k^2(s^*+t^*,\cdot)\, dt^*
    \,ds^*\\
         &=&B_q (\omega_0k)^2
    \sum_{\ell=-\infty}^\infty \left(1+|\ell| \right)^{-q}
 \int_{\tau'+\frac{\ell}{\omega_0|k|}}^{\tau''+\frac{\ell+1}{\omega_0 |k|}} 
(\Psi_\flat)_k^2(t^*,\cdot)
    \,dt^*.
\end{eqnarray*}
Thus,
\begin{eqnarray}
\nonumber
\label{above}
&&\int_{\tau'}^{\tau''} (\partial_{t^*} (\Psi_\flat)_k)^2 (t^*,\cdot)\,dt^*\\
\nonumber
&\le& B_q (\omega_0k)^2
    \sum_{\ell=-\infty}^\infty \left(1+|\ell| \right)^{-q}
 \int_{\tau'+\frac{\ell}{\omega_0|k|}}^{\tau''+\frac{\ell+1}{\omega_0 |k|}} 
\chi_{\tau',\tau''}(s^*)(\Psi_\flat)_k^2(s^*,\cdot)
    \,ds^* \\
    \nonumber
  &&+B_q (\omega_0k)^2
    \sum_{\ell=-\infty}^\infty \left(1+|\ell| \right)^{-q}
 \int_{\tau'+\frac{\ell}{\omega_0|k|}}^{\tau''+\frac{\ell+1}{\omega_0 |k|}} 
(1-\chi_{\tau',\tau''}(s^*))(\Psi_\flat)_k^2(s^*,\cdot)
    \,ds^* \\
    &\doteq& T_{1,k}+T_{2,k}
\end{eqnarray}
where $\chi_{\tau',\tau''}(s^*)=1$ if $s^*\in[\tau',\tau'']$ and $0$ otherwise.

To prove the lemma, in view of the comments in Section~\ref{genclass}
on the volume form and Plancherel, it would suffice to show
that
\begin{align}
\label{would1}
\nonumber
&\sum_{|k|\ge1} \int_{r^-_Y}^R \int_0^\pi \nu(\theta, r)\int_0^{2\pi}
T_{1,k} \,
d\phi \, d\theta\, dr \\
&\le 
 B {\omega_0}^2 \int_{\tau'}^{\tau''} \int_{r^-_Y}^R \int_0^\pi \nu(\theta, r)\int_0^{2\pi}
  (\partial_\phi \Psi_{\flat})^2 \,d\phi \, d\theta\, dr\, dt^*,
\end{align}
\begin{align}
\label{would2}
\nonumber
&\sum_{|k|\ge 1} \int_{r^-_Y}^R \int_0^\pi \nu(\theta, r)\int_0^{2\pi}
T_{2,k} \,
d\phi \, d\theta\, dr \\
&\le
 B {\omega_0}\sup_{-\infty\le \bar\tau\le \infty}\int_{\bar\tau}^{\bar\tau+1}
  \int_{r^-_Y}^R \int_0^\pi \nu(\theta, r)\int_0^{2\pi}
  (\partial_\phi \Psi_{\flat})^2 \,d\phi \, d\theta\, dr d\tilde\tau.
\end{align}

The first term on the right hand side
of $(\ref{above})$
is bounded by
\begin{eqnarray*}
T_{1,k}& \le& B_q (\omega_0k)^2
    \sum_{\ell=-\infty}^\infty \left(1+|\ell| \right)^{-q}
 \int_{\tau'}^{\tau''} 
(\Psi_\flat)_k^2(s^*,\cdot)
    \,ds^*\\
    &\le&
     B_q (\omega_0k)^2
     \int_{\tau'}^{\tau''} 
(\Psi_\flat)_k^2(s^*,\cdot)
    \,ds^*.
\end{eqnarray*}
Thus, it follows that
\begin{eqnarray*}
\sum_{|k|\ge 1} \int_0^{2\pi} T_{1,k} d\phi
    &\le&
     \sum_{|k|\ge1} B_q \omega_0^2k^2
    \int_0^{2\pi} \int_{\tau'}^{\tau''} 
(\Psi_\flat)_k^2(s^*,\cdot)
    \,ds^*d\phi\\
     &\le&
     \sum_{|k|\ge1} B_q \omega_0^2
    \int_0^{2\pi} \int_{\tau'}^{\tau''} 
(\partial_\phi\Psi_\flat)^2(s^*,\cdot)
    \,ds^*d\phi.
\end{eqnarray*}
We have established $(\ref{would1})$.

The second term on the right hand side of $(\ref{above})$ is bounded
by
\begin{eqnarray}
T_{2,k}&\le&\nonumber
 B_q(\omega_0k)^2
    \sum_{\ell=-\infty}^{-1} \left(1+|\ell| \right)^{-q}
 \int_{\tau'+\frac{\ell}{\omega_0|k|}}^{\tau'} 
(\Psi_\flat)_k^2(s^*,\cdot)
    \,ds^*\\
    \nonumber
    &&\hbox{}+ B_q (\omega_0k)^2
    \sum_{\ell=0}^\infty \left(1+|\ell| \right)^{-q}
 \int_{\tau''}^{\tau''+\frac{\ell+1}{\omega_0|k|}}  
(\Psi_\flat)_k^2(s^*,\cdot)
    \,ds^*\\
    &\doteq& T_{21,k}
  +T_{22,k}.
\end{eqnarray}

We have
 \begin{eqnarray}
 \label{yepyeni}
 \nonumber
&&\sum_k \int_{r^-_Y}^R \int_0^\pi \nu(\theta, r)\int_0^{2\pi} T_{21,k} \, d\phi \, d\theta\,  dr\\
\nonumber
&\le& B_q\omega_0^2\sum_{|k|\ge 1}
 \sum_{\ell=-\infty}^{-1} \left(1+|\ell| \right)^{-q}
 \int_{r^-_Y}^R \int_0^\pi \nu(\theta, r)\int _0^{2\pi} \int_{\tau'+\frac{\ell}{\omega_0|k|}}^{\tau'} 
k^2(\Psi_\flat)^2_k(s^*,\cdot)
    \,ds^*\,d\phi\,d\theta\, dr\\
    \nonumber
    &\le& B_q\omega_0^2\sum_{|k|\ge 1}
 \sum_{\ell=-\infty}^{-1} \left(1+|\ell| \right)^{-q}
 \int_{r^-_Y}^R \int_0^\pi \nu(\theta, r) \int _0^{2\pi} \int_{\tau'+\frac{\ell}{\omega_0}}^{\tau'} 
k^2(\Psi_\flat)^2_k(s^*,\cdot)
    \,ds^*\,d\phi\, d\theta\, dr\\
    \nonumber
        &\le& B_q\omega_0^2
 \sum_{\ell=-\infty}^{-1} \left(1+|\ell| \right)^{-q}
\sum_{|k|\ge 1} \int_{r^-_Y}^R \int_0^\pi \nu(\theta, r) \int _0^{2\pi} \int_{\tau'+\frac{\ell}{\omega_0}}^{\tau'} 
k^2(\Psi_\flat)^2_k(s^*,\cdot)
    \,ds^*\,d\phi\, d\theta\, dr\\
    \nonumber
     &\le& B_q\omega_0^2
 \sum_{\ell=-\infty}^{-1} \left(1+|\ell| \right)^{-q}
 \int_{r^-_Y}^R \int_0^\pi \nu(\theta, r)\int _0^{2\pi} \int_{\tau'+\frac{\ell}{\omega_0}}^{\tau'} 
(\partial_\phi\Psi_\flat)^2(s^*,\cdot)
    \,ds^*\,d\phi\, d\theta\, dr\\
    \nonumber
          &\le& B_q\omega_0^2
 \sum_{\ell=-\infty}^{-1} \left(1+|\ell| \right)^{-q}|\ell|\omega_0^{-1}
\sup_{-\infty\le\bar{\tau}\le \infty}
 \int_{r^-_Y}^R \int_0^\pi \nu(\theta, r) \int _0^{2\pi}  \int_{\bar\tau}^{\bar\tau+1}
(\partial_\phi\Psi_\flat)^2(s^*,\cdot)
    \,ds^*\,d\phi\, d\theta\, dr\\
          &\le& B\omega_0
\sup_{-\infty\le\bar{\tau}\le \infty} \int_{r^-_Y}^R \int_0^\pi \nu(\theta, r)
\int _0^{2\pi} \int_{\bar\tau}^{\bar\tau+1}
(\partial_\phi\Psi_\flat)^2_k(s^*,\cdot)
    \,ds^*\,d\phi\, d\theta\, dr,
 \end{eqnarray}
 for $q$ chosen sufficiently large.

As for $T_{22,k}$, we have
\begin{eqnarray*}
&&\sum_k \int_{r^-_Y}^R \int_0^\pi \nu(\theta, r)\int_0^{2\pi} T_{22,k} d\phi \, d\theta\,  dr\\
&\le& B_q\omega_0^2\sum_{|k|\ge 1} 
 \sum_{\ell=0}^{\infty} \left(1+\ell \right)^{-q}
 \int_{r^-_Y}^R \int_0^\pi \nu(\theta, r)\int_0^{2\pi} \int_{\tau''}^{\tau''+\frac{\ell+1}{\omega_0|k|}} 
k^2(\Psi_\flat)_k^2(s^*,\cdot)
    \,ds^*\,d\phi\, d\theta\, dr\\
&\le& B_q\omega_0^2\sum_{|k|\ge 1} 
 \sum_{\ell=0}^{\infty} \left(1+\ell \right)^{-q} \int_{r^-_Y}^R \int_0^\pi \nu(\theta, r)
\int_0^{2\pi} \int_{\tau''}^{\tau''+\frac{\ell+1}{\omega_0}} 
k^2(\Psi_\flat)_k^2(s^*,\cdot)
    \,ds^*\,d\phi\, d\theta\, dr\\
&\le& B_q\omega_0^2
 \sum_{\ell=0}^{\infty} \left(1+\ell \right)^{-q} \int_{r^-_Y}^R \int_0^\pi \nu(\theta, r)
 \sum_{|k|\ge 1} 
\int_0^{2\pi} \int_{\tau''}^{\tau''+\frac{\ell+1}{\omega_0}} 
k^2(\Psi_\flat)_k^2(s^*,\cdot)
    \,ds^*\,d\phi\, d\theta\, dr\\
&=& B_q\omega_0^2
 \sum_{\ell=0}^{\infty} \left(1+\ell \right)^{-q} \int_{r^-_Y}^R \int_0^\pi \nu(\theta, r)
\int_0^{2\pi} \int_{\tau''}^{\tau''+\frac{\ell+1}{\omega_0}} 
(\partial_\phi\Psi_\flat)^2(s^*,\cdot)
    \,ds^*\,d\phi\, d\theta\, dr\\
    &\le&
    B_q\omega_0^2\sum_{\ell=0}^{\infty} \left(1+\ell \right)^{-q}
\frac{1+\ell}{\omega_0} \sup_{-\infty\le \bar\tau\le \infty}
 \int_{r^-_Y}^R \int_0^\pi \nu(\theta, r)
\int_0^{2\pi}
 \int_{\bar\tau}^{\bar\tau+1} 
(\partial_\phi\Psi_\flat)^2(s^*,\cdot)
    \,ds^*\,d\phi\, d\theta\, dr\\
          &\le&
    B_q\omega_0 \sup_{-\infty\le \bar\tau\le \infty}
 \int_{r^-_Y}^R \int_0^\pi \nu(\theta, r)
\int_0^{2\pi}
 \int_{\bar\tau}^{\bar\tau+1} 
(\partial_\phi\Psi_\flat)^2(s^*,\cdot)
    \,ds^*\,d\phi\, d\theta\, dr.
    \end{eqnarray*}
The above and $(\ref{yepyeni})$ give $(\ref{would2})$.
\end{proof}

\begin{lemma}
\label{yenilemma}
Under the assumptions of the previous lemma, if $\omega_0\le 1$, then
\begin{align*}
 \int_{r^-_Y}^R&\int_0^\pi \nu(\theta, r)
 \int_0^{2\pi }\int_{\bar\tau-1}^{\bar\tau}(\partial_\phi\Psi_{\flat})^2 dt^*\, d\phi\\
&\le
B \omega_0^{-1}
\sup_{-\infty\le \tilde\tau\le  \infty}
 \int_{r^-_Y}^R \int_0^\pi \nu(\theta, r) \int_0^{2\pi} \int_{\tilde\tau-1}^{\tilde\tau}
(\partial_{\phi}\Psi)^2 dt^*\, d\phi.
\end{align*}
\end{lemma}
\begin{proof}
For any $q>0$, we have 
\begin{align*}
|(\Psi_\flat)_k(t^*,\cdot)|&\le B_q (\omega_0 |k|) \int_{-\infty}^\infty (1+|\omega_0 k (t^*-s^*)|)^{-q}
|\Psi_k|(s^*,\cdot)\, ds^*\\ &\le B_q (\omega_0 |k|) \sum_{\ell=-\infty}^{\infty} (1+|\ell|)^{-q}
\int_{t^*+\frac{\ell}{\omega_0|k|}}^{t^*+\frac{\ell+1}{\omega_0 |k|}}
|\Psi_k|(s^*,\cdot)\, ds^*\\&\le  B_q (\omega_0 |k|)^{\frac 12} \sum_{\ell=-\infty}^{\infty} (1+|\ell|)^{-q}
\left (\int_{t^*+\frac{\ell}{\omega_0|k|}}^{t^*+\frac{\ell+1}{\omega_0 |k|}}
|\Psi_k|^2(s^*,\cdot)\, ds^*\right)^{\frac 12}.
\end{align*}
It follows, with the help of Cauchy-Schwarz, that for $q>1$,
\begin{align*}
|(\Psi_\flat)_k(t^*,\cdot)|^2 &\le B_q (\omega_0 |k|) \sum_{\ell=-\infty}^{\infty} (1+|\ell|)^{-q}
\int_{t^*+\frac{\ell}{\omega_0|k|}}^{t^*+\frac{\ell+1}{\omega_0 |k|}}
|\Psi_k|^2(s^*,\cdot)\, ds^*,
\end{align*}
and thus,
\begin{align*}
 \int_{r^-_Y}^R&\int_0^\pi \nu(\theta, r)
\int_{\bar\tau-1}^{\bar\tau} 
|(\Psi_\flat)_k(t^*,\cdot)|^2 dt^*\\
&\le B_q (\omega_0 |k|) \sum_{\ell=-\infty}^{\infty} (1+|\ell|)^{-q} \int_{r^-_Y}^R\int_0^\pi \nu(\theta, r)
\int_{\bar\tau-1}^{\bar\tau} \int_{t^*+\frac{\ell}{\omega_0|k|}}^{t^*+\frac{\ell+1}{\omega_0 |k|}}
|\Psi_k|^2(s^*,\cdot)\, ds^*\,dt^*\\ &\le B_q (\omega_0 |k|) \sum_{\ell=-\infty}^{\infty} (1+|\ell|)^{-q}
 \int_{r^-_Y}^R\int_0^\pi \nu(\theta, r)
\int_{\bar\tau-1+\frac{\ell}{\omega_0|k|}}^{\bar\tau+\frac{\ell+1}{\omega_0 |k|}} \int^{s^*-\frac{\ell}{\omega_0|k|}}_{s^*-\frac{\ell+1}{\omega_0 |k|}}
|\Psi_k|^2(s^*,\cdot)\,dt^*\, ds^*\\ &\le B_q \sum_{\ell=-\infty}^{\infty} (1+|\ell|)^{-q}
 \int_{r^-_Y}^R\int_0^\pi \nu(\theta, r)\int_{\bar\tau-1+\frac{\ell}{\omega_0|k|}}^{\bar\tau+\frac{\ell+1}{\omega_0 |k|}} 
|\Psi_k|^2(s^*,\cdot)\,ds^*.
\end{align*}
We then obtain
\begin{align*}
 \int_{r^-_Y}^R&\int_0^\pi \nu(\theta, r)
\int_0^{2\pi }\int_{\bar\tau-1}^{\bar\tau}(\partial_\phi\Psi_{\flat})^2 dt^*\, d\phi\\
&=
\sum_{|k|\ge 1}
 \int_{r^-_Y}^R\int_0^\pi \nu(\theta, r)\int_0^{2\pi } \int_{\bar\tau-1}^{\bar\tau}
 k^2 |(\Psi_\flat)_k(t^*,\cdot)|^2 dt^*\\
&\le B_q \sum_{\ell=-\infty}^{\infty} (1+|\ell|)^{-q} \sum_{|k|\ge 1}
 \int_{r^-_Y}^R\int_0^\pi \nu(\theta, r)
\int_0^{2\pi }\int_{\bar\tau-1+\frac{\ell}{\omega_0|k|}}^{\bar\tau+\frac{\ell+1}{\omega_0 |k|}} 
k^2 |\Psi_k|^2(s^*,\cdot)\,ds^*\\
&\le B_q \sum_{\ell=0}^{\infty} (1+\ell)^{-q} \sum_{|k|\ge 1}
 \int_{r^-_Y}^R\int_0^\pi \nu(\theta, r)\int_0^{2\pi }\int_{\bar\tau-1-\frac{\ell}{\omega_0}}^{\bar\tau+\frac{\ell+1}{\omega_0}} 
k^2 |\Psi_k|^2(s^*,\cdot)\,ds^*\\
&= B_q \sum_{\ell=0}^{\infty} (1+\ell)^{-q}
 \int_{r^-_Y}^R\int_0^\pi \nu(\theta, r)
\int_0^{2\pi }\int_{\bar\tau-1-\frac{\ell}{\omega_0}}^{\bar\tau+\frac{\ell+1}{\omega_0}} 
(\partial_{\phi}\Psi)^2\,ds^*\\
&\le B_q \sum_{\ell=0}^{\infty} (1+\ell)^{-q}(2\ell+2){\omega_0}^{-1}
\sup_{-\infty\le \tilde\tau\le  \infty}
 \int_{r^-_Y}^R\int_0^\pi \nu(\theta, r)
\int_0^{2\pi }\int_{\tilde\tau-1}^{\tilde\tau} 
(\partial_{\phi}\Psi)^2\,ds^*\\
&\le B_q \omega_0^{-1} \sup_{-\infty\le \tilde\tau\le  \infty}
 \int_{r^-_Y}^R\int_0^\pi \nu(\theta, r)
\int_0^{2\pi }\int_{\tilde\tau-1}^{\tilde\tau} 
(\partial_{\phi}\Psi)^2\,dt^*,
\end{align*}
where we have assumed $q$ sufficiently large, and have used ${\omega_0}^{-1}\ge 1$. The lemma follows after
fixing $q$.
\end{proof}

\subsubsection{Application to $\psi^\tau_{\flat}$}
From the above lemmas, we easily obtain the following statement, which
is the form we shall use later in this paper:
\begin{proposition}
Let $\tau''\ge\tau'$ and $\omega_0\le 1$. Then
\label{compderpro2}
\begin{eqnarray*}
\int_{\mathcal{R}(\tau',\tau'')\cap \{r^-_Y\le r\le R\}}  (\partial_{t^*}\psi^\tau_{\flat})^2 &\le&
B {\omega_0}^2
\int_{\mathcal{R}(\tau',\tau'')\cap \{r^-_Y\le r\le R\}}
(\partial_\phi \psi_{\flat}^\tau)^2\\
&&\hbox{}+
B \sup_{0\le\bar\tau\le \tau}  \int_{\Sigma(\bar\tau)}{\bf q}_e(\psi).
\end{eqnarray*}
\end{proposition}
\begin{proof}
In view of $(\ref{anote})$, it follows that  
\[
(\partial_\phi\psi^\tau_{\hbox{\Rightscissors}})^2 =(\chi^\tau_{\hbox{\Rightscissors}})^2
(\partial_\phi\psi)^2 \le B\, {\bf q}_e(\psi)
\]
in the region $r^-_Y\le r\le R$. In view also of the support of $\psi^\tau_{\hbox{\Rightscissors}}$,
we may thus bound the right hand side of the statement of Lemma~\ref{yenilemma}
applied to $\psi^\tau_{\hbox{\Rightscissors}}$
by
\[
B\omega_0^{-1}\sup_{1\le \bar\tau\le \tau}  \int_{\bar\tau-1}^{\bar\tau}
\left(\int_{\Sigma(\tilde{\tau})}{\bf q}_e(\psi)\right) d\tilde\tau .
\]
The proposition now follows from Lemmas~\ref{flatlem} and~\ref{yenilemma}.
\end{proof}

\subsubsection{Comparisons for $\Psi_{\sharp}$}
First a lemma:
\begin{lemma}
\label{sharplem}
Let $\tau'\le \tau''$ let $\Psi$ be smooth and of
compact support in $t^*$. Then
\[
\int_{\mathcal{H}(\tau',\tau'')} (\partial_{t^*}\Psi_{\sharp})^2
\ge B {\omega_0}^2 \int_{\mathcal{H}(\tau',\tau'')} (\partial_\phi \Psi_{\sharp})^2
- B{\omega_0}^{-1} \sup_{-\infty\le \bar\tau\le  \infty}
\int _{\mathcal{H}(\bar\tau,\bar\tau+1) }(\partial_{t^*}\Psi_{\sharp})^2,
\]
\end{lemma}
\begin{proof}
Note first that the lemma holds trivially for $(\Psi_{\sharp})_0$. We may thus
assume that $(\Psi_{\sharp})_0=0$. For $|k|\ge 1$,
we note first that from Section~\ref{defflat} we obtain
\[
|(\Psi_{\sharp})_k(t^*,\cdot)| \le B_q \omega_0^{-1}|k|^{-1}  \int_{-\infty}^{\infty}  
\omega_0|k|\frac{\left|\log |\omega_0k(t^*-s^*)|\right|}
{(1+|\omega_0k(t^*-s^*)|)^{q}}\left|\partial_{s^*}(\Psi_{\sharp})_k(s^*,\cdot)\right| \, ds^*.
\]
Thus,
\begin{eqnarray*}
&&\int_{\tau'}^{\tau''}
k^2 (\Psi_{\sharp})_k^2(t^*,\cdot)\,dt^*\\
&\le&B_q\omega_0^{-2} \int_{\tau'}^{\tau''}  \left( \int_{-\infty}^{\infty}  
\omega_0|k|\frac{\left|\log |\omega_0k(t^*-s^*)|\right|}
{(1+|\omega_0k(t^*-s^*)|)^{q}}\left|\partial_{s^*}(\Psi_{\sharp})_k(s^*,\cdot)\right| \, ds^*
\right)^2 \,dt^*\\
&\le&B_q\omega_0^{-2} \int_{\tau'}^{\tau''}  \left( \int_{\tau'}^{\tau''}  
\omega_0|k|\frac{\left|\log |\omega_0k(t^*-s^*)|\right|}
{(1+|\omega_0k(t^*-s^*)|)^{q}}\left|\partial_{s^*}(\Psi_{\sharp})_k(s^*,\cdot)\right| \, ds^*
\right)^2 \,dt^*\\
&&\hbox{}+B_q\omega_0^{-2} \int_{\tau'}^{\tau''}  \left( \int_{-\infty}^{\tau'}  
\omega_0|k|\frac{\left|\log |\omega_0k(t^*-s^*)|\right|}
{(1+|\omega_0k(t^*-s^*)|)^{q}}\left|\partial_{s^*}(\Psi_{\sharp})_k(s^*,\cdot)\right| \, ds^*
\right)^2 \,dt^*\\
&&\hbox{}+B_q\omega_0^{-2} \int_{\tau'}^{\tau''}  \left( \int_{\tau''}^{\infty}  
\omega_0|k|\frac{\left|\log |\omega_0k(t^*-s^*)|\right|}
{(1+|\omega_0k(t^*-s^*)|)^{q}}\left|\partial_{s^*}(\Psi_{\sharp})_k(s^*,\cdot)\right| \, ds^*
\right)^2 \,dt^*\\
&\doteq&T_{1,k}+T_{2,k}+T_{3,k}.
\end{eqnarray*}

We obtain immediately that for sufficiently large $q$,
since
\[
\int_{\tau'}^{\tau''}  
\omega_0|k|\frac{\left|\log |\omega_0k(t^*-s^*)|\right|}
{(1+|\omega_0k(t^*-s^*)|)^{q}} \,ds^*\le B_q
\]
then
\begin{align*}
&\int_{\tau'}^{\tau''}  \left( \int_{\tau'}^{\tau''}  
\omega_0|k|\frac{\left|\log |\omega_0k(t^*-s^*)|\right|}
{(1+|\omega_0k(t^*-s^*)|)^{q}}\left|\partial_{s^*}(\Psi_{\sharp})_k(s^*,\cdot)\right| \, ds^*
\right)^2\, dt^*\\
&\le B_q \int_{\tau'}^{\tau''}(\partial_t\Psi_{\sharp})_k^2(t^*,\cdot)\,dt^*
\end{align*}
and thus
\[
T_{1,k}\le B_q\omega_0^{-2}\int_{\tau'}^{\tau''}(\partial_{t^*}(\Psi_{\sharp})_k)^2(t^*,\cdot)\,dt^*.
\]

On the other hand, 
\begin{eqnarray*}
&&\sum_{|k|\ge 1}\int_0^\pi \tilde{\nu}(\theta)\int_0^{2\pi}
T_{2,k}\,d\phi\, d\theta\\
&=&B_q\omega_0^{-2}
\sum_{|k|\ge 1}\int_0^\pi \tilde{\nu}(\theta)\\
&&\hbox{}\int_0^{2\pi}
 \int_{\tau'}^{\tau''}  \left( \int_{-\infty}^{\tau'}  
\omega_0k\frac{\left|\log |\omega_0k(t^*-s^*)|\right|}
{(1+|\omega_0k(t^*-s^*)|)^{q}}\left|\partial_{s^*}(\Psi_{\sharp})_k(s^*,\cdot)\right| \, ds^*
\right)^2 \,dt^*\, d\phi\, d\theta\\
&=&B_q\omega_0^{-2} \sum_{|k|\ge 1}\int_0^\pi \tilde{\nu}(\theta)\\
&&\hbox{}\int_0^{2\pi}
\int_{\tau'}^{\tau''}  \left( \sum_{\ell=0}^{\infty}
\int_{\tau'-\frac{\ell+1}{\omega_0|k|}}^{\tau'-\frac{\ell}{\omega_0|k|}}  
\omega_0k\frac{\left|\log |\omega_0k(t^*-s^*)|\right|}
{(1+|\omega_0k(t^*-s^*)|)^{q}}\left|\partial_{s^*}(\Psi_{\sharp})_k(s^*,\cdot)\right| \, ds^*
\right)^2 \,dt^*\, d\phi\, d\theta\\
&\le&B_q\omega_0^{-2}\sum_{|k|\ge 1}\int_0^\pi \tilde{\nu}(\theta)\int_0^{2\pi} \int_{\tau'}^{\tau''}  \left( \sum_{\ell=0}^{\infty}
\left(\int_{\tau'-\frac{\ell+1}{\omega_0|k|}}^{\tau'-\frac{\ell}{\omega_0|k|}}  
\omega_0^2k^2\frac{\left|\log^2 |\omega_0k(t^*-s^*)|\right|}
{(1+|\omega_0k(t^*-s^*)|)^{2q}} ds^*\right)^{1/2}\right.\\
&&\left.\hbox{}\cdot
\left(\int_{\tau'-\frac{\ell+1}{\omega_0|k|}}^{\tau'-\frac{\ell}{\omega_0|k|}} 
 \left|\partial_{s^*}(\Psi_{\sharp})_k(s^*,\cdot)\right|^2 \, ds^*\right)^{1/2}
\right)^2 \,dt^*\, d\phi\, d\theta\\
&\le&B_q\omega_0^{-2}\sum_{|k|\ge 1}\int_0^\pi \tilde{\nu}(\theta)\\
&&\hbox{}\int_0^{2\pi}
  \int_{\tau'}^{\tau''}\omega_0|k|  \left( \sum_{\ell=0}^{\infty}
\frac{\left|\log |\omega_0k(t^*-\tau')+\ell)|\right|}
{(1+|\omega_0k(t^*-\tau')+\ell|)^{q}} \left(\int_{\tau'-\frac{\ell+1}{\omega_0|k|}}^{\tau'-\frac{\ell}{\omega_0|k|}} 
 \left|\partial_{s^*}(\Psi_{\sharp})_k(s^*,\cdot)\right|^2 \, ds^*\right)^{1/2}
 \right)^{2}dt^*\, d\phi\, d\theta\\
 \end{eqnarray*}
 \begin{eqnarray*}
 &\le&B_q\omega_0^{-2}\sum_{|k|\ge 1}\int_0^\pi \tilde{\nu}(\theta)\\
 &&\hbox{}\int_0^{2\pi}
  \int_{0}^{\infty} \left( \sum_{\ell=0}^{\infty}
\frac{\left|\log |t^*+\ell)|\right|}
{(1+|t^*+\ell|)^{q}} \left(\int_{\tau'-\frac{\ell+1}{\omega_0|k|}}^{\tau'-\frac{\ell}{\omega_0|k|}} 
 \left|\partial_{s^*}(\Psi_{\sharp})_k(s^*,\cdot)\right|^2 \, ds^*\right)^{1/2}
 \right)^{2}dt^*\, d\phi\, d\theta\\
  &\le&B_q\omega_0^{-2}\sum_{|k|\ge 1}\int_0^\pi \tilde{\nu}(\theta)
  \int_0^{2\pi}  \int_{0}^{\infty} \left( \sum_{\ell=0}^{\infty}
\frac{\left|\log |t^*+\ell)|\right|}
{(1+|t^*+\ell|)^{q}} \right)\\
 &&\hbox{}
\left( \sum_{\ell=0}^{\infty}
\frac{\left|\log |t^*+\ell)|\right|}
{(1+|t^*+\ell|)^{q}} 
\int_{\tau'-\frac{\ell+1}{\omega_0|k|}}^{\tau'-\frac{\ell}{\omega_0|k|}} 
 \left|\partial_{s^*}(\Psi_{\sharp})_k(s^*,\cdot)\right|^2 \, ds^*\right)
 dt^*\, d\phi\, d\theta\\
   &\le&B_q\omega_0^{-2}  \int_{0}^{\infty} \left( \sum_{\ell=0}^{\infty}
\frac{\left|\log |t^*+\ell)|\right|}
{(1+|t^*+\ell|)^{q}} \right)\\
 &&\hbox{}
\left( \sum_{\ell=0}^{\infty}
\frac{\left|\log |t^*+\ell)|\right|}
{(1+|t^*+\ell|)^{q}} \sum_{|k|\ge 1}\int_0^\pi \tilde{\nu}(\theta)
  \int_0^{2\pi}
\int_{\tau'-\frac{\ell+1}{\omega_0|k|}}^{\tau'-\frac{\ell}{\omega_0|k|}} 
 \left|\partial_{s^*}(\Psi_{\sharp})_k(s^*,\cdot)\right|^2 \, ds^*\,
d\phi\, d\theta^*\right) \, dt^*\\
   &\le&B_q\omega_0^{-2}  \int_{0}^{\infty}
\frac{\left|\log |t^*|\right|}
{(1+|t^*|)^{q}}\\
 &&\hbox{}
\left( \sum_{\ell=0}^{\infty}
\frac{\left|\log |t^*+\ell|\right|}
{(1+|t^*+\ell|)^{q}} \sum_{|k|\ge 1}\int_0^\pi \tilde{\nu}(\theta)
  \int_0^{2\pi}
\int^{\tau'}_{\tau'-\frac{\ell+1}{\omega_0}} 
 \left|\partial_{s^*}(\Psi_{\sharp})_k(s^*,\cdot)\right|^2 \, ds^*\,
d\phi\, d\theta^*\right) \, dt^*\\
   &\le&B_q\omega_0^{-2}  \int_{0}^{\infty}
\frac{\left|\log |t^*|\right|}
{(1+|t^*|)^{q}}\\
 &&\hbox{}
\left( \sum_{\ell=0}^{\infty}
\frac{\left|\log |t^*+\ell|\right|}
{(1+|t^*+\ell|)^{q}} \int_0^\pi \tilde{\nu}(\theta)
  \int_0^{2\pi}
\int^{\tau'}_{\tau'-\frac{\ell+1}{\omega_0}} 
 \left|\partial_{s^*}(\Psi_{\sharp})(s^*,\cdot)\right|^2 \, ds^*\,
d\phi\, d\theta^*\right) \, dt^*\\
   &\le&B_q\omega_0^{-2}  \int_{0}^{\infty}
\frac{\left|\log |t^*|\right|}
{(1+|t^*|)^{q}}
 \sum_{\ell=0}^{\infty}
\frac{\left|\log |t^*+\ell|\right|}
{(1+|t^*+\ell|)^{q}} \frac{\ell+1}{\omega_0}\\
 &&\hbox{}
\left(\sup_{-\infty\le\bar\tau\le\infty}\int_0^\pi \tilde{\nu}(\theta)
  \int_0^{2\pi}
\int^{\bar\tau}_{\bar\tau+1} 
 \left|\partial_{s^*}(\Psi_{\sharp})(s^*,\cdot)\right|^2 \, ds^*\,
d\phi\, d\theta^* \right) \, dt^*\\
&\le&B_q\omega_0^{-3}
\sup_{-\infty\le\bar\tau\le\infty}\int_0^\pi \tilde{\nu}(\theta)
  \int_0^{2\pi}
\int^{\bar\tau}_{\bar\tau+1} 
 \left|\partial_{s^*}(\Psi_{\sharp})(s^*,\cdot)\right|^2 \, ds^*\,
d\phi\, d\theta^*.
\end{eqnarray*}
A similar bound holds for $T_{3,k}$. We obtain the lemma
after appropriate fixing of $q$.
\end{proof}

\begin{lemma}
Under the assumptions of the previous lemma, if $\omega_0\le 1$, 
\label{sharplem2}
\[
\sup_{-\infty\le \bar\tau\le  \infty}
\int _{\mathcal{H}(\bar\tau,\bar\tau+1) }(\partial_{t^*}\Psi_{\sharp})^2
\le 
B\omega_0^{-1}\sup_{-\infty \le \bar\tau\le  \infty}
\int _{\mathcal{H}(\bar\tau,\bar\tau+1) }(\partial_{t^*}\Psi)^2
\]
\end{lemma}
\begin{proof}
Since
\[
\partial_{t^*}\Psi_{\sharp}=\partial_{t^*}\Psi
-\partial_{t^*}\Psi_{\flat},
\]
and ${\omega_0}^{-1}\ge 1$, it suffices in fact to prove
\[
\sup_{-\infty\le \bar\tau\le  \infty}
\int _{\mathcal{H}(\bar\tau,\bar\tau+1) }(\partial_{t^*}\Psi_{\flat})^2
\le B\omega_0^{-1}
\sup_{-\infty \le \bar\tau\le  \infty}
\int _{\mathcal{H}(\bar\tau,\bar\tau+1) }(\partial_{t^*}\Psi)^2.
\]

Recall from Section~\ref{defflat} we have
\[
|(\partial_{t^*}\Psi_{\flat})_k| 
\le
B_q \int_{-\infty}^{\infty}\omega_0|k|(1+\omega_0|k| |s^*-t^*|)^{-q}|(\partial_{t^*}
\Psi)_k|\,ds^*.
\]
We obtain
\begin{eqnarray*}
&&\int _{\tilde\tau}^{\tilde\tau+1}(\partial_{t^*}\Psi_{\flat})^2_k(t^*,\cdot)dt^*\\
&\le& 
B_q\int_{\tilde\tau}^{\tilde\tau+1}
\left( \int_{-\infty}^{\infty}\frac{\omega_0|k|}{(1+\omega_0|k| |s^*-t^*|)^{q}}|(\partial_{t^*}
\Psi)_k(s^*,\cdot)|\,ds^*\right)^2 dt^*\\
&\le& 
B_q \int_{\tilde\tau}^{\tilde\tau+1}
\left(\sum_{\ell=-\infty}^{\infty} \int_{t^*+\frac{\ell}{\omega_0|k|}}^{
t^*+\frac{\ell+1}{\omega_0|k|}}\frac{\omega_0|k|}{(1+\omega_0|k| |s^*-t^*|)^{q}}|(\partial_{t^*}
\Psi)_k(s^*,\cdot)|\,ds^*\right)^2 dt^*\\
&\le& 
B_q (\omega_0 k)^2 \int_{\tilde\tau}^{\tilde\tau+1}
\left(\sum_{\ell=-\infty}^{\infty}(1+|\ell|)^{-q} \int_{t^*+\frac{\ell}{\omega_0|k|}}^{t^*+
\frac{\ell+1}{\omega_0|k|}}|(\partial_{t^*}
\Psi)_k(s^*,\cdot)|\,ds^*\right)^2 dt^*\\
&\le& 
B_q (\omega_0 k)^2 \int_{\tilde\tau}^{\tilde\tau+1}
\left(\sum_{\ell=-\infty}^{\infty}(1+|\ell|)^{-q} 
\right)
\sum_{\ell=-\infty}^{\infty}(1+|\ell|)^{-q} (\omega_0|k|)^{-1}
\int_{t^*+\frac{\ell}{\omega_0|k|}}^{t^*+
\frac{\ell+1}{\omega_0|k|}}|(\partial_{t^*}
\Psi)_k(s^*,\cdot)|^2\,ds^*\,dt^*\\
&\le& 
B_q (\omega_0|k|) \int_{\tilde\tau}^{\tilde\tau+1}
\sum_{\ell=-\infty}^{\infty}(1+|\ell|)^{-q} 
\int_{t^*+\frac{\ell}{\omega_0|k|}}^{t^*+
\frac{\ell+1}{\omega_0|k|}}|(\partial_{t^*}
\Psi)_k(s^*,\cdot)|^2\,ds^*\,dt^*\\
&\le& 
B_q (\omega_0|k|)
\sum_{\ell=-\infty}^{\infty}(1+|\ell|)^{-q}  \int_{\tilde\tau}^{\tilde\tau+1}
\int_{t^*+\frac{\ell}{\omega_0|k|}}^{t^*+
\frac{\ell+1}{\omega_0|k|}}|(\partial_{t^*}
\Psi)_k(s^*,\cdot)|^2\,ds^*\,dt^*\\
&\le& 
B_q (\omega_0|k|)
\sum_{\ell=-\infty}^{\infty}(1+|\ell|)^{-q}  \int_{\tilde\tau+\frac{\ell}{\omega_0|k|}}^{\tilde\tau+1
+\frac{\ell+1}{\omega_0|k|}}
\int_{s^*-\frac{\ell+1}{\omega_0|k|}}^{s^*-
\frac{\ell}{\omega_0|k|}}|(\partial_{t^*}
\Psi)_k(s^*,\cdot)|^2\,dt^*\,ds^*\\
&\le& 
B_q
\sum_{\ell=-\infty}^{\infty}(1+|\ell|)^{-q}  \int_{\tilde\tau+\frac{\ell}{\omega_0|k|}}^{\tilde\tau+1
+\frac{\ell+1}{\omega_0|k|}}
|(\partial_{t^*}
\Psi)_k(s^*,\cdot)|^2\,ds^*\\
&\le& 
B_q
\sum_{\ell=0}^{\infty}(1+|\ell|)^{-q}  \int_{\tilde\tau-\frac{\ell}{\omega_0}}^{\tilde\tau+1
+\frac{\ell+1}{\omega_0}}
|(\partial_{t^*}
\Psi)_k(s^*,\cdot)|^2\,ds^*
\end{eqnarray*}
for $q>1$.

Integrating and summing over $k$, we obtain
\begin{align*}
\int_0^\pi \tilde{\nu}(\theta) \int_0^{2\pi}& \int_{\tilde{\tau}}^{\tilde{\tau}+1}
(\partial_{t^*}\Psi_{\flat})^2(t^*,\cdot)dt^*\, d\phi\, d\theta\\
&\le 
B_q
\sum_{\ell=0}^{\infty}(1+\ell)^{-q}\int_0^\pi \tilde{\nu}(\theta) \int_0^{2\pi}  
\int_{\tilde\tau-\frac{\ell}{\omega_0}}^{\tilde\tau+1
+\frac{\ell+1}{\omega_0}}
|(\partial_{t^*}
\Psi)(s^*,\cdot)|^2\,ds^*\\
&\le 
B_q\omega_0^{-1}
\sum_{\ell=0}^{\infty}(1+\ell)^{-q}(2\ell+2)
\sup_{-\infty<\bar{\tau}<\infty}\int_0^\pi \tilde{\nu}(\theta) \int_0^{2\pi}  \int_{\bar\tau}^{\bar\tau+1}
|(\partial_{t^*}
\Psi)(s^*,\cdot)|^2\,ds^*,
\end{align*}
where we have used in the last line that $\omega_0\le 1$.
 The lemma follows after fixing $q>2$.
\end{proof}

\subsubsection{Application to $\psi^\tau_{\sharp}$}
We may now easily prove
\begin{proposition}
\label{compderpro}
Let $0\le \tau' <\tau''\le \tau$ and let
$\psi$ be as in Theorem~\ref{kevtriko}. We have
\[
\int_{\mathcal{H}(\tau',\tau'')} (\partial_{t^*}\psi^\tau_{\sharp})^2
\ge B {\omega_0}^2 \int_{\mathcal{H}(\tau',\tau'')} (\partial_\phi \psi^\tau_{\sharp})^2
- B{\omega_0}^{-2}e^{-1} \sup_{0\le \bar \tau\le  \tau-1}
\int _{\mathcal{H}(\bar\tau,\bar\tau+1) }J_\mu^{N_e}(\psi)n^\mu.
\]
\end{proposition}

\begin{proof}
To prove the proposition from the above lemmas, 
we first remark that it suffices to prove the inequality
under the assumption
\[
(\psi_{\sharp}^\tau)_0= \int_0^{2\pi} \psi_{\sharp}^\tau d\phi=0,
\]
for the inequality is trivially true for $(\psi_{\sharp}^\tau)_0$. 
By $(\ref{whereas})$, this is equivalent to assuming
\[
\int_0^{2\pi} \psi_{\hbox{\Rightscissors}}^\tau d\phi=0,
\]
and by $(\ref{anote})$
\[
\int_0^{2\pi}\psi\, d\phi=0
\]
in the support of $\psi_{\hbox{\Rightscissors}}$.
Of course, under this assumption it follows that this holds in all of $\mathcal{R}$.
Thus we may assume
\begin{equation}
\label{olur}
\int_0^{2\pi}\psi^2\, d\phi \le \int_0^{2\pi}(\partial_\phi\psi)^2 \, d\phi
\end{equation}
in the relevant region.
From the above lemma, we just notice that on $\mathcal{H}(0,\tau)$
\begin{eqnarray*}
\int_0^{2\pi}(\partial_{t^*}\psi^\tau_{\hbox{\Rightscissors}})^2\, d\phi 
&\le&\int_0^{2\pi} ( (\partial_{t^*} \chi^\tau_{\hbox{\Rightscissors}}) \psi +
 \chi^\tau_{\hbox{\Rightscissors}} \partial_{t^*}\psi)^2\, d\phi\\
&\le& \int_0^{2\pi} (B\psi^2 + B(\partial_t\psi)^2)\, d\phi \\
&\le& \int_0^{2\pi} (B(\partial_\phi\psi)^2 +B(\partial_t\psi)^2) d\phi \\
&\le& \int_0^{2\pi} Be^{-1}J^{N_e}_\mu n^\mu_{\mathcal{H}},
\end{eqnarray*}
where we have used
$(\ref{olur})$, $(\ref{tangcont})$ and $(\ref{tangcont2})$.
The proposition follows.
\end{proof}

\subsection{Comparing ${\bf q}_e(\psi_{\flat}^\tau)$,  ${\bf q}_e(\psi_{\sharp}^\tau)$ and ${\bf q}_e(\psi)$}
\label{sugkri}

In view of $(\ref{fusika})$, we clearly have the pointwise relation
\begin{equation}
\label{fusika2}
{\bf q}_e(\psi)\le 2\left({\bf q}_e(\psi_{\flat}^\tau)+
{\bf q}_e(\psi_{\sharp}^\tau)\right) 
\end{equation}
in $\mathcal{R}(1,\tau-1)$.
It will be necessary, however, to compare also in the opposite direction.
We have
\begin{proposition}
\label{yenicomp}
Let $\omega_0\le 1\le\tau_{\rm step}\le \tau' \le \tau-\tau_{\rm step}$. Then
\begin{eqnarray*}
\int_{\tau'-1}^{\tau'} dt^* \int_{\Sigma(t^*)} {\bf q}_e(\psi_{\flat}^\tau)&\le& 
B{\omega_0}^{-1}\sup_{\tau'-\tau_{\rm step} \le \bar\tau\le \tau'+\tau_{\rm step}} 
\int_{\Sigma(\bar\tau)} {\bf q}_e(\psi)\\
&&\hbox{}+B{\omega_0}^{-7}e^{-1}
\tau_{\rm step}^{-2} \sup_{0\le\bar \tau\le \tau}\int_{\Sigma(\bar\tau)} {\bf q}_e(\psi),
\end{eqnarray*}
\begin{eqnarray*}
\int_{\tau'-1}^{\tau'} dt^* \int_{\Sigma(t^*)} {\bf q}_e(\psi_{\sharp}^\tau)&\le& 
B\omega_0^{-1}\sup_{\tau'-\tau_{\rm step} \le \bar\tau\le \tau'+\tau_{\rm step}} 
\int_{\Sigma(\bar\tau)}{\bf q}_e(\psi)\\
&&\hbox{}
+B{\omega_0}^{-7}
e^{-1}\tau_{\rm step}^{-2} \sup_{0\le\bar\tau\le \tau}\int_{\Sigma(\bar\tau)} {\bf q}_e(\psi).
\end{eqnarray*}
\end{proposition}
\begin{proof}
Since $\psi^\tau_\sharp=\psi^\tau_{\hbox{\Rightscissors}}-\psi^\tau_\flat$ and 
${\bf q}_e(\psi^\tau_\sharp) \le 2({\bf q}_e(\psi^\tau)+{\bf q}_e(\psi^\tau_\flat))$
it will be sufficient to prove the first statement of the proposition. We begin with the
following
\begin{lemma}
Let $\Psi$ be smooth of compact
support in $t^*$ and $\omega_0\le 1$. Then
\begin{align*}
\int_{\tau'-1}^{\tau'} dt^*
\int_{\Sigma(t^*)} \Psi_\flat^2 &\le B \omega_0^{-1}
\sup_{\tau'-\tau_{\rm step} \le \bar\tau\le \tau'+\tau_{\rm step}} 
\int_{\Sigma(\bar\tau)} \Psi^2 \\ &+ B \sup_{-\infty\le\bar\tau\le \tau'-\tau_{\rm step} \cup 
\tau'+\tau_{\rm step}\le\bar\tau\le\infty} \omega_0^{-7} |\bar\tau-\tau_{\rm step}|^{-6}
\int_{\Sigma(\bar\tau)} \Psi^2.
\end{align*}
\end{lemma}
\begin{proof}
For any $q>0$, we have 
\begin{align*}
|(\Psi_\flat)_k(t^*,\cdot)|&\le B_q (\omega_0 |k|) \int_{-\infty}^\infty (1+|\omega_0 k (t^*-s^*)|)^{-q}
|\Psi_k|(s^*,\cdot)\, ds^*\\ &\le B_q (\omega_0 |k|) \sum_{\ell=-\infty}^{\infty} (1+|\ell|)^{-q}
\int_{t^*+\frac{\ell}{\omega_0|k|}}^{t^*+\frac{\ell+1}{\omega_0 |k|}}
|\Psi_k|(s^*,\cdot)\, ds^*\\&\le  B_q (\omega_0 |k|)^{\frac 12} \sum_{\ell=-\infty}^{\infty} (1+|\ell|)^{-q}
\left (\int_{t^*+\frac{\ell}{\omega_0|k|}}^{t^*+\frac{\ell+1}{\omega_0 |k|}}
|\Psi_k|^2(s^*,\cdot)\, ds^*\right)^{\frac 12}.
\end{align*}
Therefore,
\begin{align*}
\int_{\tau'-1}^{\tau'} |(\Psi_\flat)_k(t^*,\cdot)|^2 dt^* &\le B_q (\omega_0 |k|)
\sum_{\ell=-\infty}^{\infty} (1+|\ell|)^{-q}
\int_{\tau'-1}^{\tau'} dt^* \int_{t^*+\frac{\ell}{\omega_0|k|}}^{t^*+\frac{\ell+1}{\omega_0 |k|}}
|\Psi_k|^2(s^*,\cdot)\, ds^*\\ &\le B_q (\omega_0 |k|)
\sum_{\ell=-\infty}^{\infty} (1+|\ell|)^{-q}
\int_{\tau'-1}^{\tau'} dt^* \int_{t^*+\frac{\ell}{\omega_0|k|}}^{t^*+\frac{\ell+1}{\omega_0 |k|}}
|\Psi_k|^2(s^*,\cdot)\, ds^*\\ &\le B_q 
\sum_{\ell=-\infty}^{\infty} (1+|\ell|)^{-q}
\int_{\tau'-1+\frac{\ell}{\omega_0|k|}}^{\tau'+\frac{\ell+1}{\omega_0 |k|}}
|\Psi_k|^2(s^*,\cdot)\, ds^*.
\end{align*}
As a consequence,
\begin{align*}
\int_{\tau'-1}^{\tau'}\int_{2M}^\infty&\int_0^{\pi} \nu(r,\theta)\int_0^{2\pi}
|(\Psi_\flat)(t^*,\cdot)|^2 d\phi\, d\theta\, dr\, dt^*\\
=&
\sum_{|k|\ge 1} \int_{\tau'-1}^{\tau'} \int_{2M}^\infty\int_0^{\pi} \nu(r,\theta)\int_0^{2\pi}
 |(\Psi_\flat)_k(t^*,\cdot)|^2 dt^*
\\ \le& B_q 
\sum_{|k|\ge 1} \sum_{\ell= -\infty}^{\infty} (1+|\ell|)^{-q}
\int_{\tau'-1+\frac{\ell}{\omega_0|k|}}^{\tau'+\frac{\ell+1}{\omega_0 |k|}}
\int_{2M}^\infty\int_0^{\pi} \nu(r,\theta)\int_0^{2\pi}
\int^{s^*-\frac \ell{\omega_0 |k|}}_{s^*-\frac{\ell+1}{\omega_0 |k|}} 
|\Psi_k|^2(s^*,\cdot)\,dt^*\, ds^*\\ \le& B_q 
\sum_{\ell\ge 0} (1+|\ell|)^{-q}\int_{2M}^\infty\int_0^{\pi} \nu(r,\theta)\int_0^{2\pi}
\int_{\tau'-1-\frac{\ell}{\omega_0}}^{\tau'+\frac{\ell+1}{\omega_0}}
|\Psi|^2(s^*,\cdot)\, ds^*\\ \le& B\omega_0^{-1} \sup_{\tau'-\tau_{\rm step} \le \bar\tau\le \tau'+\tau_{\rm step}} 
\int_{\Sigma(\bar\tau)} \Psi^2 \\ &+ B \sup_{-\infty\le\bar\tau\le \tau'-\tau_{\rm step} \cup 
\tau'+\tau_{\rm step}\le\bar\tau\le\infty} \omega_0^{-7} \left (\tau_{\rm step}+|\bar\tau-\tau_{\rm step}|\right )^{-6}
\int_{\Sigma(\bar\tau)} \Psi^2
\end{align*}
for $q$ chosen sufficiently large.
\end{proof}

Note that
$$
\left ((\pa_v\Psi)^2+(\pa_u\Psi)^2+|\nabb\Psi|^2+ e\frac{(\pa_u\Psi)^2}{(1-\mu)^2} \right)\sim
{\bf q}_e(\Psi),
$$
and
as a consequence,
$$
(\pa_v\Psi)_\flat^2+(\pa_u\Psi)_\flat^2+(\partial_{z^A}\Psi)_\flat^2+
(\partial_{z^B}\Psi)_\flat^2 +e\left(\frac{\pa_u\Psi}{1-\mu}\right)_\flat^2 
\sim
{\bf q}_e(\Psi_\flat),
$$
where $z^A$ denote alternative coordinates
 $x^A$ or $\tilde{x}^A$  of our atlas $(\ref{theatlas})$.
Thus, we obtain from the above lemma applied to
$\Psi=\partial_v\psi^\tau_{\hbox{\Rightscissors}}$, $\Psi=\partial_u\psi^\tau_{\hbox{\Rightscissors}}$,
$\Psi= \partial_{z^A}\psi^\tau_{\hbox{\Rightscissors}}$,\footnote{Of course, one needs to multiply
this by a cutoff on the sphere to make it a well defined smooth function.} $\Psi=\sqrt{e}\frac{\partial_u \psi^\tau_{\hbox{\Rightscissors}}}
{1-\mu}$, 
the statement
\begin{align}
\label{thest}
\int_{\tau'-1}^{\tau'} dt^*& \int_{\Sigma(t^*)} {\bf q}_e(\psi_{\flat}^\tau)\le
B\omega_0^{-1}\sup_{\tau'-\tau_{\rm step} \le \bar\tau\le \tau'+\tau_{\rm step}} 
\int_{\Sigma(\bar\tau)} {\bf q}_e(\psi^\tau_{\hbox{\Rightscissors}})\\
\nonumber
&+B\sup_{-\infty\le\bar\tau\le \tau'-\tau_{\rm step} \cup 
\tau'+\tau_{\rm step}\le\bar\tau\le\infty} {\omega_0}^{-7}
\left(\tau_{\rm step}+|\bar\tau-\tau_{\rm step}|\right)^{-6} \int_{\Sigma(\bar\tau)} {\bf q}_e(\psi^\tau_{\hbox{\Rightscissors}}).
\end{align}

Note that it is sufficient to prove the inequality under the assumption $\psi_0=0$, and thus
we may assume $(\ref{olur})$.
Note the inequality
\begin{equation}
\label{NN}
{\bf q}_e(\psi^\tau_{\hbox{\Rightscissors}})(t,r,\cdot)
\le {\bf q}_e
(\psi)(t,r,\cdot)
+B \chi(t^++1-\tau) \chi(-t^-) \psi^2.
\end{equation}
Now, $B$ can be chosen such that
 in the support of
the first term on the right hand side of $(\ref{thest})$
either $r\le B$ or 
$\psi^\tau_{\hbox{\Rightscissors}}=\psi$.
In view of $(\ref{olur})$, it follows
thus that we may there replace
${\bf q}_e(\psi^\tau_{\hbox{\Rightscissors}})$ with 
${\bf q}_e(\psi)$.

Turning to the second supremum term of $(\ref{thest})$ and
applying 
\[
|t^*-\tau_{\rm step}|^{-4}({}^+\chi^\tau_{\hbox{\Rightscissors}}+
^{-}\chi^\tau_{\hbox{\Rightscissors}}) \le B r^{-2},
\] 
the statement of the proposition follows immediately 
in view of the restriction on $\tau_{\rm step}$ and $(\ref{olur})$.
\end{proof}

\subsection{Estimates for $\mathcal{E}$}
\label{errest}
In view of the cutoffs, $\psi^\tau_{\flat}$ and $\psi^\tau_{\sharp}$ no
longer satisfy $(\ref{waveee})$. 

Define
\begin{equation}
\label{Ftaudef}
F^{\tau}_{\hbox{\Rightscissors}}
 = \psi\,\Box_g\chi_{\hbox{\Rightscissors}}^\tau  + g^{\mu\nu}\partial_\mu(\chi_
{\hbox{\Rightscissors}}^\tau)\partial_\nu\psi.
\end{equation}
Note that
$F^\tau_{\hbox{\Rightscissors}}$ is supported in  $\mathcal{R}^{+}(\tau-1,\tau)\cup
\mathcal{R}^{-}(0,1)$.

We may write
\begin{equation}
\label{rhslow}
\Box_g\psi^\tau_{\flat}=F^{\tau}_{\flat},
\end{equation}
\begin{equation}
\label{rhshigh}
\Box_g\psi^\tau_{\sharp}=F^{\tau}_{\sharp},
\end{equation}
where $F^\tau_{\flat}$ and $F^\tau_{\sharp}$ are defined from $F^\tau_{\hbox{\Rightscissors}}$
as in Section~\ref{defflat}.

The right hand sides of $(\ref{rhslow})$ and $(\ref{rhshigh})$ generate
error terms in applying $(\ref{basicid})$ with our various currents. 
We have the following
\begin{proposition}
\label{exoume}
Let
$\omega_0\le 1\le \tau_{\rm step}\le \tau'\le \tau'' \le \tau-\tau_{\rm step}$
and consider $V={\bf X}$, 
${N_e}$, or $T$. Then 
the following holds
\begin{eqnarray}
\label{xamla9}
\int_{\mathcal{R}(\tau',\tau'')}  \mathcal{E}^V(\psi^\tau_{\flat}) &\le&  B\omega_0^{-8}
\tau^{-2}_{\rm step}e^{-1}
\sup_{0\le\bar\tau \le \tau}\int_{\Sigma(\bar\tau)}{\bf q}_e(\psi),
\end{eqnarray}
\begin{eqnarray}
\label{ynlla9}
\int_{\mathcal{R}(\tau',\tau'')}  \mathcal{E}^V(\psi^\tau_{\sharp}) &\le&
B\omega_0^{-8}\tau^{-2}_{\rm step}e^{-1}
\sup_{0\le \bar\tau \le \tau}\int_{\Sigma(\bar\tau)}{\bf q}_e(\psi).
\end{eqnarray}
\end{proposition}

\begin{proof}
Decompose
\[
F^\tau_{\hbox{\Rightscissors}}= {}^1F^\tau_{\hbox{\Rightscissors}}+ {}^2F^\tau_{\hbox{\Rightscissors}}
\]
where
\[
{}^1 F^\tau_{\hbox{\Rightscissors}} = {}^{-} \chi^\tau_{\hbox{\Rightscissors}} 
F^\tau_{\hbox{\Rightscissors}}, \qquad
{}^2 F^\tau_{\hbox{\Rightscissors}} = {}^{+} \chi^\tau_{\hbox{\Rightscissors}} 
F^\tau_{\hbox{\Rightscissors}},
\]
and consider ${}^jF^\tau_{\flat}$, ${}^jF^\tau_{\sharp}$, defined in Section~\ref{defflat},
for $j=1,2$.

Recall the definitions $(\ref{basikne})$, $(\ref{basikne2})$ of $\mathcal{E}^V$ and $\mathcal{E}^{V,w}$.
Since $(F^\tau_{\flat})_0=0$ and $(F^\tau_{\sharp})_0=
\int_0^{2\pi} F^\tau_{\hbox{\Rightscissors}}\,d\phi= 0$
in $\mathcal{R}(\tau',\tau'')$,
it follows that $\mathcal{E}^V((\psi^\tau_\flat)_0)=\mathcal{E}^V((\psi^\tau_{\sharp})_0)=0$
in $\mathcal{R}(\tau',\tau'')$ and 
thus equations $(\ref{xamla9})$ and
$(\ref{ynlla9})$ are trivially satisfied. By subtraction,
we may thus assume in what follows
that  
\[
\psi_0=\int_0^{2\pi} \psi \, d\phi  = 0,
\]
and thus 
\begin{equation}
\label{ilerde}
r^{-2} \int_0^{2\pi} \psi^2 d\phi \le r^{-2} \int_0^{2\pi} (\partial_\phi \psi)^2 \, d\phi
\le B e^{-1}\int_0^{2\pi} J^{N_e}_\mu (\psi) n_{\Sigma^{+}}^\mu\, d\phi
\end{equation}
and similarly with $n_{\Sigma^{-}}^\mu$.

\begin{lemma}
For any $q\ge 0$, $\tau_0\le \tau-1$,
 there exists a $B_q$ such that
\[
|({}^2F^\tau_{\flat})_k| (t^{+}=\tau_0, \cdot)
\le   B_q{\omega_0}^{1-q} (\tau- \tau_0)^{-q} k^{1-q} \int_{\tau-1}^\tau|({}^2F^\tau_{\hbox{\Rightscissors}})_k(t^{+},\cdot)|
 \,dt^{+},
\]
\[
|({}^2F^\tau_{\sharp})_k| (t^{+}=\tau_0, \cdot)
\le   B_q{\omega_0}^{1-q}(\tau- \tau_0)^{-q} k^{1-q} 
\int_{\tau-1}^\tau|({}^2F^\tau_{\hbox{\Rightscissors}})_k(t^{+},\cdot)|
 \,dt^{+}.
\]
\end{lemma}
\begin{proof}
This is standard.
\end{proof}

It follows from the above lemma 
applied to $q=6$,
the restriction on $\tau'$, $\tau''$,
and the relation between $t^*$ and $t^{+}$
 that
\begin{eqnarray*}
\int_{\tau'}^{\tau''}   (1+(\tau-t^*))^3 ({}^2F^\tau_{\flat})^2_k \,dt^*
&\le&   B\tau_{\rm step}^{-5}{\omega_0}^{-8} \int_{\tau-1}^{\tau} r^{-2}({}^2F^\tau_{\hbox{\Rightscissors}})^2_k
\, dt^{+}\\
&\le&    B\tau_{\rm step}^{-5} {\omega_0}^{-8} \int_{\tau-1}^{\tau}  r^{-2}(\psi^2_k +
e^{-1} \,J^{N_e}_\mu(\psi_k)n^\mu_{\Sigma^+})\, dt^{+}.
\end{eqnarray*}
We remark that the powers of $\tau_{\rm step}^{-1}$ and $r^{-1}$ can be chosen 
arbitrarily above, at the expense of the constant $B$ and powers of ${\omega_0}^{-1}$,
but this would 
give no advantage in what follows.
Thus,
\begin{align*}
\int_{\mathcal{R}(\tau',\tau'')} &   (1+(\tau-t^*))^3({}^2F^\tau_{\flat})^2 \\
&\le
\sum_k    B\tau_{\rm step}^{-5}{\omega_0}^{-8} \int_{\mathcal{R}^{+}(\tau-1,\tau)}
( r^{-2}\psi^2_k+ r^{-2}e^{-1}\, J^{N_e}_\mu(\psi_k)n^\mu_{\Sigma^+}   )\\
&\le   B\tau_{\rm step}^{-5} {\omega_0}^{-8}\sup_{\tau-1 \le \bar\tau\le \tau} 
\int_{\Sigma^{+}(\bar\tau)}r^{-2} (\partial_\phi\psi)^2+ r^{-2}e^{-1}\, 
J^{N_e}_\mu(\psi)n^\mu_{\Sigma^+} 
 \\
&\le   B\tau_{\rm step}^{-5}e^{-1} {\omega_0}^{-8}\sup_{\tau-1 \le \bar\tau\le \tau}
\int_{\Sigma^{+}(\bar\tau)}  J^{N_e}_\mu(\psi) n_{\Sigma^{+}}^\mu,
\end{align*}
where we have used $(\ref{ilerde})$.
On the other hand, by conservation of energy we have that
\[
\sup_{\tau-1 \le \bar\tau\le \tau}
\int_{\Sigma^{+}(\bar\tau)}  J^{N_e}_\mu(\psi) n_{\Sigma^{+}}^\mu
\le 2 \sup_{\tau -1\le \bar\tau \le \tau}
\int_{\Sigma(\bar\tau)}   J^{N_e}_\mu(\psi) n_{\Sigma}^\mu,
\]
and thus,
\begin{eqnarray}
\label{ara}
\nonumber
\int_{\mathcal{R}(\tau',\tau'')}    (1+(\tau-t^*))^3 ({}^2F^\tau_{\flat})^2 \,dt^*
&\le&  
   B\tau_{\rm step}^{-5}e^{-1}{\omega_0}^{-8}   \sup_{\tau-1 \le \bar\tau\le \tau}
\int_{\Sigma(\bar\tau)} J^{N_e}_\mu(\psi) n_{\Sigma}^\mu\\
\nonumber
&\le&     B\tau_{\rm step}^{-5}e^{-1}  {\omega_0}^{-8} \sup_{\tau-1 \le \bar\tau\le \tau}
\int_{\Sigma(\bar\tau)} {\bf q}_e(\psi)\\
&\le&      B\tau_{\rm step}^{-5}e^{-1} {\omega_0}^{-8} \sup_{0 \le \bar\tau\le \tau}
\int_{\Sigma(\bar\tau)} {\bf q}_e(\psi).
\end{eqnarray}
Clearly, an identical bound holds for
\[
\int_{\mathcal{R}(\tau',\tau'')}    (1+t^*)^3 ({}^1F^\tau_{\flat})^2 \,dt^*.
\]

Let us consider
first the cases where $V\ne{\bf X}$. For $V=T, N_e$ we have
\begin{eqnarray*}
\int_{\mathcal{R}(\tau',\tau'')}\mathcal{E}^V(\psi^\tau_{\flat})
&=& 
\int_{\mathcal{R}(\tau',\tau'')} {}^1F^\tau_{\flat} V^\nu(\psi^\tau_{\flat})_\nu+
{}^2F^\tau_{\flat} V^\nu(\psi^\tau_{\flat})_\nu\\
&\le& \int_{\mathcal{R}(\tau',\tau'')} (1+t^*)^3 \,({}^1F^\tau_{\flat})^2
+\int_{\mathcal{R}(\tau',\tau'')} (1+t^*)^{-3}
(V^\nu(\psi^\tau_{\flat})_\nu)^2\\
&&\hbox{}+\int_{\mathcal{R}(\tau',\tau'')}   (1+(\tau-t^*))^3 ({}^2F^\tau_{\flat})^2\\
&&\hbox{}
+\int_{\mathcal{R}(\tau',\tau'')} (1+(\tau-t^*))^{-3}(V^\nu(\psi^\tau_{\flat})_\nu)^2 \\
&\le&  B\tau_{\rm step}^{-5}e^{-1}{\omega_0}^{-8}  \sup_{0 \le \bar\tau\le \tau}
\int_{\Sigma(\bar\tau)}  {\bf q}_e(\psi)
+B \int_{\mathcal{R}(\tau',\tau'')} (1+t^*)^{-3}
\, {\bf q}_e(\psi^\tau_{\flat})\\
&&\hbox{}
+B \int_{\mathcal{R}(\tau',\tau'')}    (1+(\tau-t^*))^{-3}\,  {\bf q}_e(\psi^\tau_{\flat})\\
&\le&  B\tau_{\rm step}^{-5}e^{-1} {\omega_0}^{-8} \sup_{0 \le \bar\tau\le \tau}
\int_{\Sigma(\bar\tau)}  {\bf q}_e(\psi)\\
&&\hbox{}
+B\tau_{\rm step}^{-2}  \sup_{\tau_{\rm step} \le \bar\tau\le \tau-\tau_{\rm step}}\int_{\bar\tau-1}^{\bar\tau}
\left(\int_{\Sigma(\bar\tau)}  {\bf q}_e(\psi^\tau_{\flat})\right)d\tilde\tau \\
&\le& (B\tau_{\rm step}^{-5}e^{-1}{\omega_0}^{-8}+
\tau_{\rm step}^{-2}B(\omega_0^{-1}+\omega_0^{-7}e^{-1}\tau_{\rm step}^{-2})) 
 \sup_{0 \le \bar\tau\le \tau}
\int_{\Sigma(\bar\tau)}  {\bf q}_e(\psi)
\end{eqnarray*}
where for the last inequality we have used Proposition~\ref{yenicomp}.
We argue similarly for $\mathcal{E}^V(\psi^\tau_{\sharp})$.

For the case of $\mathcal{E}^V$ where $V={\bf X}$, we have
an additional error term
$$
\tilde {\mathcal E}^X(\psi^\tau_\flat)=-\frac14\left(2f'_b+4\frac{1-\mu}rf_b
-\frac{4M(1-\mu)f_b}{r^2}\right)\psi^\tau_\flat F^\tau_\flat.
$$
Recall that $|f_b|\le B\chi $, and $|f'_b|\le Br^{-2}\chi$, where $\chi$ is a cutoff function
such that $\chi=0$ in $r^*\le 0$.
Arguing as in the previous bound we obtain
\begin{eqnarray*}
\int_{\mathcal{R}(\tau',\tau'')}\tilde{\mathcal{E}}^X(\psi^\tau_{\flat})
&\le& \int_{\mathcal{R}(\tau',\tau'')} (1+t^*)^3 \,({}^1F^\tau_{\flat})^2
+B\int_{\mathcal{R}(\tau',\tau'')} (1+t^*)^{-3}
\chi^2 r^{-2} (\psi^\tau_{\flat})^2\\
&&\hbox{}+\int_{\mathcal{R}(\tau',\tau'')}   (1+(\tau-t^*))^3 ({}^2F^\tau_{\flat})^2\\
&&\hbox{}
+B\int_{\mathcal{R}(\tau',\tau'')} (1+(\tau-t^*))^{-3} \chi^2 r^{-2} (\psi^\tau_{\flat})^2 \\
&\le&  B\tau_{\rm step}^{-5}e^{-1} {\omega_0}^{-8} \sup_{0 \le \bar\tau\le \tau}
\int_{\Sigma(\bar\tau)}  {\bf q}_e(\psi)
+\int_{\mathcal{R}(\tau',\tau'')} (1+t^*)^{-3}
\, {\bf q}_e(\psi^\tau_{\flat})\\
&&\hbox{}
+\int_{\mathcal{R}(\tau',\tau'')}    (1+(\tau-t^*))^{-3}\,  {\bf q}_e(\psi^\tau_{\flat})\\
&\le&  B\tau_{\rm step}^{-5}e^{-1} {\omega_0}^{-8} \sup_{0 \le \bar\tau\le \tau}
\int_{\Sigma(\bar\tau)}  {\bf q}_e(\psi)\\
&&\hbox{}
+B\tau_{\rm step}^{-2}  \sup_{\tau_{\rm step} \le \bar\tau\le \tau-\tau_{\rm step}}\int_{\bar\tau-1}^{\bar\tau}
\left(\int_{\Sigma(\bar\tau)}  {\bf q}_e(\psi^\tau_{\flat})\right)d\tilde\tau\\
&\le&  (B\tau_{\rm step}^{-5}e^{-1}{\omega_0}^{-8}+\tau_{\rm step}^{-2}B(\omega_0^{-1}+
\omega_0^{-7}e^{-1}
\tau_{\rm step}^{-2})) 
 \sup_{0 \le \bar\tau\le \tau}
\int_{\Sigma(\bar\tau)}  {\bf q}_e(\psi).
\end{eqnarray*}
In the above, we have used again Proposition~\ref{yenicomp} as well as
the inequality
\[
r^{-2}\int_0^{2\pi} (\psi^\tau_{\flat})^2 d\phi \le
r^{-2}\int_0^{2\pi}(\partial_\phi \psi^\tau_{\flat})^2 d\phi \le {\bf q}_e(\psi^\tau_{\flat})
\]
in the support of $\chi$.
The other terms of $\mathcal{E}^{\bf X}$ can be handled as in the argument
for $\mathcal{E}^T$, $\mathcal{E}^{N_e}$. 
Again, the argument for $\psi^\tau_{\sharp}$
is
identical.
\end{proof}

\subsection{Revisiting the relation between ${\bf q}_e(\psi_{\flat}^\tau)$,  
${\bf q}_e(\psi_{\sharp}^\tau)$ and ${\bf q}_e(\psi)$}
\label{revisit}
With the Proposition of the previous section, we may now refine 
Proposition~\ref{yenicomp} to a pointwise-in-time bound:
\begin{proposition}
\label{yenicomp'}
Let $\omega_0\le1\le\tau_{\rm step}\le \tau' \le \tau-\tau_{\rm step}$.
 Then
\begin{eqnarray*}
 \int_{\Sigma(\tau')} {\bf q}_e(\psi_{\flat}^\tau)&\le& 
B\sup_{\tau'-\tau_{\rm step} \le \bar\tau\le \tau'+\tau_{\rm step}} 
\int_{\Sigma(\bar\tau)} {\bf q}_e(\psi)\\
&&\hbox{}+B{\omega_0}^{-8}e^{-1}
\tau_{\rm step}^{-2} \sup_{0\le\bar \tau\le \tau}\int_{\Sigma(\bar\tau)} {\bf q}_e(\psi),
\end{eqnarray*}
\begin{eqnarray*}
\int_{\Sigma(\tau')} {\bf q}_e(\psi_{\sharp}^\tau)&\le& 
B\sup_{\tau'-\tau_{\rm step} \le \bar\tau\le \tau'+\tau_{\rm step}} 
\int_{\Sigma(\bar\tau)}{\bf q}_e(\psi)\\
&&\hbox{}
+B{\omega_0}^{-8}
e^{-1}\tau_{\rm step}^{-2} \sup_{0\le\bar\tau\le \tau}\int_{\Sigma(\bar\tau)} {\bf q}_e(\psi).
\end{eqnarray*}
\end{proposition}
\begin{proof}
Once again it is sufficient to establish this for $\psi_\flat^\tau$.

We write the energy identity $(\ref{basicid})$
for the vector field $N_e$ to obtain
\begin{align*}
\int_{\mathcal{H}(\tau_0,\tau')} J^{N_e}_\mu(\psi^\tau_\flat) n^\mu_{\mathcal{H}}
+\int_{\Sigma(\tau')} J^{N_e}_\mu(\psi^\tau_\flat) n_{\Sigma}^\mu+
\int_{\mathcal{R}(\tau_0,\tau')}
 K^{N_e}(\psi^\tau_\flat)\\
 =\int_{\mathcal{R}(\tau_0,\tau')} {\mathcal E} ^{N_e}(\psi^\tau_\flat)
+\int_{\Sigma(\tau_0)} J^{N_e}_\mu(\psi^\tau_\flat) n_\Sigma^\mu.
\end{align*}
By $(\ref{muhim})$, $(\ref{muhimb})$, and the nonnegativity of the first term
on the left hand side above, we obtain
\[
\int_{\Sigma(\tau')}  {\bf q}_e(\psi^\tau_\flat) \le
e B\int_{\tau_0}^{\tau'}\int_{\Sigma(t^*)}  {\bf q}_e(\psi^\tau_\flat) dt^* +
B\left|\int_{\mathcal{R}(\tau_0,\tau')} {\mathcal E} ^{N_e}(\psi^\tau_\flat)\right| +
B\int_{\Sigma(\tau_0)}  {\bf q}_e(\psi^\tau_\flat).
\]
We integrate the above inequality with respect to $\tau_0$ between $\tau'-1$ and $\tau'$ 
and use Propositions \ref{yenicomp}, \ref{exoume} to obtain the desired estimate.
\end{proof}

\section{The main estimates}
\subsection{Estimates for $\psi^\tau_{\flat}$}
\label{flest}

Let us assume always
\begin{equation}
\label{upo9et}
\tau_{\rm step}\le \tau' \le \tau''\le \tau-\tau_{\rm step}.
\end{equation}

\begin{proposition}
\label{lemma1}
For $\psi_{\flat}^\tau$ we have
\begin{eqnarray*}
\int_{\tau'}^{\tau''}\left(\int_{\Sigma(\bar\tau)} 
 {\bf q}_e^\bigstar (\psi_{\flat}^\tau)\right)d\bar\tau 
&\le& B\int_{\mathcal{R}(\tau'.\tau'')} \left( K^{\bf X}+ K^{N_e}\right)(\psi^\tau_{\flat})\\
&&\hbox{}+ B
\sup_{0\le \bar\tau \le \tau} \int_{\Sigma(\bar\tau)}  {\bf q}_e (\psi).
\end{eqnarray*}
\end{proposition}
\begin{proof}
Recall that $\int_0^{2\pi} \psi^\tau_{\flat}d\phi=0$.

In the region
$r\le r^-_Y$, we have immediately from $(\ref{refto3})$ that
\[
\int_0^{2\pi}   {\bf q}_e^\bigstar(\psi^\tau_{\flat})\, d\phi
 \le B\int_0^{2\pi} (  K^{\bf X}+ K^{N_e})(\psi^\tau_{\flat}) d\phi.
\]
Similarly, in the region $r\ge R$, we have
from $(\ref{refto1})$ that
\[
\int_0^{2\pi}   {\bf q}_e^\bigstar(\psi^\tau_{\flat})
d\phi  \le B\int_0^{2\pi} (  K^{\bf X}+ K^{N_e})(\psi^\tau_{\flat}) \, d\phi.
\]

For $r^-_Y\le r\le R$, we have from $(\ref{refto2})$ that
\begin{eqnarray*}
\int_0^{2\pi}{\bf q}_e^\bigstar (\psi^\tau_{\flat})
\,d\phi &\le & B\int_0^{2\pi}(  K^{\bf X}+ K^{N_e})(\psi^\tau_{\flat}) \,d\phi\\
&&\hbox{}
- \int_0^{2\pi} \left(b|\nabb\psi_{\flat}^\tau|^2-
B(\partial_t\psi_{\flat}^\tau)^2\right)\,d\phi
\end{eqnarray*}
Note also that
\[
\int_0^{2\pi}|\nabb\psi^\tau_{\flat}|^2\,d\phi
\ge b \int_0^{2\pi}|\partial_\phi \psi^\tau_{\flat}|^2\,d\phi 
\]
for constant $(r,\theta, t)$ curves in the region $r^-_Y\le r\le R$.
We have thus
\begin{eqnarray*}
\int_0^{2\pi}{\bf q}_e^\bigstar (\psi^\tau_{\flat})
\,d\phi &\le & B\int_0^{2\pi}(  K^{\bf X}+ K^{N_e})(\psi^\tau_{\flat}) \,d\phi\\
&&\hbox{}
- \int_0^{2\pi} \left(b(\partial_\phi\psi_{\flat}^\tau)^2-
B(\partial_t\psi_{\flat}^\tau)^2\right)\,d\phi
\end{eqnarray*}
The Proposition follows now from Proposition~\ref{compderpro2}
for ${\omega_0}$ chosen appropriately, in view also 
of our remarks on the measure of integration.
\end{proof}

{\bf \emph{In what follows we shall consider $\omega_0$ to have been chosen
and absorb such factors into the constants $B$.}}

\begin{proposition}
\label{withx}
For $\psi^\tau_{\flat}$, we have
\begin{eqnarray*}
\int_{\tau'}^{\tau''}\left(\int_{\Sigma(\bar\tau)}  {\bf q}_e^\bigstar(\psi_{\flat}^\tau)\right)
d\bar\tau
&\le&  B\left( \int_{\Sigma(\tau')} 
 J^{N_e}_\mu(\psi^\tau_{\flat}) n^\mu_\Sigma
+\int_{\Sigma(\tau'')} J^{N_e}_\mu(\psi^\tau_{\flat}) n^\mu_\Sigma\right.\\
&&\hbox{}\left. +\int_{\mathcal{H}(\tau',\tau'')} 
J^{N_e}_\mu(\psi^\tau_{\flat}) n^\mu_{\mathcal{H}}
\right)
\\
&&\hbox{}+ B
\sup_{0\le \bar\tau \le \tau} \int_{\Sigma(\bar\tau)}  {\bf q}_e (\psi).
\end{eqnarray*}
\end{proposition}
\begin{proof}
To prove Proposition~\ref{withx} from Proposition~\ref{lemma1},
note that from $(\ref{basicid})$ applied to the current $ J^{\bf X}+ 
 J^{N_e}$
we have
\begin{eqnarray*}
\int_{\mathcal{R}(\tau',\tau'')} ( K^{{\bf X}}+ K^{N_e})(\psi^\tau_{\flat}) &\le&
\left|\int_{\Sigma(\tau')}( J^{{\bf X}}_\mu+
 J^{N_e}_\mu) (\psi^\tau_{\flat}) n^\mu_\Sigma
 \right|\\
 &&\hbox{} + \left|\int_{\Sigma(\tau'')}( J^{{\bf X}}_\mu+
  J^{N_e}_\mu) (\psi^\tau_{\flat}) n^\mu_{\Sigma}
 \right.\\
 &&\hbox{}\left.+ 
\int_{\mathcal{H}(\tau',\tau'')}( J^{{\bf X}}_\mu+
 J^{N_e}_\mu) (\psi^\tau_{\flat}) n^\mu_{\mathcal{H}}
 \right| \\
&&\hbox{}
+\left| \int_{\mathcal{R}(\tau',\tau'')} \mathcal{E}^{{\bf X}+N_e}(\psi^\tau_{\flat})\right|\\
&\le&B\left( \int_{\Sigma(\tau')} 
 J^{N_e}_\mu(\psi^\tau_{\flat}) n^\mu_\Sigma
+\int_{\Sigma(\tau'')}  J^{N_e}_\mu(\psi^\tau_{\flat}) n^\mu_{\Sigma} \right.\\
&&\hbox{}\left.+\int_{\mathcal{H}(\tau',\tau'')} 
 J^{N_e}_\mu(\psi^\tau_{\flat}) n^\mu_{\mathcal{H}}
\right)\\
&&\hbox{}+B\tau_{\rm step}^{-2}e^{-1}
\sup_{0\le \bar\tau \le \tau} \int_{\Sigma(\bar\tau)}  {\bf q}_e (\psi).
\end{eqnarray*}
Above we have used $(\ref{yeniyIldIz})$ and Proposition~\ref{exoume}.
It will be important for later
that $eB\ll1$.
The Proposition now follows immediately. 
\end{proof}

\begin{proposition}
\label{notquiteconslow0}
\begin{align*}
&\int_{\mathcal{H}(\tau',\tau'')}  J^T_\mu (\psi^\tau_{\flat}) n^\mu_{\mathcal{H}}+
\int_{\Sigma(\tau'')}  J^T_\mu (\psi^\tau_{\flat}) n^\mu_\Sigma\\
&\le B  \int_{\Sigma(\tau')}  {\bf q}_e(\psi_{\flat})
+ B\tau_{\rm step}^{-2}e^{-1} \sup_{0\le \bar\tau\le \tau} \int_{\Sigma(\bar\tau)}
 {\bf q}_e(\psi).
\end{align*}
\end{proposition}
\begin{proof}
This follows from
the divergence
identity $(\ref{basicid})$ for the current $J^T(\psi^\tau_{\flat})$ 
and the fact that $K^T=0$ and the inequality
\[
\int_{\mathcal{R}(\tau',\tau'')} 
\mathcal{E}^T(\psi^\tau_{\flat}) \le 
B\tau_{\rm step}^{-2}e^{-1}
\sup_{0\le \bar\tau \le \tau} \int_{\Sigma(\bar\tau)}  {\bf q}_e (\psi)
\]
of Proposition~\ref{exoume}.
\end{proof}

\begin{proposition}
\label{prong}
\begin{eqnarray*}
\int_{\mathcal{H}(\tau',\tau'')}  J_{\mu}^{N_e} (\psi_{\flat}^\tau)n^\mu_{\mathcal{H}}  +
\int_{\Sigma(\tau'')}  J_{\mu}^{N_e} (\psi_{\flat}^\tau) n^\mu_\Sigma &\le&
\int_{\Sigma(\tau')}  J_{\mu}^{N_e} (\psi_{\flat}^\tau) n^\mu_\Sigma \\
&&\hbox{}
+ Be \int_{\tau'}^{\tau''}\left( \int_{\Sigma(\bar\tau)}
  {\bf q}_e^\bigstar(\psi_{\flat}^\tau)\right) d\bar\tau \\
&&\hbox{}+B\tau_{\rm step}^{-2}e^{-1} \sup_{0\le \bar\tau \le \tau} \int_{\Sigma(\bar\tau)} 
 {\bf q}_e (\psi).
\end{eqnarray*}
\end{proposition}
\begin{proof}
This follows just from the divergence identity $(\ref{basicid})$ for $J^{N_e}$ together
with the bounds $(\ref{muhim})$ and  $(\ref{xamla9})$.
\end{proof}

\begin{proposition}
\label{forlow}
\begin{align*}
& \int_{\Sigma(\tau'')}  {\bf q}_e(\psi_{\flat}^\tau)
 +
 \int_{\tau'}^{\tau''} \left(\int_{\Sigma(\bar\tau)}  {\bf q}_e^\bigstar(\psi_{\flat}^\tau)\right) d\bar\tau\\
&\le B \int_{\Sigma(\tau')}  {\bf q}_e (\psi^\tau_{\flat})+
 B  \sup_{0\le \bar\tau \le \tau} \int_{\Sigma(\bar\tau)}  {\bf q}_e (\psi).
\end{align*}
\end{proposition}
\begin{proof}
This follows immediately from Propositions~\ref{withx} and~\ref{prong} in view of
$(\ref{muhimb})$ and the fact that
for $e$ small we have $Be\ll1$.
\end{proof}

\begin{proposition}
\label{notquiteconslow}
\begin{align*}
&
\int_{\Sigma(\tau'')}  J^T_\mu (\psi^\tau_{\flat}) n^\mu_\Sigma\\
&\le B  \int_{\Sigma(\tau')} {\bf q}_e(\psi^\tau_{\flat})
+ (B\tau_{\rm step}^{-2}e^{-1}+B \epsilon_{\rm close}e^{-1})
 \sup_{0\le \bar\tau\le \tau} \int_{\Sigma(\bar\tau)} {\bf q}_e(\psi).
\end{align*}
\end{proposition}
\begin{proof}
This follows from Propositions~\ref{notquiteconslow0},~\ref{prong}
and~\ref{forlow} together with the one-sided bound
\[
-\int_{\mathcal{H}(\tau',\tau'')}  J^T_\mu (\psi^\tau_{\flat}) n^\mu_{\mathcal{H}}
\le B\epsilon_{\rm close}e^{-1}
\int_{\mathcal{H}(\tau',\tau'')}  J^{N_e}_\mu (\psi^\tau_{\flat})n^\mu_{\mathcal{H}}.
\]
\end{proof}

\subsection{Estimates for $\psi^\tau_{\sharp}$}
\label{estest}
We assume always $(\ref{upo9et})$.

\begin{proposition}
\label{notquitecons}
For $\psi^\tau_{\sharp}$,
\begin{align*}
\frac12\int_{\mathcal{H}(\tau',\tau'')} \left|J^T_\mu(\psi^\tau_{\sharp})n^\mu_{\mathcal{H}}\right|+
\int_{\Sigma(\tau'')}  J^T_\mu (\psi^\tau_{\sharp}) n^\mu
\le \int_{\Sigma(\tau')}  J^T_\mu (\psi^\tau_{\sharp}) n^\mu\\
+ B\tau_{\rm step}^{-2}\epsilon^{-1}
\sup_{0\le \bar\tau\le \tau} \int_{\Sigma(\bar\tau)} {\bf q}_e(\psi)
+ B\epsilon_{\rm close} e^{-1}\sup_{0\le \bar\tau\le  \tau}
\int _{\mathcal{H}(\bar\tau,\bar\tau+1) } J_\mu^{N_e}(\psi)n^\mu_{\mathcal{H}}.
\end{align*}
\end{proposition}
\begin{proof}
From $(\ref{basicid})$ 
applied to $\psi_{\sharp}^\tau$ with $V=T$ we have
\begin{align*}
&\int_{\Sigma(\tau'')} J^T_\mu (\psi^\tau_{\sharp}) n^\mu_\Sigma\\
=&\int_{\Sigma(\tau')}  J^T_\mu (\psi^\tau_{\sharp}) n^\mu_\Sigma
-\int_{\mathcal{H}(\tau',\tau'')}  J^T_\mu (\psi^\tau_{\sharp}) n^\mu_{\mathcal{H}}\\
&+ \int_{\mathcal{R}(\tau',\tau'')}\mathcal{E}^T(\psi^\tau_{\sharp}).
\end{align*}
On the other hand, by $(\ref{planeeq})$, we have
the one-sided bound
\begin{eqnarray*}
-\int_{\mathcal{H}(\tau',\tau'')}  J^T_\mu(\psi^\tau_{\sharp}) n^\mu_{\mathcal{H}}
&\le&	B\epsilon_{\rm close}
\int_{\mathcal{H}(\tau',\tau'')} \partial_t\psi^\tau_{\sharp}\partial_\phi\psi^\tau_{\sharp}
\\
&&\hbox{}
-b\int_{\mathcal{H}(\tau',\tau'')}(\partial_t \psi^\tau_{\sharp})^2 \\
&\le&B \epsilon_{\rm close} \int_{\mathcal{H}(\tau',\tau'')}(\partial_\phi \psi^\tau_{\sharp})^2
\\
&&\hbox{}
-b\int_{\mathcal{H}(\tau',\tau'')}(\partial_t \psi^\tau_{\sharp})^2 
\end{eqnarray*}
and thus by Proposition~\ref{compderpro} 
we have
\begin{eqnarray*}
-\int_{\mathcal{H}(\tau',\tau'')}  J^T_\mu(\psi^\tau_{\sharp}) n^\mu_{\mathcal{H}}
\le B\epsilon_{\rm close} e^{-1} \sup_{0\le \bar\tau\le  \tau}
\int _{\mathcal{H}(\bar\tau,\bar\tau+1)}  J_\mu^{N_e}(\psi)n^\mu_{\mathcal{H}}.
\end{eqnarray*}
The desired result now follows from Proposition~\ref{exoume}.
\end{proof}

\begin{proposition}
\label{thisisthe}
\begin{align*}
&\int_{\Sigma(\tau'')}  J^{N_e}_\mu(\psi^\tau_{\sharp})n^\mu_{\Sigma}
+\int_{\mathcal{R}(\tau',\tau'')} K^{N_e}(\psi^\tau_{\sharp})\\
&\le \int_{\Sigma(\tau')}  J^{N_e}_\mu(\psi^\tau_{\sharp}) n^\mu_{\Sigma}
+B\tau^{-2}_{\rm step} e^{-1}\sup_{0\le \bar\tau\le \tau} \int_{\Sigma(\bar\tau)}
 {\bf q}_e(\psi).
\end{align*}
\end{proposition}
\begin{proof}
This is the energy identity $(\ref{basicid})$ for $N_e$ in view
of the nonnegativity of the flux on the horizon and
the estimate $(\ref{ynlla9})$. 
\end{proof}

\begin{proposition} 
\label{forhigh}
\begin{align*}
& b\int_{\Sigma(\tau'')}  {\bf q}_e(\psi_{\sharp}^\tau)
 +
b \int_{\tau'}^{\tau''}\left( \int_{\Sigma(\bar\tau)} 
 {\bf q}_{e}(\psi_{\sharp}^\tau) \right)d\bar\tau\\
&\le (\tau''-\tau') \int_{\Sigma(\tau_{\rm step})} 
 J^T_\mu (\psi^\tau_{\sharp}) n^\mu_{\Sigma}
+B\int_{\Sigma(\tau')}
 {\bf q}_e(\psi^\tau_{\sharp})\\
&+
(\tau''-\tau'+1) \left( B \tau_{\rm step}^{-2}e^{-1}
\sup_{0\le \bar\tau \le \tau} \int_{\Sigma(\bar\tau)}  {\bf q}_e (\psi)
+B\epsilon_{\rm close} e^{-1}\sup_{0\le \bar\tau\le  \tau}
\int _{\mathcal{H}(\bar\tau,\bar\tau+1) }
 J_\mu^{N_e}(\psi)n^\mu_{\mathcal{H}}\right).
\end{align*}
\end{proposition}
\begin{proof}
The proof follows from Propositions~\ref{notquitecons} (applied
with $\tau'=\tau_{\rm step}$ and $\tau''=\bar\tau$),
Proposition~\ref{thisisthe} applied to the given $\tau'$ and $\tau''$,  $(\ref{muh2})$ 
and $(\ref{muhimb})$.
\end{proof}

\section{The bootstrap}
Let $C$ be given, and consider the set $\mathcal{T}\subset[0,\infty)$ of
all $\tau$ such that for $0\le\bar\tau\le \tau$, we have
\begin{equation}
\label{bstrap}
\int_{\Sigma(\bar\tau)} {\bf q}_e(\psi) \le C \int_{\Sigma(0)}{\bf q}_e(\psi).
\end{equation}
Theorem~\ref{kevtriko} would follow from 
\begin{proposition}
\label{mainprop}
For suitable choice of $C$, then $\mathcal{T}=[0,\infty)$,
i.e.~$(\ref{bstrap})$ holds for all $\tau\ge 0$.
\end{proposition}
For this it suffices to show that $\mathcal{T}$ is non-empty, open
and closed. The non-emptyness is clear for sufficiently large $C$.
It thus suffices to show that $C$ can be chosen such that for all
$\tau\in \mathcal{T}$, then
\begin{equation}
\label{impbstrap}
\int_{\Sigma(\bar\tau)} {\bf q}_e(\psi) \le \frac{C}2 \int_{\Sigma(0)} {\bf q}_e(\psi)
\end{equation}
for $0\le\bar\tau\le \tau$.

\subsection{Evolution for time $\tau_{\rm step}$}
We will need the following proposition
\begin{proposition}
\label{willneed}
Let $\tau_{\rm step}$ be given. 
For small enough $e$ depending on $\tau_{\rm step}$, $\epsilon_{\rm close}\ll e$,
it follows that for all $\tau_0$ and $\bar\tau\in[\tau_0,\tau_0+\tau_{\rm step}]$,
\begin{equation}
\label{Cstab}
\int_{\Sigma(\bar\tau)}
{\bf q}_e(\psi) \le 2 \int_{\Sigma(\tau_{0})} {\bf q}_e(\psi).
\end{equation}
\end{proposition}
\begin{proof}
We write the energy identity $(\ref{basicid})$
for the vector field $N_e$ to obtain
\begin{align*}
\int_{\mathcal{H}(\tau_0,\bar\tau)} J^{N_e}_\mu(\psi) n^\mu_{\mathcal{H}}
+\int_{\Sigma(\bar\tau)} J^{N_e}_\mu(\psi) n_{\Sigma}^\mu+
\int_{\mathcal{R}(\tau_0,\bar\tau)}
 K^{N_e}(\psi) \\
=\int_{\Sigma(\tau_0)} J^{N_e}_\mu(\psi) n_\Sigma^\mu.
\end{align*}
By $(\ref{muhim})$, $(\ref{muhimb})$, and the nonnegativity of the first term
on the left hand side above, we obtain
\[
\int_{\Sigma(\bar\tau)}  {\bf q}_e(\psi) \le
e B\int_{\tau_0}^{\bar\tau}\int_{\Sigma(\hat\tau)}  {\bf q}_e(\psi) d\hat\tau +
\int_{\Sigma(\tau_0)}  {\bf q}_e(\psi)
\]
and thus
\[
\int_{\Sigma(\bar\tau)}  {\bf q}_e(\psi) \le
\exp(eB(\bar\tau-\tau_0)) \int_{\Sigma(\tau_0)} {\bf q}_e(\psi).
\]
The result follows thus if $e$ is chosen
so that
\[
\exp(eB\tau_{\rm step}) \le 2.
\] 
\end{proof}

\subsection{Estimate for the local horizon flux of $ J^{N_e}_\mu(\psi)$}
A corollary of the proof of the previous Proposition is the following
\begin{proposition}
\label{horprop0}
\[
\int_{\mathcal{H}(\bar\tau,\bar\tau+1)}  J_\mu^{N_e}(\psi) n^\mu \le
B\int_{\Sigma(\bar\tau)} {\bf q}_e(\psi).
\]
\end{proposition}
Of course, if we choose $e$ to be sufficiently small as in the previous
Proposition, we may replace $B$ with $2$.

\subsection{Bounds for $\psi^\tau_{\flat}$}
From Proposition~\ref{forlow} applied for $\tau'=n\tau_{\rm step}$,
$\tau''=(n+1)\tau_{\rm step}$, $n=1,2,\ldots, n_f$ where $n_f$ is the
largest integer such that $(n_f+1)\tau_{\rm step}\le \tau-\tau_{\rm step}$,
Proposition~\ref{yenicomp'} and
the bootstrap assumption $(\ref{bstrap})$, 
it follows
that in each interval $[n\tau_{\rm step},(n+1)\tau_{\rm step}]$,
we can find a $\tau_n$ such that
\begin{eqnarray*}
\int_{\Sigma(\tau_n)}  {\bf q}_e^\bigstar (\psi^\tau_{\flat}) 
&\le& \frac{B}{\tau_{\rm step}}\sup_{0\le \bar\tau\le \tau}\int_{\Sigma(\bar\tau)} {\bf q}_e(\psi)
+\frac{B}{\tau_{\rm step}}\int_{\Sigma(n\tau_{\rm step})} {\bf q}_e(\psi^\tau_{\flat})\\
&\le& B\tau_{\rm step}^{-1} C \int_{\Sigma(0)} {\bf q}_e(\psi)
\end{eqnarray*}
for appropriate choice of $\tau_{\rm step}$.
On the other hand, by Proposition~\ref{notquiteconslow} applied with
$\tau'=\tau_{\rm step}$, $\tau''=\tau_n$,
Proposition~\ref{yenicomp'}, 
$(\ref{Cstab})$ applied to $\tau_0=0$ and again to $\tau_0=\tau_{\rm step}$,
and the bootstrap assumption $(\ref{bstrap})$,
we have
\begin{eqnarray}
\label{lowteliko}
\nonumber
\int_{\Sigma(\tau_n)} J_\mu^T(\psi^\tau_{\flat}) n^\mu_\Sigma
&\le& B  \int_{\Sigma(\tau_{\rm step})} {\bf q}_e(\psi^\tau_{\flat})\\
\nonumber
&&\hbox{}
+(B\tau^{-2}_{\rm step}e^{-1}+B\epsilon_{\rm close}e^{-1})
\sup_{0\le \bar\tau\le \tau}\int_{\Sigma(\bar\tau)} {\bf q}_e(\psi)\\
\nonumber
&\le&B\sup_{0\le\bar\tau\le2\tau_{\rm step}}
\int_{\Sigma(\bar\tau)} {\bf q}_e(\psi)\\
\nonumber
&&\hbox{}
+(B\tau^{-2}_{\rm step}e^{-1}+B\epsilon_{\rm close}e^{-1})
\sup_{0\le \bar\tau\le \tau}\int_{\Sigma(\bar\tau)} {\bf q}_e(\psi)\\
\nonumber
&\le&(B+B\tau_{\rm step}^{-2}e^{-1}C+ B\epsilon_{\rm close}e^{-1}C)\\
&&\hbox{}
\cdot\int_{\Sigma(0)} {\bf q}_e(\psi).
\end{eqnarray}
It follows that
\begin{eqnarray}
\label{edwavafora}
\nonumber
\int_{\Sigma(\tau_n)} {\bf q}_e(\psi^\tau_{\flat}) 
&\le &
B\int_{\Sigma(\tau_n)}  {\bf q}_e^\bigstar (\psi^\tau_{\flat}) +
\int_{\Sigma(\tau_n)} J_\mu^T(\psi^\tau_{\flat}) n^\mu_{\Sigma}\\
\nonumber
&\le&
(B+B\epsilon_{\rm close}e^{-1}C +
B\tau_{\rm step}^{-2}e^{-1}C+B\tau_{\rm step}^{-1}C)\\
&&\hbox{}\cdot \int_{\Sigma(0)} {\bf q}_e(\psi).
\end{eqnarray}

\subsection{Bounds for $\psi^\tau_{\sharp}$}
\label{boubou}
Since 
\[
\int_{\Sigma(\bar\tau)}  J^T_\mu (\psi^\tau_{\sharp}) \le
B\int_{\Sigma(\bar\tau)} {\bf q}_e(\psi^\tau_{\sharp}), 
\]
it follows from Proposition~\ref{forhigh}
applied to $\tau'= n\tau_{\rm step}$, $\tau''=(n+1)\tau_{\rm step}$, $n=1,2,\ldots, n_f$, 
Proposition~\ref{horprop0} and $(\ref{Cstab})$ applied (twice) with $\tau_0=0$ and
$\tau_0=\tau_{\rm step}$, and Proposition~\ref{yenicomp'}
that in each interval  $[n\tau_{\rm step},(n+1)\tau_{\rm step}]$, 
we can find an $\tau_n$ such that
\begin{eqnarray}
\label{veohighteliko}
\nonumber
b \int_{\Sigma(\tau_n)} {\bf q}_{e}(\psi^\tau_{\sharp})&\le&  \int_{\Sigma(\tau_{\rm
step})} J_\mu^T(\psi^\tau_{\sharp})n^\mu_{\Sigma}
+\tau_{\rm step}^{-1}B\int_{\Sigma(n\tau_{\rm step})} {\bf q}_e(\psi^{\tau}_{\sharp})\\
\nonumber
&&\hbox{}
+B\tau^{-2}_{\rm step}e^{-1}C\int_{\Sigma(0)} {\bf q}_e(\psi) \\
\nonumber
&&\hbox{}
+B\epsilon_{\rm close} \sup_{0\le \bar\tau\le \tau-1}
\int_{\mathcal{H}(\bar\tau,\bar\tau+1)} J_\mu^{N_e}(\psi)n^\mu_{\mathcal{H}}\\
\nonumber
&\le& B  \int_{\Sigma(\tau_{\rm
step})} {\bf q}_e(\psi^\tau_{\sharp})+
\tau_{\rm step}^{-1}B\int_{\Sigma(n\tau_{\rm step})} {\bf q}_e(\psi^{\tau}_{\sharp})\\
\nonumber
&&\hbox{}+ (B\tau_{\rm step}^{-2}e^{-1}C+B\epsilon_{\rm close}C)\int_{\Sigma(0)} 
{\bf q}_e(\psi)\\
\nonumber
&\le& B\sup_{0\le \bar\tau\le 2\tau_{\rm step}}
 \int_{\Sigma(\bar\tau)} {\bf q}_e(\psi)\\
 \nonumber
 &&\hbox{} +(\tau_{\rm step}^{-1}BC+\tau_{\rm step}^{-2}e^{-1}BC
 +B\epsilon_{\rm close}C)\int_{\Sigma(0)} {\bf q}_e(\psi)
\\
\nonumber
&\le& (B+ B \tau_{\rm step}^{-1}C+B\tau_{\rm step}^{-2}e^{-1}C +B\epsilon_{\rm close}
C )\\
&&\hbox{}
\cdot\int_{\Sigma(0)} {\bf q}_e(\psi).
\end{eqnarray}

\subsection{Bounds for $\psi$}
Choosing $C$ sufficiently large, $\tau_{\rm step}$ sufficiently large, 
$e$ sufficiently small so that Proposition~\ref{willneed} holds,
and $\epsilon_{\rm close}\ll e$ sufficiently small,
from $(\ref{fusika2})$, $(\ref{edwavafora})$, and $(\ref{veohighteliko})$
it follows that
\[
\int_{\Sigma(\tau_n)}  {\bf q}_{e} (\psi)
\le
\frac{C}8  \int_{\Sigma(0)}  {\bf q}_e(\psi).
\]

From Proposition~\ref{willneed} it follows that 
for $\bar\tau\in [\tau_n,\tau_{n+1}]$
\begin{eqnarray*}
\int_{\Sigma(\bar\tau)} {\bf q}_{e}(\psi) &\le&
2 \int_{\Sigma(\tau_n)} {\bf q}_{e}(\psi)\\
&\le& \frac{C}4  \int_{\Sigma(0)} {\bf q}_e(\psi).
\end{eqnarray*}
For $\bar\tau\in[0,\tau_1]$ we 
apply twice Proposition~\ref{willneed}
\begin{eqnarray*}
\int_{\Sigma(\bar\tau)} {\bf q}_{e}(\psi) &\le&
2 \int_{\Sigma(\tau_{\rm step})} {\bf q}_{e}(\psi)\\
&\le&4\int_{\Sigma(0)} {\bf q}_e(\psi)\\
&\le&\frac{C}2\int_{\Sigma(0)} {\bf q}_e(\psi),
\end{eqnarray*}
as long as we assume $C\ge 8$.
Similarly, for $\bar\tau\in[\tau_{n_f},\tau]$, 
we apply Proposition~\ref{willneed} twice to obtain
\begin{eqnarray*}
\int_{\Sigma(\bar\tau)} {\bf q}_{e}(\psi) &\le&
2\int_{\Sigma(\tau-\tau_{\rm step})} {\bf q}_{e}(\psi)\\
&\le&
4\int_{\Sigma(\tau_{n_f})} {\bf q}_{e}(\psi)\\
&\le& \frac{C}2  \int_{\Sigma(0)} {\bf q}_e(\psi).
\end{eqnarray*}
We have shown $(\ref{impbstrap})$.

\section{Estimate for the total horizon and null-infinity flux of $J^T_\mu(\psi)$}
For the rest of this paper, all small quantities can be considered fixed. 
We will use in the following $C$ as a general constant depending on 
the constant of Proposition~\ref{mainprop}.

\begin{proposition}
\label{horprop}
For all $\tau\ge 0$ we have
\begin{equation}
\label{firstavis}
\int_{\mathcal{H}(0,\tau)} \left|J_\mu^T(\psi)n^\mu_{\mathcal{H}}\right|  \le
C\int_{\Sigma(0)} {\bf q}_e(\psi)
\end{equation}
\end{proposition}
\begin{proof}
For $(\ref{firstavis})$,
in view of the relations
\[
J^T_\mu(\psi)\le 
B\left(\left| J^T_\mu(\psi^\tau_{\sharp})\right|+\left| J^T_\mu(\psi^\tau_{\flat})\right|
\right)
\]
valid in $\mathcal{R}(1,\tau-1)$,
and
\begin{equation}
\label{immed}
\left| J^T_\mu(\Psi) n^\mu_{\mathcal{H}}\right| 
\le C\, J^{N_e}_\mu(\Psi) n^\mu_{\mathcal{H}}
\end{equation}
on $\mathcal{H}^+$, it follows from
Proposition~\ref{horprop0} and Proposition~\ref{mainprop},  
that it suffices to show
\begin{equation}
\label{suffices1}
\int_{\mathcal{H}(\tau_{\rm step},\tau-\tau_{\rm step})} 
\left| J_\mu^T(\psi^\tau_{\flat})n^\mu \right| 
\le
C\int_{\Sigma(0)} {\bf q}_e(\psi),
\end{equation}
\begin{equation}
\label{suffices2}
\int_{\mathcal{H}(\tau_{\rm step},\tau- \tau_{\rm step})} 
\left| J_\mu^T(\psi^\tau_{\sharp}) n^\mu \right| \le
C\int_{\Sigma(0)} {\bf q}_e(\psi).
\end{equation}

Inequality $(\ref{suffices1})$
follows immediately from $(\ref{immed})$ applied to $\psi^\tau_{\flat}$,
and 
Propositions~\ref{prong},~\ref{forlow},~\ref{yenicomp'}
and~\ref{mainprop}.

For $(\ref{suffices2})$, 
in view of the bound
\[
\int_{\Sigma(\tau')} J_\mu^T(\psi^\tau_{\sharp}) n^\mu_{\Sigma}
\le C\int_{\Sigma(\tau')} {\bf q}_e(\psi_{\sharp}^\tau)
\]
we need only apply Propositions~\ref{yenicomp'},~\ref{notquitecons},~\ref{horprop0},
and~\ref{mainprop}.
\end{proof}

\begin{proposition}
\begin{equation}
\label{secondavis}
\int_{\mathcal{I}^+}  J_\mu^T(\psi) n^\mu \le
C\int_{\Sigma(0)} {\bf q}_e(\psi).
\end{equation}
\end{proposition}
\begin{proof}
This follows now by the previous and
the statement $K^T=0$. We omit the details concerning
the definition of the left hand side of the inequality.
\end{proof}
The propositions of this section prove in particular $(\ref{stovorizovta})$ and $(\ref{stoapeiro})$.
The complete statement of Theorem~\ref{prwto} is thus proven.

\section{Higher order energies and pointwise bounds}
\label{hes}
A deficiency of previous understanding of boundedness, even in the Schwarzschild
case, is that it relied on commuting the equation with a full basis
of angular momentum operators $\Omega_i$, $i=1\ldots 3$.
In view of the loss of symmetry when passing from Schwarzschild to Kerr, this 
approach is no longer tenable. A much more robust approach to
boundedness is via commutation with $n_{\Sigma}$, or equivalently,
the vector field $\hat{Y}$ to be discussed below. It turns out that the 
dangerous extra terms arising have a good sign. This can be viewed of
as yet another manifestation
of the redshift effect.

In Section~\ref{highord} below, we will first derive 
$L^2$ estimates for higher order energies.
These will rely on certain elliptic estimates derived in Section~\ref{elliptprops}.
Pointwise estimates will then follow in Section~\ref{pwise}
from standard Sobolev inequalities.

\subsection{Higher order energies}
\label{highord}
Let us consider now the quantity
\[
{\bf q}^j(\Psi)\doteq \sum_{i=0}^j J_\mu^{n_\Sigma}(n^i_{\Sigma}(\Psi)) n^\mu_{\Sigma}
\]
where $n^j\Psi$ denotes $n(n(n\ldots \psi))$ where $j$ $n$'s appear.
Under our smoothness assumptions, coupled with our assumptions about the support, we
have that
\[
\int_{\Sigma(\tau)}{\bf q}^j(\psi) <\infty.
\]
We have the following 
\begin{theorem}
For all $j\ge0$, there exist constants $C_j$ depending only on $j$ and $M$ such that
under the assumptions of Theorem~\ref{kevtriko}, then for all $\tau\ge 0$,
\[
\int_{\Sigma(\tau)} {\bf q}^j(\psi) \le C_j \int_{\Sigma(0)} {\bf q}^j(\psi).
\]
\end{theorem}

\begin{proof} We shall give the proof only for the case $j=1$, as this will be sufficient for
deriving pointwise bounds for $\psi$.

Commute $(\ref{waveee})$ with $T$. 
One obtains from $(\ref{newstatement})$ that for $\tau\ge 0$
\begin{equation}
\label{notethat0}
\int_{\Sigma(\tau)} {\bf q}(T\psi)  \le C 
 \int_{\Sigma(0)} {\bf q}(T\psi).
\end{equation}
Note that from $(\ref{newstatement})$ we have for $\tau''>\tau'\ge0$,
\begin{equation}
\label{notethat1}
\int_{\mathcal{R}(\tau',\tau'')}  {\bf q}(\psi)  \le C 
(\tau''-\tau') \int_{\Sigma(0)} {\bf q}(\psi),
\end{equation}
and from $(\ref{notethat0})$,
\begin{equation}
\label{notethat2}
\int_{\mathcal{R}(\tau',\tau'')}  {\bf q}(T\psi)  \le C 
(\tau''-\tau') \int_{\Sigma(0)} {\bf q}(T\psi).
\end{equation}

Now commute $(\ref{waveee})$ with the vector field 
\[
\hat{Y}\doteq \frac{1}{1-\mu}\partial_u
\]
where $\partial_u$ refers to the coordinate system described in Section~\ref{Ysec}.
We obtain
\begin{lemma}
Let $\psi$ satisfy $(\ref{waveee})$. Then we may write
\[
\Box_{g}(\hat{Y}\psi) 
 = \frac{2}r \hat{Y}(\hat{Y}(\psi))-\frac{4}{r}(\hat Y(T\psi))+  P_1\psi
 -\frac{2}{r}P_2\psi +[\hat{Y}, P_2]\psi
\]
where
$P_1$ is the first order operator $P_1\doteq \frac{2}{r^2}(T\psi -\hat{Y}\psi)$,
and
$P_2$ is the second order operator $P_2=\Box_{g_M}-\Box_g$.
\end{lemma}

Now apply the basic identity $(\ref{basicid})$ to $\hat{Y}$ with $Y$. We have that
for $r\le r^-_Y$,
\[
K^Y(\hat Y\psi) \ge b\, {\bf q} (\hat Y\psi)
\]
while for $r\ge r^-_Y$,
\[
K^Y(\hat Y\psi) \le B\, {\bf q} (\hat Y\psi).
\]
On the other hand,
\begin{eqnarray*}
\mathcal{E}^Y(\hat Y\psi) &=& -\left( \frac{2}r  \hat Y(\hat Y(\psi)) -\frac{4}{r}
 (\hat Y(T\psi))+ P_1\psi -\frac{2}rP_2\psi+ [\hat{Y},P_2]\psi \right) 
 Y(\hat Y(\psi))\\
 &\le&
-\left(\frac2r(\hat{Y}-Y) \hat{Y} \psi -\frac{4}{r}
 (\hat Y(T\psi)) +
 P_1\psi-\frac{2}rP_2\psi + [\hat Y,P_2]\psi \right) Y(\hat Y(\psi)) .
\end{eqnarray*}

The following lemmas are immediate:
\begin{lemma}
\[
\int_{\mathcal{R}(\tau',\tau'')} \frac{16}{r^2}(Y(T\psi))^2 \le B \int_{\mathcal{R}(\tau',\tau'')}
 {\bf q}(T\psi),
\]
\[
\int_{\mathcal{R}(\tau',\tau'')} (P_1\psi)^2  \le B \int_{\mathcal{R}(\tau',\tau'')}
 {\bf q}(\psi).
\]
\end{lemma}

\begin{lemma}
\[
\int_{\mathcal{R}(\tau',\tau'')}
4\frac{(P_2\psi)^2}{r^2}+
([P_2,\hat Y]\psi )^2\le B \epsilon_{\rm close}^2
\int_{\mathcal{R}(\tau',\tau'')}  {\bf q}^1(\psi).
\]
\end{lemma}

\begin{lemma}
Given $r_{\hat{Y}}>2M$, we may choose a $\delta_{\hat{Y}}$ (with $\delta_{\hat{Y}}\to0$
as $r_{\hat{Y}}\to 2M$) such that
\[
\int_{\mathcal{R}(\tau',\tau'')}
\frac{4}{r^2}((\hat{Y}-Y)\hat{Y}\psi)^2
\le B \delta_{\hat{Y}}
\int_{\mathcal{R}(\tau',\tau'')\cap \{r\le r_{\hat{Y}}\}}
{\bf q}^1(\psi)+
B \int_{\mathcal{R}(\tau',\tau'')\cap \{r\ge r_{\hat{Y}}\}}
{\bf q}^1(\psi).
\]
\end{lemma}
For convenience, let us require in what follows that $\delta_{\hat{Y}}\le \delta^-_{Y}$.

It follows from $(\ref{basicid})$, the above lemmas
and Cauchy-Schwarz (applied with a small parameter $\lambda$)
that
\begin{eqnarray*}
\int_{\mathcal{R}(\tau',\tau'')\cap \{r\le r^-_Y\}}
{\bf q}(\hat{Y}\psi) 
&\le&B\int_{\mathcal{R}(\tau',\tau'')\cap \{r\ge r^-_Y\}}
{\bf q}(\hat Y\psi)\\
&&\hbox{}+
B\lambda^{-1}\int_{\mathcal{R}(\tau',\tau'') } {\bf q}(T\psi)+{\bf q}(\psi)\\
&&\hbox{}
+B\lambda^{-1} \epsilon_{\rm close}^2 \int_{\mathcal{R}(\tau',\tau'')} {\bf q}^1(\psi)\\
&&\hbox{}+B\lambda^{-1}  \delta_{\hat{Y}}
\int_{\mathcal{R}(\tau',\tau'')\cap \{r\le r_{\hat{Y}}\}}
{\bf q}^1(\psi)\\
&&\hbox{}+
B \lambda^{-1}\int_{\mathcal{R}(\tau',\tau'')\cap \{r\ge r_{\hat{Y}}\}}
{\bf q}^1(\psi)\\
&&\hbox{}+ B\lambda
\int_{\mathcal{R}(\tau',\tau'') } (Y(\hat{Y}\psi))^2\\
&&\hbox{}+ B\int_{\Sigma(\tau')} {\bf q}(\hat Y\psi).
\end{eqnarray*}
Since $(Y(\hat{Y}\psi))^2\le B\, {\bf q}(\hat Y\psi)$, it follows that $\lambda$ can be chosen so 
that
\begin{eqnarray*}
\int_{\mathcal{R}(\tau',\tau'')\cap \{r\le r^-_Y\}}
{\bf q}(\hat{Y}\psi) 
&\le&B\int_{\mathcal{R}(\tau',\tau'')\cap \{r\ge r^-_Y\}}
{\bf q}(\hat Y\psi)\\
&&\hbox{}+
B\int_{\mathcal{R}(\tau',\tau'') } {\bf q}(T\psi)+{\bf q}(\psi)\\
&&\hbox{}
+B \epsilon_{\rm close}^2 \int_{\mathcal{R}(\tau',\tau'')} {\bf q}^1(\psi)\\
&&\hbox{}+B \delta_{\hat{Y}}
\int_{\mathcal{R}(\tau',\tau'')\cap \{r\le r_{\hat{Y}}\}}
{\bf q}^1(\psi)\\
&&\hbox{}+
B \int_{\mathcal{R}(\tau',\tau'')\cap \{r\ge r_{\hat{Y}}\}}
{\bf q}^1(\psi)\\
&&+\hbox{}B\int_{\Sigma(\tau')} {\bf q}(\hat Y\psi)\\
&\le&
B\int_{\mathcal{R}(\tau',\tau'')\cap \{r\ge r^-_Y\}}
{\bf q}(\hat Y\psi)\\
&&\hbox{}+
B\int_{\mathcal{R}(\tau',\tau'') } {\bf q}(T\psi)+{\bf q}(\psi)\\
&&\hbox{}
+B \epsilon_{\rm close}^2 \int_{\mathcal{R}(\tau',\tau'')\cap \{r\le r^-_Y\}} {\bf q}^1(\psi)\\
&&\hbox{}
+B \delta_{\hat{Y}}
\int_{\mathcal{R}(\tau',\tau'')\cap \{r\le r_{\hat{Y}}\}}
{\bf q}^1(\psi)\\
&&\hbox{}
+B \epsilon_{\rm close}^2 \int_{\mathcal{R}(\tau',\tau'')\cap \{r\ge r^-_Y\}} {\bf q}^1(\psi)\\
&&\hbox{}+
B \int_{\mathcal{R}(\tau',\tau'')\cap \{r\ge r_{\hat{Y}}\}}
{\bf q}^1(\psi)\\
&&+\hbox{}B\int_{\Sigma(\tau')} {\bf q}(\hat Y\psi).
\end{eqnarray*}
From the Propositions of Section~\ref{elliptprops} below, we obtain
\begin{eqnarray*}
\int_{\mathcal{R}(\tau',\tau'')\cap \{r\le r^-_Y\}}
{\bf q}(\hat{Y}\psi) 
&\le&
B\int_{\mathcal{R}(\tau',\tau'')}
{\bf q}(T\psi)+{\bf q}(\psi)\\
&&\hbox{}
+B (\epsilon_{\rm close}^2+\delta_{\hat{Y}}) \int_{\mathcal{R}(\tau',\tau'')\cap \{r\le r^-_Y\}} {\bf q}(\hat{Y}\psi)
+{\bf q}(T\psi)+{\bf q}(\psi)\\
&&\hbox{}
+B(r_{\hat{Y}}) \int_{\mathcal{R}(\tau',\tau'')} {\bf q}(T\psi)
+{\bf q}(\psi)\\
&&+\hbox{}B\int_{\Sigma(\tau')} {\bf q}(\hat Y\psi),
\end{eqnarray*}
and thus, for small enough $\epsilon_{\rm close}$, and choosing
$r_{\hat{Y}}$ close enough to $2M$ (and
thus small enough $\delta_{\hat{Y}}$), we
obtain
\begin{eqnarray*}
\int_{\mathcal{R}(\tau',\tau'')\cap \{r\le r^-_Y\}}
{\bf q}(\hat{Y}\psi) 
&\le&
B\int_{\mathcal{R}(\tau',\tau'') } {\bf q}(T\psi)+{\bf q}(\psi)\\
&&+\hbox{}B\int_{\Sigma(\tau')} {\bf q}(\hat Y\psi).
\end{eqnarray*}
(The choice of $r_{\hat{Y}}$ having been made, we have written
above $B(r_{\hat{Y}})$ as $B$
following our convention.)
From $(\ref{notethat1})$ and $(\ref{notethat2})$, it now follows
that
\begin{eqnarray*}
\int_{\mathcal{R}(\tau',\tau'')\cap \{r\le r^-_Y\}}
{\bf q}(\hat{Y}\psi) 
&\le& 
B(\tau''-\tau') \int_{\Sigma(0)} {\bf q}(T\psi)+{\bf q}(\psi)\\
&&+\hbox{}B\int_{\Sigma(\tau')} {\bf q}(\hat Y\psi).
\end{eqnarray*}
It follows immediately that 
there exists a sequence $0=\tau_0<\tau_i<\tau_{i+1}$ 
such that 
\begin{equation}
\label{seqsuchthat}
|\tau_i-\tau_j|\le B,\qquad \tau_i\to \infty
\end{equation}
with
\begin{eqnarray*}
\int_{\Sigma(\tau_i)\cap \{r\le r^-_Y\}}
{\bf q}(\hat{Y}\psi) 
&\le& 
B \int_{\Sigma(0)} {\bf q}(T\psi)+{\bf q}(\psi)\\
&&+\hbox{}B\int_{\Sigma(0)} {\bf q}(\hat Y\psi).
\end{eqnarray*}
From $(\ref{notethat0})$, we have on the other hand
\begin{eqnarray*}
\int_{\Sigma(\tau_i)}
{\bf q}(T\psi) 
\le 
B \int_{\Sigma(0)} {\bf q}(T\psi)+{\bf q}(\psi),
\end{eqnarray*}
and from $(\ref{newstatement})$
\begin{eqnarray*}
\int_{\Sigma(\tau_i)}
{\bf q}(\psi) 
\le 
B \int_{\Sigma(0)} {\bf q}(\psi).
\end{eqnarray*}
From Proposition~\ref{elliptprop1} it follows that
\begin{eqnarray*}
\int_{\Sigma(\tau_i)\cap \{r\le r^-_Y\}}
{\bf q}^1(\psi) 
\le 
B \int_{\Sigma(0)} {\bf q}(\hat Y\psi)+{\bf q}(T\psi)+{\bf q}(\psi),
\end{eqnarray*}
while from Proposition~\ref{elliptprop2}, it follows
that
\begin{eqnarray*}
\int_{\Sigma(\tau_i)\cap \{r\ge r^-_Y\}}
{\bf q}^1(\psi) 
\le 
B \int_{\Sigma(0)} {\bf q}(T\psi)+{\bf q}(\psi).
\end{eqnarray*}
Thus in fact,
\begin{eqnarray*}
\int_{\Sigma(\tau_i)}
{\bf q}^1(\psi) 
\le 
B \int_{\Sigma(0)}{\bf q}(\hat Y\psi)+ {\bf q}(T\psi)+{\bf q}(\psi).
\end{eqnarray*}
In view of $(\ref{seqsuchthat})$, we obtain now easily
\begin{eqnarray*}
\int_{\Sigma(\tau)}
{\bf q}^1(\psi) 
&\le& 
B \int_{\Sigma(0)}{\bf q}(\hat Y\psi)+ {\bf q}(T\psi)+{\bf q}(\psi)\\
&\le&
B\int_{\Sigma(0)}{\bf q}^1(\psi).
\end{eqnarray*}

\end{proof}
\subsection{Elliptic estimates}
\label{elliptprops}
We have the following elliptic estimate on spheres:
\begin{proposition}
\label{elliptprop1}
Let $S_r$ denote a set of constant $r$ in a $t^*$, $r$, $x^A$, $x^B$ coordinate system.
For $\psi$ a solution of $\Box_g\psi=0$, we have
\[
\int_{S_r} {\bf q}^1(\psi)
\le B 
\int _{S_r} {\bf q}(T\psi)+{\bf q}(\hat{Y}\psi)+{\bf q}(\psi).
\]
\end{proposition}
\begin{proof}
Note first that 
\begin{eqnarray}
\label{SHMEIWSE}
\nonumber
{\bf q}^1(\psi) &\le& B\left( |\nabb^2 \psi|^2+
|\nabb (T\psi)|^2+ |\nabb (\hat{Y}\psi)|^2 +
|TT\psi|^2\right.\\
\nonumber
&&\hbox{}\left.+|\hat{Y}\hat{Y}\psi|^2+|T\hat{Y}\psi|^2
+ {\bf q}(\psi)\right)\\
&\le&B\left( |\nabb^2 \psi|^2+
{\bf q}(T\psi)+{\bf q}(\hat{Y}\psi)+{\bf q}(\psi)\right).
\end{eqnarray}
Let $\triangle_{\mathbb S^2}$ denote the standard Laplacian on the unit sphere.
In the coordinates of the first paragraph of Section~\ref{Ysec}, we may write
\[
\frac1{r^2}\triangle_{\mathbb S^2}\psi
=\partial_v(\hat{Y}\psi)-\frac2r(T\psi -\hat{Y}\psi) -P_2\psi.
\]
Integrating over $S_r$ endowed with metric of the standard unit sphere, 
we obtain the elliptic estimate
\[
\frac{1}{r^2}\int_{S_r} |\nabla^2_{\mathbb S^2}\psi|^2 \, dA_{\mathbb S^2}
\le 
B\int _{S_r} {\bf q}(T\psi)+{\bf q}(\hat{Y}\psi)+{\bf q}(\psi)+(P_2\psi)^2 \, dA_{\mathbb S^2},
\]
i.e., in view of the assumptions $(\ref{locsmal})$ on the metric,
\begin{equation}
\label{i.e.}
\int_{S_r} |\nabb^2\psi| ^2
\le 
B\int _{S_r} {\bf q}(T\psi)+{\bf q}(\hat{Y}\psi)+{\bf q}(\psi)+(P_2\psi)^2.
\end{equation}
On the other hand, from $(\ref{locsmal})$, $(\ref{locsmal2})$ we have
\begin{eqnarray*}
(P_2\psi)^2&\le &B\epsilon_{\rm close}{\bf q}^1(\psi).
\end{eqnarray*}
The above, $(\ref{i.e.})$ and $(\ref{SHMEIWSE})$ yield the proposition,
for $\epsilon_{\rm close}$ sufficiently small.
\end{proof}

In addition, we have the following elliptic estimates away from the event horizon.
\begin{proposition}
\label{elliptprop2}
For $\psi$ a solution of $\Box_g\psi=0$, and $r_0$ a parameter with 
$r_0>2M$, then, for $\epsilon_{\rm close}$ sufficiently small,
we have
\[
\int_{\Sigma(\tau')\cap\{r\ge r_0\} }{\bf q}^1(\psi)
\le B(r_0) \int_{\Sigma(\tau') }{\bf q}(T\psi)+{\bf q}(\psi),
\]
\[
\int_{\mathcal{R}(\tau',\tau'')\cap\{r\ge r_0\} }{\bf q}^1(\psi)
\le B(r_0) \int_{\mathcal{R}(\tau',\tau'') }{\bf q}(T\psi)+{\bf q}(\psi).
\]
\end{proposition}
\begin{proof}
The proof of this straightforward elliptic estimate is left to the reader.
\end{proof}

\subsection{Pointwise bounds}
\label{pwise}
We have the following Sobolev-type estimate on Schwarzschild
\begin{proposition}
Let $\Psi$ be a smooth function on $\Sigma_{\tau}$ 
of compact support.
Then
\[
\sup_{\Sigma_{\tau}}\Psi^2 \le B \int_{\Sigma_{\tau}}  
|\nabla^2_{\Sigma_{\tau}}\Psi|^2_{(g_M)_{\Sigma_\tau}} +|\nabla_{\Sigma_{\tau}}\Psi|^2_{
(g_M)_{\Sigma_\tau}}.
\]
\end{proposition}
This in turn follows from the following Euclidean space estimate:
\begin{proposition}
There exists a constant $K$ such that the following holds.
Let $\Psi$ be a smooth function on $\mathbb R^3\cap \{r \ge 1\}$ 
of compact support.
Then
\[
\sup_{r\ge 1} \Psi^2 \le K \int_{\{r\ge 1\}}( |\nabla^2\Psi|^2+ 
|\nabla\Psi|^2) \, dx^1 dx^2 dx^3.
\]
\end{proposition}
\begin{proof}
Omitted.
\end{proof}

We obtain
\begin{theorem}
Let $n\ge 0$. There exists a positive constant $\epsilon_{\rm close}$,
depending only on $M$, and a positive constant $C_n$, depending on $M$ and $n$,
such that the following holds.Let $g$, $\Sigma(\tau)$ be as in Section~\ref{genclass}
and let $\uppsi$, $\uppsi'$, $\psi$ be as in Section~\ref{sols} where
$\psi$ satisfies $(\ref{waveee})$.
Then, for $\tau\ge0$,
\[
|\nabla^{(n)}\psi|^2_{g_{\Sigma_\tau}} \le  \lim_{r\to\infty}\uppsi^2
+C_{n} \int_{\Sigma(0)}{\bf q}^{n+1}(\psi)
\]
where $\nabla^{(n)}\psi$ denotes the $n$'th order spacetime covariant derivative
tensor and $|\cdot|_{g_{\Sigma_\tau}}$ denotes the induced Riemannian norm. 
\end{theorem}
Theorem~\ref{deutero} in particular follows from the above.

\section{Further notes}
\subsection{The Schwarzschild case}
\label{further1}

In the Schwarzschild case, we may apply the estimates proven here for $\psi_{\sharp}$ 
in Sections~\ref{estest} and~\ref{boubou} directly 
to the whole $\psi$. Since no frequency decomposition need be made, 
no associated error terms arise and the whole argument
can be reduced to a few pages.
The resulting energy estimate, coupled with the higher order and pointwise
estimates of Section~\ref{hes},
yield a new proof
for uniform boundedness of solutions
to the wave equation on Schwarzschild which is in some sense the simplest one yet--using 
neither the discrete isometry exploited by Kay-Wald~\cite{kw:lss}, nor
the vector field $X$ of our~\cite{dr3} or~\cite{dr4}, nor commuting
with angular momentum operators. Moreover, one shows the uniform boundedness
of all derivatives on the event horizon up to all order, whereas previous results
could control only tangential derivatives.

In fact, one can obtain a much more general
statement applying to all static spherically symmetric
non-extremal black holes.
We have
\begin{theorem}
\label{stat}
Let $(\mathcal{D},g)$ be a static spherically symmetric asymptotically
flat exterior black hole spacetime bounded by a non-extremal event horizon
$\mathcal{H}^+$.
Then the estimates of Theorems~\ref{prwto} and~\ref{deutero} hold.
\end{theorem}

\subsection{Kerr-de Sitter}
Our argument is easily adapted to spacetimes which are small
perturbations of non-extremal
Schwarzschild-de Sitter, in particular to slowly rotating non-extremal
Kerr-de Sitter, or Kerr-Newman-de Sitter. 
See~\cite{dr5} for the setting. One fixes the manifold
structure on a subregion $\mathcal{D}\cap J^+(\Sigma_0)$ where 
$\mathcal{D}$ is here the region between
a set of black/white hole and cosmological horizons and
$\Sigma_0$ is a Cauchy surface crossing both horizons
to the future of the bifurcate spheres.
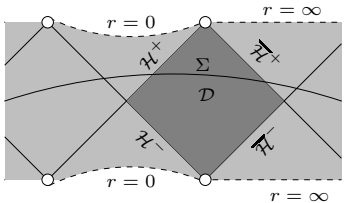
\[
\begin{picture}(0,0)%
\includegraphics{desitforkerr.pstex}%
\end{picture}%
\setlength{\unitlength}{2368sp}%
\begingroup\makeatletter\ifx\SetFigFont\undefined%
\gdef\SetFigFont#1#2#3#4#5{%
  \reset@font\fontsize{#1}{#2pt}%
  \fontfamily{#3}\fontseries{#4}\fontshape{#5}%
  \selectfont}%
\fi\endgroup%
\begin{picture}(3549,2164)(2914,-4694)
\put(5626,-2686){\makebox(0,0)[lb]{\smash{{\SetFigFont{7}{8.4}{\rmdefault}{\mddefault}{\updefault}{\color[rgb]{0,0,0}$r=\infty$}%
}}}}
\put(5701,-4636){\makebox(0,0)[lb]{\smash{{\SetFigFont{7}{8.4}{\rmdefault}{\mddefault}{\updefault}{\color[rgb]{0,0,0}$r=\infty$}%
}}}}
\put(3991,-2821){\makebox(0,0)[lb]{\smash{{\SetFigFont{7}{8.4}{\rmdefault}{\mddefault}{\updefault}{\color[rgb]{0,0,0}$r=0$}%
}}}}
\put(3991,-4486){\makebox(0,0)[lb]{\smash{{\SetFigFont{7}{8.4}{\rmdefault}{\mddefault}{\updefault}{\color[rgb]{0,0,0}$r=0$}%
}}}}
\put(4276,-3961){\rotatebox{315.0}{\makebox(0,0)[lb]{\smash{{\SetFigFont{7}{8.4}{\rmdefault}{\mddefault}{\updefault}{\color[rgb]{0,0,0}$\mathcal{H}^-$}%
}}}}}
\put(5476,-3061){\rotatebox{315.0}{\makebox(0,0)[lb]{\smash{{\SetFigFont{7}{8.4}{\rmdefault}{\mddefault}{\updefault}{\color[rgb]{0,0,0}$\overline{\mathcal{H}}^+$}%
}}}}}
\put(4951,-3586){\makebox(0,0)[lb]{\smash{{\SetFigFont{7}{8.4}{\rmdefault}{\mddefault}{\updefault}{\color[rgb]{0,0,0}$\mathcal{D}$}%
}}}}
\put(4413,-3281){\rotatebox{45.0}{\makebox(0,0)[lb]{\smash{{\SetFigFont{7}{8.4}{\rmdefault}{\mddefault}{\updefault}{\color[rgb]{0,0,0}$\mathcal{H}^+$}%
}}}}}
\put(4920,-3266){\makebox(0,0)[lb]{\smash{{\SetFigFont{7}{8.4}{\rmdefault}{\mddefault}{\updefault}{\color[rgb]{0,0,0}$\Sigma$}%
}}}}
\put(5626,-4186){\rotatebox{45.0}{\makebox(0,0)[lb]{\smash{{\SetFigFont{7}{8.4}{\rmdefault}{\mddefault}{\updefault}{\color[rgb]{0,0,0}$\overline{\mathcal{H}}^-$}%
}}}}}
\end{picture}%

\]
One continues as in the Schwarzschild case. The argument is in fact easier
at several points. Because $r$ is bounded
in $\mathcal{D}$, the zero-order terms pose no difficulty. In particular,
one need not introduce the $\Sigma^{+}$ and $\Sigma^{-}$ surfaces, nor
must one modify $J^{X_a}$ by the addition of $J^{X_b,w_b}$.
We leave the details for a subsequent paper.

\subsection{Non-quantitative decay}
\label{nqd}
As a final application, we note that uniform boundedness is sufficient to 
translate non-quantitative results for fixed angular frequency into 
non-quantitative results for $\psi$ itself. For instance
\begin{corollary}
\label{nqc}
Suppose for each $k$ we have $\psi_k(\cdot, t)\to 0$ where $\psi_k$ denotes
the projection to the $k$'th azimuthal mode. 
Then $\psi(\cdot, t)\to 0$.
\end{corollary}
The assumption of the above corollary is obtained
in~\cite{fksy} away from the event horizon for Kerr solutions for the very special case where
the initial data is supported away from the horizon.

\end{document}